\documentclass[12pt]{article}
\usepackage[utf8]{inputenc}
\usepackage{booktabs} % creating tables
\usepackage{mathtools} % creating equation
\usepackage{amsthm} % proofs and theorems
\usepackage{amssymb} % common math symbols
\usepackage{amsmath} % more math related environments
\usepackage[scr=boondoxo,scrscaled=1]{mathalfa}
\usepackage{framed} % formatting algorithms or pseudo-code
\usepackage{fontenc} % other fonts
\usepackage{natbib}[square]
\usepackage{array} % new table headers
\usepackage{float} % because figures are a pain in the ass
\usepackage{newfloat} % to create own float objects
\usepackage{comment} % leave large blocks of comments
\usepackage{relsize} % increase size of math
\usepackage{subcaption} % place fugures side by side
\usepackage{caption}
\usepackage[titletoc,title]{appendix} % appendices
\usepackage{algorithm2e}
\usepackage{adjustbox}

\RequirePackage[colorlinks,citecolor=blue, urlcolor=blue, linkcolor=blue]{hyperref} % formats hyperlinks
\setcitestyle{square}

\newtheorem{theorem}{Theorem}[section]
\newtheorem{assumption}{Assumption}[section]
\newtheorem{corollary}{Corollary}[section]
\newtheorem*{setup}{Setup}

\DeclareFloatingEnvironment[fileext=loa, name=Algorithm]{palgorithm}
\newcommand{\blind}{0}

\theoremstyle{break}

\begin{document}

\def\spacingset#1{\renewcommand{\baselinestretch}%
{#1}\small\normalsize} \spacingset{1}

%%%%%%%%%%%%%%%%%%%%%%%%%%%%%%%%%%%%%%%%%%%%%%%%%%%%%%%%%%%%%%%%%%%%%%%%%%%%%%

\if0\blind
{
  \title{\bf ECoHeN: A Hypothesis Testing Framework for Extracting Communities from Heterogeneous Networks}
  \author{Connor P. Gibbs\hspace{.2cm}\\
    Department of Statistics, Colorado State University\\
    and \\
    Bailey K. Fosdick \\
    Department of Biostatistics and Informatics, Colorado School of Public Health\\
    and \\
    James D. Wilson \\
    Department of Mathematics and Statistics, University of San Francisco}
  \maketitle
} \fi

\if1\blind
{
  \bigskip
  \bigskip
  \bigskip
  \begin{center}
    {\LARGE\bf ECoHeN: A Hypothesis Testing Framework for Extracting Communities from Heterogeneous Networks}
\end{center}
  \medskip
} \fi

% 200 words or fewer
\bigskip
\begin{abstract}
Community discovery is the general process of attaining assortative communities from a network: collections of nodes that are densely connected within yet sparsely connected to the rest of the network. While community discovery has been well studied, few such techniques exist for heterogeneous networks, which contain different types of nodes and possibly different connectivity patterns between the node types.  In this paper, we introduce a framework called ECoHeN, which \textbf{e}xtracts \textbf{co}mmunities from a \textbf{he}terogeneous \textbf{n}etwork in a statistically meaningful way. Using a heterogeneous configuration model as a reference distribution, ECoHeN identifies communities that are significantly more densely connected than expected given the node types and connectivity of its membership. Specifically, the ECoHeN algorithm extracts communities one at a time through a dynamic set of iterative updating rules, is guaranteed to converge, and imposes no constraints on the type composition of extracted communities. To our knowledge this is the first discovery method that distinguishes and identifies both homogeneous and heterogeneous, possibly overlapping, community structure in a network. We demonstrate the performance of ECoHeN through simulation and in application to a political blogs network to identify collections of blogs which reference one another more than expected considering the ideology of its' members.
\end{abstract}

\noindent%
\vfill

\newpage
\spacingset{1.5} % DON'T change the spacing!
\section{Introduction}
\label{sec:introduction}

Complex phenomena, from biological systems \citep{nacu2007gene} to world trade patterns \citep{garcia2019stochastic}, are often modeled as networks (graphs) which consist of entities (nodes), connections between them (edges), and known covariates about the entities and the connections (node and edge attributes). Many disciplines have devoted significant efforts to the analysis and application of network models including statistics, physics, computer science, biology, and the social sciences. One focus that has garnished much attention is \textit{community discovery}: the general process of assigning nodes to collections whose members are densely connected within the collection yet sparsely connected to the rest of the network, calling these relatively dense subsets of nodes \textit{communities} \citep{girvan2002community}. After community members were shown to embody structural and functional similarities \citep{newman2004finding}, a flurry of algorithms designed to partition a graph to create this disparity in connectedness within and between partitions were proposed, e.g., \citep{wu2004finding, radicchi2004defining}. The process of partitioning a network is generally referred to as \textit{community detection}. 

The rapid growth of literature in community detection prompted several comparative studies in the late 2000s \citep{lancichinetti2009community, fortunato2010community} and generalizations to community detection in the 2010s \citep{traag2009community, psorakis2011overlapping, yang2013community}. One generalization of interest is a shift from community detection, a graph partitioning problem, to \textit{community extraction}, a set-identification problem. Community extraction methods seek to discover communities from a network one at a time, allowing for arbitrary structure in the rest of the network. Unlike community detection methods, community extraction methods readily identify \textit{background nodes}, defined as nodes that are not preferentially attached to any well-defined community, and overlapping community structures, where nodes may belong to more than one community. While some extraction methods seek to optimize an extraction criterion \citep{zhao2011community, wilson2017community}, others seek to define the statistical significance of a community's connections under a global null model \citep{lancichinetti2011finding, wilson2014testing}, a strategy we follow in this manuscript.

There has been notable progress in the development of methodology for extracting communities from homogeneous networks; however, few methods exist for a more general heterogeneous network. Formally defined in Section \ref{sec:heterogeneous_nets}, a \textit{heterogeneous network} is a network with different types of nodes. Most networks representing real systems are in fact heterogeneous \citep{shi2016survey}. For example, large-scale biological systems are often represented as heterogeneous networks \citep{pinero2016disgenet,alshahrani2017neuro}, where the biological entities, i.e. nodes, are distinguished by their biological function, i.e. node types. These networks comprise node types like proteins, diseases, phenotypes, and genetic variants \citep{callahan2020framework}. Often extreme interest is placed on understanding fundamental relationships both $\textit{among}$ and $\textit{between}$ these biological entities. While community discovery is a common tool in this setting \citep{choobdar2018open}, few methods are designed to identify communities which are densely connected considering the node types of the community members. Thus, practitioners are forced to either ignore node type altogether, treating the network as homogeneous, or adapt standard community discovery methods to analyze subgraphs of the heterogeneous network separately. In either case, any information gleaned from differences in the rates of connectivity between nodes of different type are subsequently ignored. 

Knowing, for example, that a link between a gene and a phenotype is relatively rare compared to a link between two genes or a link between two phenotypes is valuable information when determining whether a set of nodes composed of both genes and phenotypes should be deemed an assortative community. Relatively few connections between the genes and the phenotypes of this set might be deemed substantial when considering the general propensity for nodes of these types to share a connection in the network.  However, this is an unattainable conclusion if communities in a heterogeneous network are not determined according to the topology of the network as it relates to the node types of the community membership. 

A disparity in connectedness related to node type is not unique to biological networks. \citet{mcpherson2001birds} describe the tendency for kindred individuals to connect as \textit{homophily}, where homophily can occur on categories such as gender, class, or political ideology. Conditioning, for example, on political ideology and accounting for the relative propensity for political actors to connect can provide rich information about the underlying community structure, including communities of mixed political ideology which are undetectable using contemporary methods. Methods designed to account for the heterogeneity in the node types and different connectivity patterns between pairs of node types tend to assume communities should be densely connected with a) similar node types among the community members, e.g., \citep{liu2014framework, sengupta2015spectral, smith2016partitioning, li2018community}, or b) dissimilar node types among the members, e.g., \citep{zhang2020modularity}. As a result, existing methods facilitate the discovery of homogeneous or heterogeneous community structure, but not both. 

In this paper, we introduce ECoHeN: an algorithm designed to \textbf{e}xtract \textbf{co}mmunities from \textbf{he}terogeneous \textbf{n}etworks. The significance of connectivity between a node and a set of nodes is measured using a $p$-value derived from the reference distribution under a heterogeneous degree configuration model. Using these $p$-values, ECoHeN extracts communities one at a time through a dynamic set of iterative updating rules which are guaranteed to converge. ECoHeN is a generalization and refinement of an extraction method ESSC which stands for \textbf{e}xtraction of \textbf{s}tatistically \textbf{s}ignificant \textbf{c}ommunities \citep{wilson2014testing}. Unlike ESSC, ECoHeN accounts for existing differences in the connectivity patterns between pairs of node types to identify communities that are more densely connected than expected given the node types and connectivity of its membership. Existing community discovery methods for heterogeneous networks have treated the discovery of homogeneous and heterogeneous community structure as two separate objectives requiring separate algorithms. In comparison, ECoHeN makes no assumption and places no constraint on the resulting type composition of each community, allowing ECoHeN to distinguish and identify both homogeneous and heterogeneous community structures that may overlap and can identify nodes that are not preferentially attached to any community.

We start by formally defining a heterogeneous graph and outlining notation in Section \ref{sec:heterogeneous_nets}. In Section \ref{sec:hdcm}, we introduce the heterogeneous degree configuration model (HDCM), a null model for studying heterogeneous networks. In Section \ref{sec:sig_of_connection}, we provide the theoretical foundation for ECoHeN before outlining the algorithm in Section \ref{sec:ecohen} and discussing the algorithm's parameter choices in Section \ref{sec:effects_of_choices}. We illustrate the performance of ECoHeN relative to other methods in Section \ref{sec:simulation} before applying ECoHeN to a well-known political blogs data set in Section \ref{sec:results}. Finally, we conclude with a discussion in Section \ref{sec:discussion}.

\section{Heterogeneous Networks}
\label{sec:heterogeneous_nets}

Heterogeneous networks are a special case of a more general class of networks known as $\textit{node-attributed networks}$. Consistent with the terminology of \citet{wasserman1994social}, node-attributed networks consist of a $\textit{structural dimension}$ with nodes and interactions among them, a $\textit{compositional dimension}$ containing the attributes (also called features or metadata) for the nodes, and an $\textit{affiliation dimension}$ describing the group memberships. Since known affiliations can be expressed as node attributes, we interpret affiliation as being some latent membership to be learned through community discovery. We use the term heterogeneous network to mean a node-attributed network whose nodes are distinguished by a categorical feature called \textit{node types}. Informally, a heterogeneous network is a node-colored network where each unique color represents a unique node type. 

Heterogeneous networks are common in the social and biological sciences. For example, the political blogs network introduced by \citet{adamic2005political} is a heterogeneous network in which political blogs, represented as nodes, are distinguished by their political ideology, represented as node types (liberals in blue and conservatives in red, as in Figure \ref{fig:heterogeneous_net_examples}\textcolor{blue}{a}). Used to represent large-scale biological processes, knowledge graphs are another example of heterogeneous networks where biological actors, represented as nodes, are differentiated by their biological function in the network, represented as node types. Common node types in a knowledge graph include genetic variants, proteins, and phenotypes \citep{callahan2020framework}. A smaller, contrived example of a heterogeneous network is provided in Figure \ref{fig:heterogeneous_net_examples}\textcolor{blue}{b} to help solidify notation and concepts to come. This network has two node types, depicted in blue and orange. If multiple categorical features are present, we encode each possible combination of features as a different type.

\begin{figure}[hbt!]
\centering
   \begin{subfigure}{0.48\linewidth} \centering
     \includegraphics[scale=0.55]{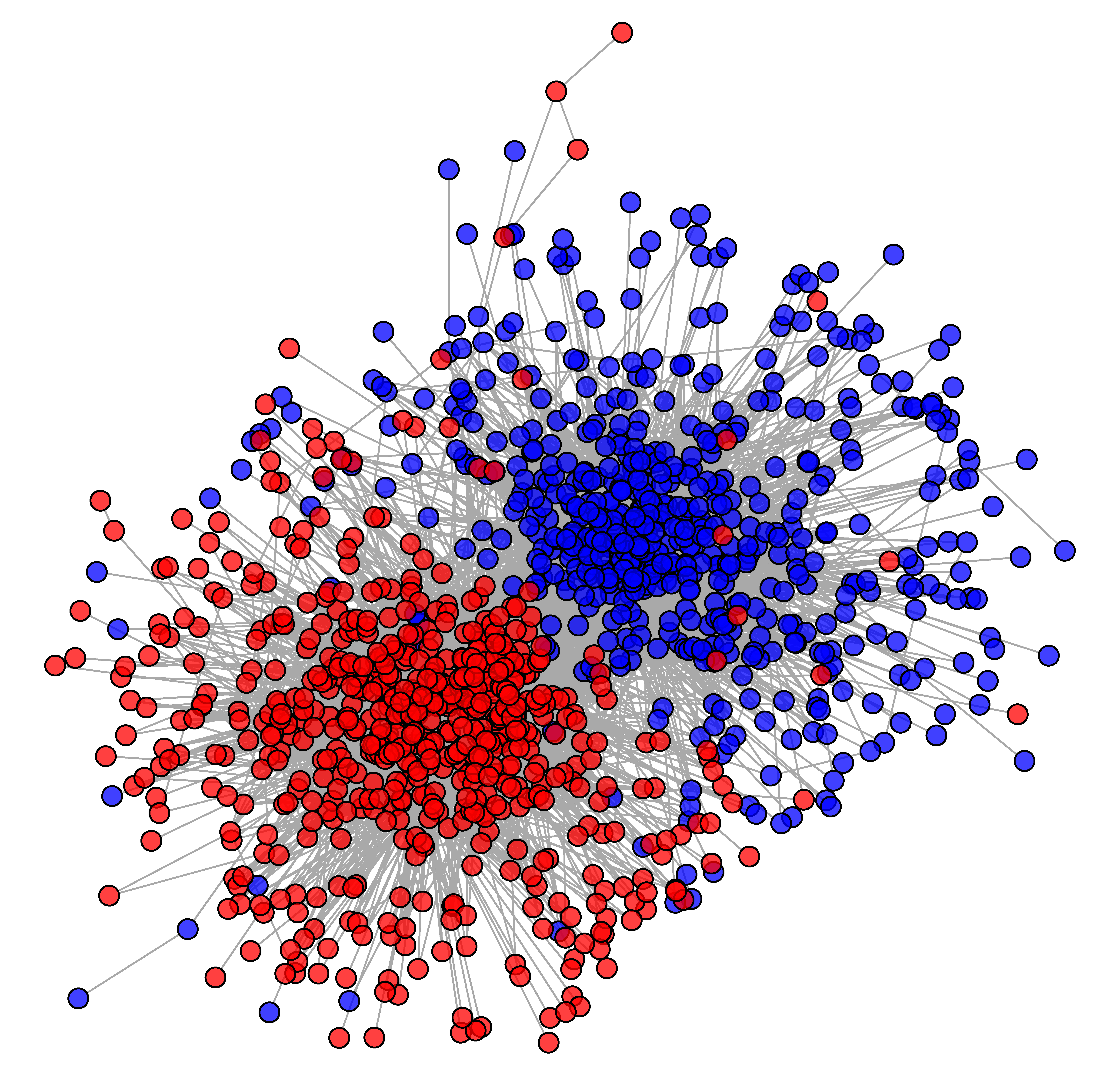}
     \caption{Political blogs network}
   \end{subfigure}
   \begin{subfigure}{0.48\linewidth} \centering
     \includegraphics[scale=0.40]{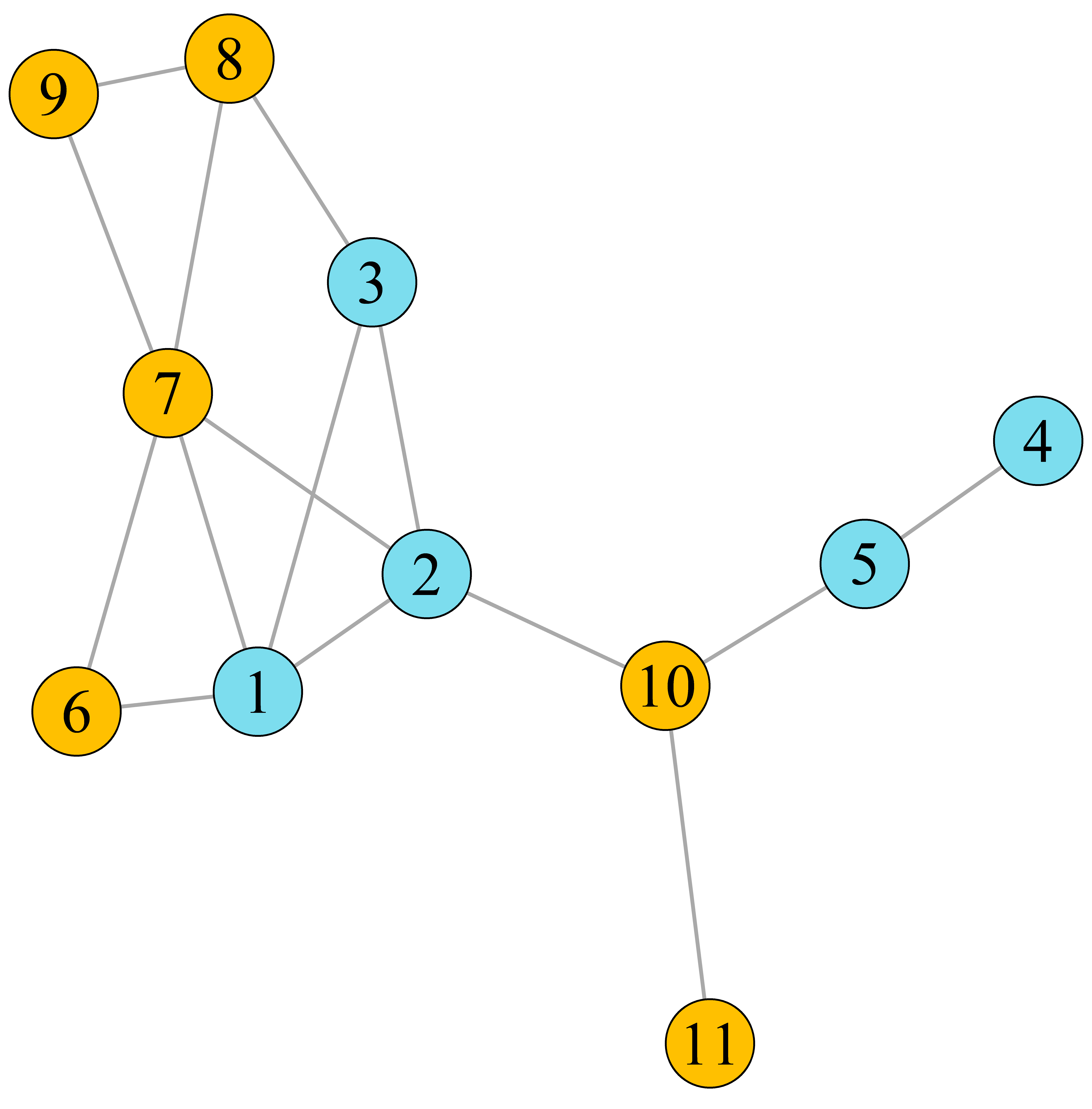}
     \caption{Toy heterogeneous network}
   \end{subfigure}
\caption{Panel (a) is a force-directed layout of the political blogs network: a network of political blogs and the hyperlinks between them. Each blog, represented as a node, is colored according to the political ideology, i.e. its node type, where red is used to indicate conservative ideology and blue is used to indicated liberal ideology. Panel (b) is a toy example of a heterogeneous network $\mathcal{G}$ with 11 nodes and 15 edges, implying $|\mathcal{V}|=11$ and $|\mathcal{E}|=15$. There are five type one nodes (colored in blue) and six type two nodes (colored in orange), implying $|V^{[1]}|=5$ and $|V^{[2]}|=6$. The edge set consists of nine within-type edges (four among type one nodes and five among type two) and six between-type edges, implying $|E^{[11]}|=4$, $|E^{[22]}|=5$, and $|E^{[12]}|=6$.}
\label{fig:heterogeneous_net_examples}
\end{figure}

Formally, let $\mathcal{G}=(\mathcal{V}, \mathcal{E})$ denote an observed, undirected heterogeneous network with $n$ nodes labeled $1, \dots, n$ and $K$ node types labeled $1,\dots, K$. The node set $\mathcal{V}=\bigcup_{k = 1}^{K}{V^{[k]}}$ where $V^{[k]}$ denotes the node set containing $|V^{[k]}|$ nodes of type $k$ and $V^{[k]}\cap V^{[l]}=\emptyset$. When necessary, the node type of an arbitrary node $u$ is denoted in bracketed superscript such that $u^{[l]}$ implies that node $u$ is of type $l\in\{1, \dots, K\}$. Let $\mathcal{T}$ denote an $n$-dimensional vector of node types where the $u^{\text{th}}$ element of $\mathcal{T}$ provides the node type of node $u$. We allow the heterogeneous network $\mathcal{G}$ to be a \textit{multigraph}, meaning there could be self-loops or multi-edges; if there are no self-loops or multi-edges, the heterogeneous network is deemed \textit{simple}. Figure \ref{fig:heterogeneous_net_examples}\textcolor{blue}{b}, for example, is a simple, heterogeneous network with two node types. The edge multiset of $\mathcal{G}$ can be partitioned according to adjacent nodes' types: $\mathcal{E} = \cup_{1\leq k \leq l \leq K} E^{[kl]}$ where $E^{[kl]}$ contains the links between nodes of type $k$ and nodes of type $l$. If $v^{[k]}$ is adjacent to $u^{[l]}$, then $\{v^{[k]}, u^{[l]}\} \in E^{[kl]}$. Since $\mathcal{G}$ is undirected $E^{[kl]}=E^{[lk]}$.

Partitioning the node and edge sets motivates the fact that $\mathcal{G}$ is the collection of $K$ unipartite and $K(K-1)/2$ bipartite subgraphs. Figure \ref{fig:heterogeneous_net_augmentation} highlights the two unipartite subgraphs and one bipartite subgraph that when augmented together yield the heterogeneous network depicted in Figure \ref{fig:heterogeneous_net_examples}\textcolor{blue}{b}. Let $G^{[kl]} = (V^{[k]} \cup V^{[l]}, E^{[kl]})$ denote the subgraph composed of type $k$ and type $l$ nodes and the connections between them. If $k=l$, then $G^{[kk]}$ is a unipartite subgraph; otherwise, $G^{[kl]}$ is a bipartite subgraph. Since edges are assumed bidirectional, $G^{[kl]}$ is equal to $G^{[lk]}$.

\begin{figure}[t]
\centering
   \begin{subfigure}{0.32\linewidth} \centering
     \includegraphics[scale=0.33]{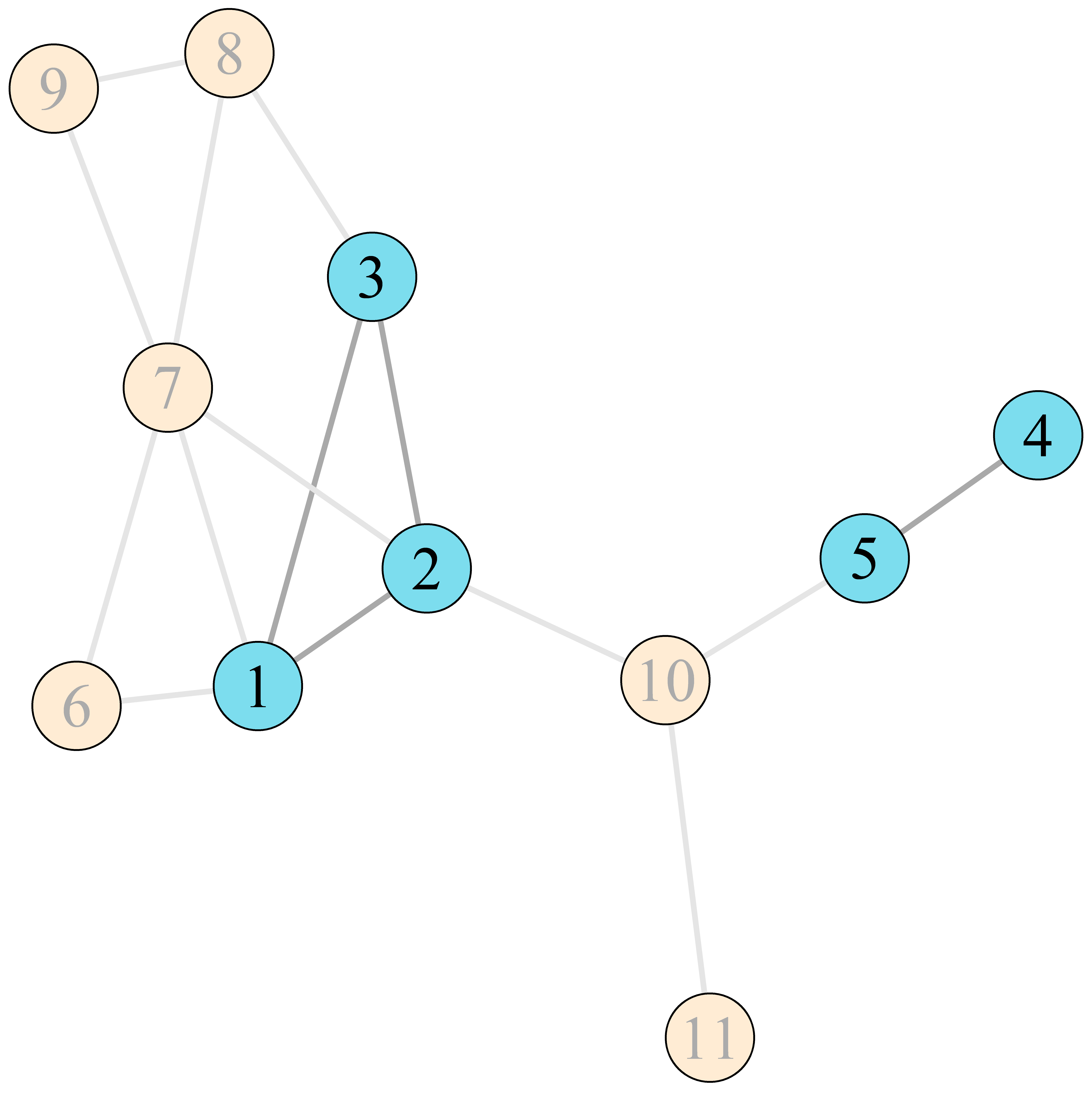}
     \caption{$G^{[11]}=(V^{[1]}, E^{[11]})$}
   \end{subfigure}
   \unskip\ \vrule\
   \begin{subfigure}{0.32\linewidth} \centering
     \includegraphics[scale=0.33]{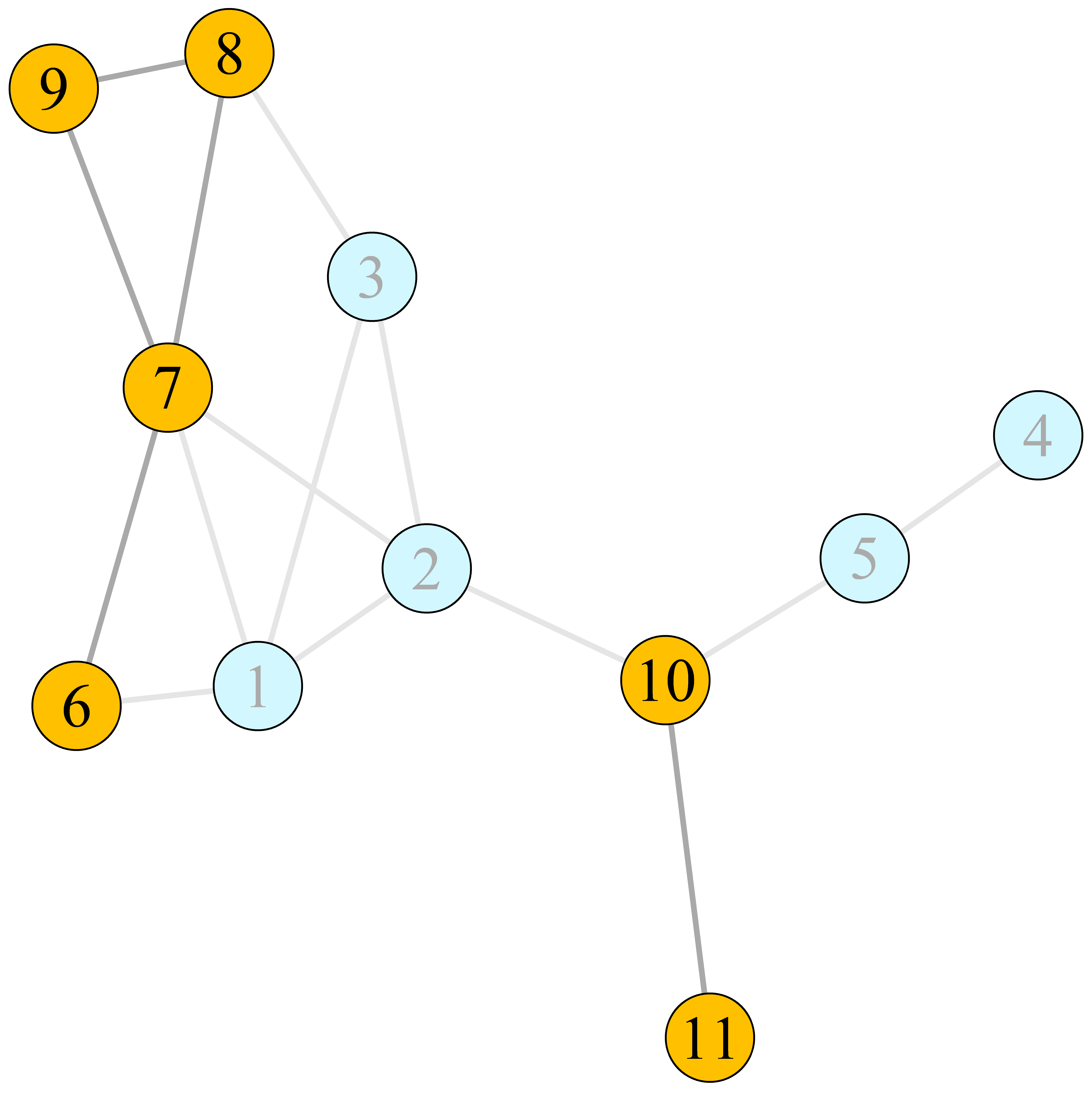}
     \caption{$G^{[22]}=(V^{[2]}, E^{[22]})$}
   \end{subfigure}
   \unskip\ \vrule\
   \begin{subfigure}{0.32\linewidth} \centering
     \includegraphics[scale=0.33]{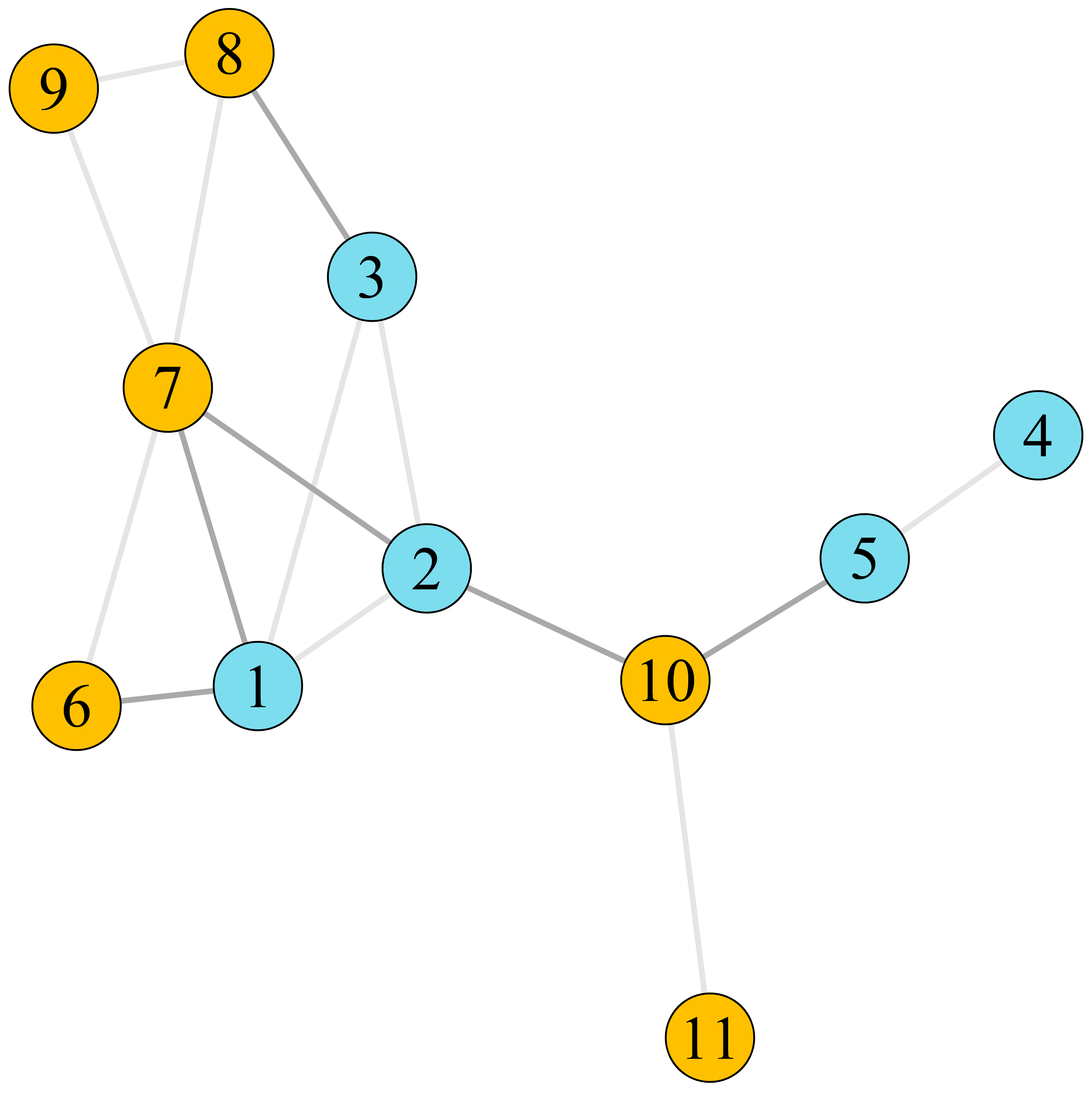}
     \caption{$G^{[12]}=(V^{[1]}\cup V^{[2]}, E^{[12]})$}
   \end{subfigure}
\caption{Decomposition of the toy network example provided in Figure \ref{fig:heterogeneous_net_examples}\textcolor{blue}{b} into two unipartite subgraphs ($G^{[11]}, G^{[22]}$) and a bipartite subgraph ($G^{[12]}$). Each subgraph $G^{[kl]}$ contains nodes of type $k$ and $l$ and the observed links between them highlighted. The degree of $3^{[1]}$, $d(3^{[1]})$, is three. The type 1 degree of $3^{[1]}$, $d^{[1]}(3^{[1]})$, is two, and the heterogeneous degree sequence of node $3^{[1]}$ is $\boldsymbol{d}(3^{[1]})=(2, 1)$.} 
\label{fig:heterogeneous_net_augmentation}
\end{figure}

Important for subsequent development, the \textit{degree} of $u^{[l]}$, denoted $d(u^{[l]})$, is the number of edges incident to $u^{[l]}$. At times, it will be useful to consider the number of edges incident to $u^{[l]}$ also incident to a type $k$ node; we refer to this quantity as the \textit{type $k$ degree} of $u^{[l]}$, denoted as $d^{[k]}(u^{[l]})$\footnote{Allowing for self-loops and multi-edges can complicate the interpretation and discussion surrounding $d^{[k]}(u^{[l]})$. In an undirected, simple setting, there is no distinction between counting the number of edges incident to $u^{[l]}$ also incident to a type $k$ node (i.e., $d^{[k]}(u^{[l]})$) and counting the number of type $k$ nodes adjacent to $u^{[l]}$, so we treat them as the same objective here. When discussing the number of type $k$ nodes adjacent to $u^{[l]}$, we would count, for example, a type $k$ node as twice adjacent to $u^{[l]}$ should there exist two edges between them.}. We let
\begin{equation}
\label{eq:hdegseq}
\boldsymbol{d}(u^{[l]})=\left(d^{[1]}(u^{[l]}), \dots, d^{[K]}(u^{[l]})\right)    
\end{equation}
denote the \textit{heterogeneous degree sequence} of $u^{[l]}$, providing the number of type $k$ nodes adjacent to $u^{[l]}$ for all types $k\in\{1, \dots, K\}$. Note that $d^{[k]}(u^{[l]})$ represents degree of $u^{[l]}$ in $G^{[kl]}$, such that $d(u^{[l]})=\sum_k d^{[k]}(u^{[l]})$ represents the degree of $u^{[l]}$ in $\mathcal{G}$.

\section{Methodology}
\label{sec:methods}

We now outline a method to extract communities from heterogeneous networks which we refer to as ECoHeN. ECoHeN is designed to account for differences in connectivity patterns between nodes of various types. At its core, ECoHeN evaluates the significance of connectivity between a node and a set of nodes using a $p$-value derived from the reference distribution arising from a heterogeneous degree configuration model (HDCM). Using these $p$-values, ECoHeN extracts communities one at a time through a dynamic set of iterative updating rules which are guaranteed to converge. Identified communities are more densely connected than expected given the node types and connectivity of its membership, allowing for the discovery of both homogeneous and heterogeneous community structure. Resulting communities may overlap and exclude nodes altogether which are not preferentially attached to any community.

We start by introducing the HDCM in Section \ref{sec:hdcm}. In Section \ref{sec:sig_of_connection}, we define a $p$-value for the hypothesis that an arbitrary node is well-connected to an arbitrary set of nodes and provide a theorem for the asymptotic distribution of a random variable under the HDCM which serves as the foundation for a reasonable approximation of the $p$-value. In Section \ref{sec:ecohen}, we introduce the ECoHeN algorithm and discuss initialization, extraction, and convergence. Furthermore, we discuss a procedure for paring down the set of discovered communities for practitioners. Finally, in Section \ref{sec:effects_of_choices}, we discuss the algorithm's parameter choices.

\subsection{Heterogeneous Degree Configuration Model}
\label{sec:hdcm}

The degree configuration model (DCM) is a classic random network null model defined as the uniform distribution over the space of networks maintaining a given degree sequence. In the heterogeneous network setting, simply fixing the degree of each node ignores the differences among the nodes and the resulting differences in the types of connections they form. Motivated by \citet{zhang2020modularity}, we define the \textit{heterogeneous degree configuration model} (HDCM) as the uniform distribution over the space of networks with node types $\mathcal{T}$ maintaining a given collection of heterogeneous degree sequences, $\boldsymbol{\mathcal{D}}$, where
\begin{equation}
\boldsymbol{\mathcal{D}} = \left\{\boldsymbol{d}(u):u\in\mathcal{V}\right\} = 
\left\{\left(d^{[1]}(u), \dots, d^{[K]}(u)\right):u\in\mathcal{V}\right\}
.
\end{equation}
The model, denoted $\text{HDCM}(\mathcal{T}, \boldsymbol{\mathcal{D}})$, assumes all edge configurations adhering to the given collection of heterogeneous degree sequences are equally likely.

One can construct a random, heterogeneous network with node types $\mathcal{T}$ and degree collection $\boldsymbol{\mathcal{D}}$ through a generative process informally known as ``stub matching." Figure \ref{fig:hdcm_flow_stub} provides an illustration of this process for the toy network example in Figure \ref{fig:heterogeneous_net_examples}\textcolor{blue}{b}. Initially, each node type is assigned a corresponding color, e.g., node type one (two) is assigned blue (orange). Starting with $n$ isolate nodes, each node is assigned a color according to its respective node type provided in $\mathcal{T}$, e.g., node $3^{[1]}$ is assigned blue. Colored stubs, which act as half-edges, are attached to each node according to each node's heterogeneous degree sequence $\boldsymbol{\mathcal{D}}$, e.g., node $3^{[1]}$ is assigned two blue stubs and one orange stub since $\boldsymbol{d}(3^{[1]}) = (d^{[1]}(3^{[1]}), d^{[2]}(3^{[1]})) = (2, 1)$. Finally, stubs are attached uniformly at random within constraints dictated by the color of the stubs and nodes. At each stage, a stub is selected uniformly at random from the set of available stubs. The color of the stub indicates the color of the node to which it must connect. If, for example, the first stub selected is an orange stub emanating from a blue node, then it is matched randomly with an available blue stub emanating from an orange node. This process is repeated until no stubs remain. See Figure \ref{fig:hdcm_flow_stub} for an illustration of this process. Supplement A provides further discussion on the similarities and differences between the DCM and the HDCM, and an efficient means for generating a heterogeneous network from the model through a constrained permutation of the node labels in the edge multiset $\mathcal{E}$.

\begin{figure}[htb!]
    \centering
    \includegraphics[width=\textwidth]{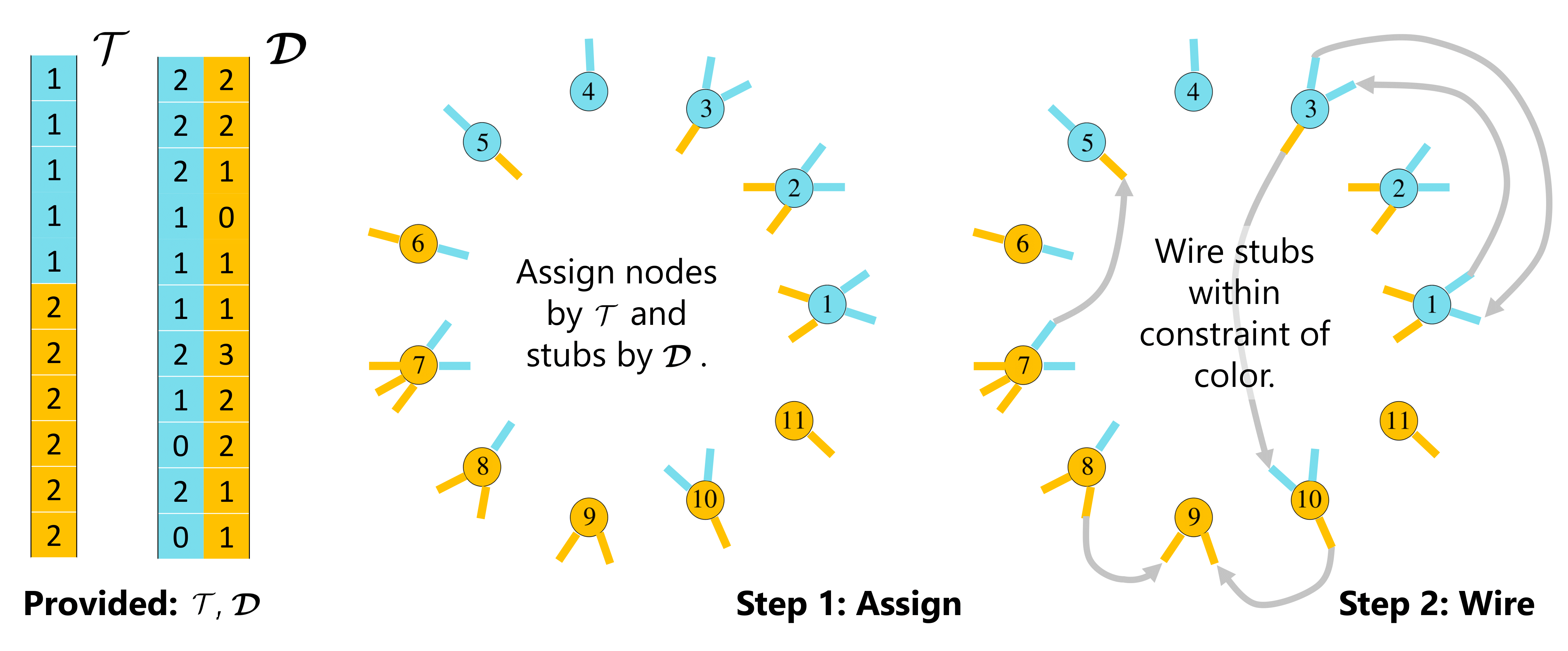}
    \caption{A schematic of the heterogeneous degree configuration model: a model of a random network maintaining the collection of heterogeneous degree sequences, $\boldsymbol{\mathcal{D}}$. A generative form of the model begins with $n$ isolate nodes with node type corresponding to $\mathcal{T}$. Colored stubs, which act as half-edges, are assigned according to each node's heterogeneous degree sequence and the adjacent nodes' type, provided by $\boldsymbol{\mathcal{D}}$. A stub is selected where the color of the stub selected indicates the color of the node to which it must be matched. The partnering stub is selected uniformly at random from the set of available partners, and the process is repeated until no stubs remain.}
    \label{fig:hdcm_flow_stub}
\end{figure}

\subsection{Defining Significance of Connectivity}
\label{sec:sig_of_connection}

We evaluate the significance of the connectivity between an arbitrary node of type $l$, $u^{[l]}\in V^{[l]}$, and an arbitrary subset of the node set, $\mathcal{B}\subseteq\mathcal{V}$, using a $p$-value derived from the reference distribution under the heterogeneous degree configuration model. Since each node maintains exactly one node type,  $\mathcal{B}$ can be partitioned into $K$ subsets: $\mathcal{B}=\cup_{k=1}^{K}{B^{[k]}}$ where $B^{[k]}$ denotes the set of type $k$ nodes in $\mathcal{B}$. To assess the significance of connectivity between $u^{[l]}$ and the set $\mathcal{B}$, we compare the observed number of connections between $u^{[l]}$ and nodes in $B^{[k]}$ for all $k\in\{1, \dots, K\}$ to the distribution under the HDCM. Let $x^{[k]}(u^{[l]} : \mathcal{B})$ denote the observed number of type $k$ nodes in $\mathcal{B}$ (equivalently, in $B^{[k]}$) adjacent to $u^{[l]}$ for $k\in\{1, \dots, K\}$. For comparison, let $X^{[k]}(u^{[l]} : \mathcal{B})$ represent the random variable for the number of type $k$ nodes in $\mathcal{B}$ adjacent to $u^{[l]}$ in a random network contrived from the $\text{HDCM}$ for $k\in\{1,\dots, K\}$. 

Under the HDCM, if the joint probability of attaining at least as many nodes in $B^{[k]}$ adjacent to $u^{[l]}$ as observed for all $k\in\{1, \dots, K\}$ is sufficiently small, then we say $u^{[l]}$ is \textit{well-connected} to the set $\mathcal{B}$ where the joint probability is expressed as
\begin{align}
    p_{\mathcal{B}}(u^{[l]}) &\vcentcolon= 
    \mathbb{P}\left(\bigcap_{k=1}^{K}\left\{X^{[k]}(u^{[l]} : \mathcal{B}) \geq x^{[k]}(u^{[l]} : \mathcal{B}) \right\}\right) \label{eqn:pbu} \\ &=
    \prod_{k=1}^{K} \mathbb{P}\left(X^{[k]}(u^{[l]} : \mathcal{B}) \geq x^{[k]}(u^{[l]} : \mathcal{B})\right) \label{eqn:pbu_ind}.
\end{align}
Small values of $p_\mathcal{B}(u^{[l]})$ are unlikely under the assumption that edges were randomly constructed under the $\text{HDCM}$, providing evidence that $u^{[l]}$ is well-connected to $\mathcal{B}$. In this way, $p_{\mathcal{B}}(u^{[l]})$ is reminiscent of a $p$-value for testing the hypothesis that $u^{[l]}$ is well-connected to nodes in $\mathcal{B}$, and if $p_\mathcal{B}(u^{[l]})$ is less than some prespecified $\alpha\in(0, 1)$, then we conclude that $u^{[l]}$ belongs to $\mathcal{B}$. Since the connections between different pairs of node types are independent in the HDCM, \eqref{eqn:pbu} is the product of $K$ simpler probability statements quantifying the connectivity between $u^{[l]}$ and each node type, as shown in \eqref{eqn:pbu_ind}. 

To compute $p_\mathcal{B}(u^{[l]})$, we must characterize the distribution of $X^{[k]}(u^{[l]} : \mathcal{B})$: the random number nodes in $B^{[k]}$ adjacent to $u^{[l]}$ under the HDCM. The exact distribution for $X^{[k]}(u^{[l]} : \mathcal{B})$ is difficult to attain since enumerating every network in the sample space is infeasible even for networks on the order of 100 nodes. It is possible to approximate the distribution by sampling networks via $\text{HDCM}(\mathcal{T}, \boldsymbol{\mathcal{D}})$, but this process can be memory intensive. Instead, we analytically characterize the random generative process for the $\text{HDCM}(\mathcal{T}, \boldsymbol{\mathcal{D}})$ and provide Theorem \ref{thm:asym_dist}, which describes the asymptotic behavior of $X^{[k]}(u^{[l]} : \mathcal{B})$ and provides the theoretical foundation for a reasonable approximation of $p_\mathcal{B}(u^{[l]})$.

\begin{setup}
Let $\{\mathcal{G}_n\}_{n\geq 1}$ denote a sequence of observed, heterogeneous networks where $\mathcal{G}_n=(\mathcal{V}_n, \mathcal{E}_n)$ and $\mathcal{V}_{n}=\{1, \dots, n\}$ with node types $\mathcal{T}_n$. The collection of heterogeneous degree sequences of $\mathcal{G}_n$ is denoted $\boldsymbol{\mathcal{D}}_n$.  Let $\{\mathcal{H}_n\}_{n\geq 1}$ denote a sequence of random, heterogeneous networks where $\mathcal{H}_n$ is constructed via $\text{HDCM}(\mathcal{T}_n, \boldsymbol{\mathcal{D}}_n)$. Let $\{\mathcal{B}_n\}_{n\geq 1}$ denote a sequence of sets of nodes where $\mathcal{B}_n\subseteq\mathcal{V}_n$, and let $c\geq 1$ be fixed. For each $n\geq 1$, let $V_{n}^{[k]}=\{v\in\mathcal{V}_n : v \text{ is of type } k\}$ and $V^{[k]}\cap V^{[l]}=\emptyset$ for all $k\neq l$ such that $\mathcal{V}_n=\cup_{k=1}^{K} V_{n}^{[k]}$. Furthermore, let $B_n^{[k]}=\{v\in\mathcal{B}_n : v\in V_{n}^{[k]}\}$ such that $\mathcal{B}_n = \cup_{k=1}^{K} B_{n}^{[k]}$. Let $u_{n}^{[l]}\in V^{[l]}_n$ denote a type $l$ node with a type $k$ degree of $c$, i.e. $d^{[k]}(u_n^{[l]})=c$, and let $F^{[k]}_{l,n}$ denote the empirical distribution of type $k$ degrees of type $l$ nodes in $\mathcal{V}_n$, i.e., the empirical distribution of $\{d^{[k]}(v) : v\in V^{[l]}\}$.
\end{setup}

\noindent Here we introduce two necessary assumptions. Assumption \ref{eqn:lttype_prop} states that the limiting proportion of type $k$ nodes is non-zero for $k\in\{1, \dots, K\}$. In particular, no node type vanishes as the network grows. Assumption \ref{eqn:nonzero_ktypedeg_mean} posits the existence of a limiting distribution of the type $k$ degrees of type $l$ nodes, for which there exists a finite, limiting mean type $k$ degree of type $l$ nodes for all $k,l\in\{1, \dots, K\}$. Informally speaking, the general affinity for nodes of type $k$ and $l$ to connect is limiting for all $k,l\in\{1, \dots, K\}$.

\begin{assumption}
\label{eqn:lttype_prop}

$|V_n^{[k]}|/|\mathcal{V}_n|\to\gamma^{[k]}$ as $n\to\infty$ for all $k\in\{1, \dots, K\}$ such that $\gamma^{[k]}\in(0, 1)$ and $\sum_{k=1}^{K}\gamma^{[k]}=1$.
\end{assumption}

\begin{assumption}
\label{eqn:nonzero_ktypedeg_mean}
\label{eqn:aconvdistn}
\label{eqn:aconvmean}

There exists a cumulative distribution function $F_{l}^{[k]}$ on $[0,\infty)$ with $$0 \leq \mu_{l}^{[k]}\vcentcolon=\int_{\mathbb{R}^{+}}{x\,dF_{l}^{[k]}(x)}<\infty,$$ such that $F^{[k]}_{l,n} \overset{d}{\to} F_{l}^{[k]}$ and $\mu_{l,n}^{[k]} \vcentcolon= \int_{\mathbb{R}^{+}}{x\,dF_{l,n}^{[k]}(x)} \to \mu_{l}^{[k]}$ as $n\to\infty$ for all $k,l\in\{1, \dots, K\}$. 
\end{assumption}

\begin{theorem}
\label{thm:asym_dist}
\noindent Under Assumptions \ref{eqn:lttype_prop} and \ref{eqn:nonzero_ktypedeg_mean}, if $\{X_n^{[k]}(u_n^{[l]} : \mathcal{B}_n)\}_{n\geq 1}$ denotes the sequence of random variables of the number of type $k$ nodes in $\mathcal{B}_n$ adjacent to $u_n^{[l]}$ in $\mathcal{H}_n$, then $d_{TV}\left(X_n^{[k]}(u_n^{[l]}, \mathcal{B}_n), Y_{l,n}^{[k]}(\mathcal{B}_n) \right)\to 0$ as $n\to\infty$, where $Y_{l,n}^{[k]}(\mathcal{B}_n)\sim \text{Binom}(c, p_{l,n}^{[k]}(\mathcal{B}_n))$ and
\begin{align}
p_{l,n}^{[k]}(\mathcal{B}_n) =
\frac{\sum_{w\in B_n^{[k]}}{
d^{[l]}(w)}}
{2^{\mathbb{I}(k=l)}|E_n^{[kl]}|}
. \label{eqn:psuccess}
\end{align}
\end{theorem}

\noindent This theorem states that the total variation distance, $d_{TV}$, between the distribution of the number of type $k$ nodes in a subset of the node set adjacent to an arbitrary node of type $l$ and the distribution of a binomial random variable tends to zero as the network grows. Given this result, one can reasonably approximate the significance of connectivity between a node, $u^{[l]}$, and a set of nodes, $\mathcal{B}$, as $\prod_{k=1}^{K} \mathbb{P}\left(Y_{l}^{[k]}(\mathcal{B}) \geq x^{[k]}(u^{[l]} : \mathcal{B})\right)$. However, the parameter in \eqref{eqn:psuccess} depends on whether $u^{[l]}$ is an element of $\mathcal{B}$, capturing the possibility of $u^{[l]}$ connecting with itself when $u^{[l]}$ is in $\mathcal{B}$. This is unsatisfactory as it impacts the convergence of ECoHeN.  Therefore, in Corollary \ref{cor:asym_dist}, we provide an alternative specification of \eqref{eqn:psuccess} that disregards the possibility of $u^{[l]}$ connecting to itself. This result mirrors Theorem \ref{thm:asym_dist} since the density of self-loops in the HDCM tends to zero as the network becomes large.

\begin{corollary}
\label{cor:asym_dist}
Under Assumptions \ref{eqn:lttype_prop} and \ref{eqn:nonzero_ktypedeg_mean}, if $Y_{l,n}^{[k]}(u_n^{[l]}, \mathcal{B}_n)\sim \text{Binom}(c, p_{l,n}^{[k]}(u_n^{[l]}, \mathcal{B}_n))$ where
\begin{align}
p_{l,n}^{[k]}(u_n^{[l]}, \mathcal{B}_n) =
\frac{\left[\sum_{w\in B_n^{[k]}}{
d^{[l]}(w)}\right] - \mathbb{I}(k=l)\mathbb{I}(u_n^{[l]}\in\mathcal{B}_n)c}
{2^{\mathbb{I}(k=l)}|E_n^{[kl]}| - \mathbb{I}(k=l)c}, \label{eqn:psuccess_adj}
\end{align}
then $d_{TV}\left(X_n^{[k]}(u_n^{[l]}, \mathcal{B}_n), Y_{l,n}^{[k]}(u_n^{[l]}, \mathcal{B}_n) \right)\to 0$ as $n\to\infty$.
\end{corollary}

\noindent Provided Corollary \ref{cor:asym_dist}, one can then reasonably approximate the significance of connectivity between a node, $u^{[l]}$, and a set of nodes, $\mathcal{B}$, as $\hat{p}_\mathcal{B}(u^{[l]})=\prod_{k=1}^{K} \mathbb{P}\left(Y_{l}^{[k]}(u^{[l]}, \mathcal{B}) \geq x^{[k]}(u^{[l]} : \mathcal{B})\right)$, where $\hat{p}_\mathcal{B}(u^{[l]})$ is no longer dependent on whether $u^{[l]}$ is in $\mathcal{B}$, taking on the same value in either case and allowing us to characterize the convergence properties of ECoHeN in Section \ref{sec:extraction}. We provide a proof of Theorem \ref{thm:asym_dist} and Corollary \ref{cor:asym_dist} in Supplement B.

\subsection{ECoHeN Algorithm}
\label{sec:ecohen}

The ECoHeN algorithm comprises three operations: \textbf{Initialization}, \textbf{Extraction}, and \textbf{Refinement}. At the initialization step, a collection of seed sets of nodes is chosen, where $\mathcal{B}_0$ denotes a single seed set. Following, each seed set is iteratively updated according to the set of $\{\hat{p}_{\mathcal{B}_{i}}(u) : u\in\mathcal{V}\}$ until no nodes are included or excluded, where $\mathcal{B}_i$ denotes the set after $i$ updates. The resulting collection of non-empty sets of nodes represent the extracted communities and is denoted $\boldsymbol{\mathcal{C}}_T = \{\mathcal{C}_t\}_{t\in T}$, where $\mathcal{C}_t$ denotes a single extracted community. Lastly, the collection of extracted communities are refined according to a practitioner's preferences, such as desired community sizes and maximal pairwise overlap between communities. The refined communities are denoted $\boldsymbol{\mathcal{C}}_H$, where $H\subseteq T$. The full algorithm is presented in Algorithm \ref{algo:ecohen_algo}, with supplemental extraction routines provided in Supplement D. We detail each component of the algorithm in the following subsections: 

\begin{palgorithm}[htbp!]
    \centering
    \includegraphics[scale=0.98]{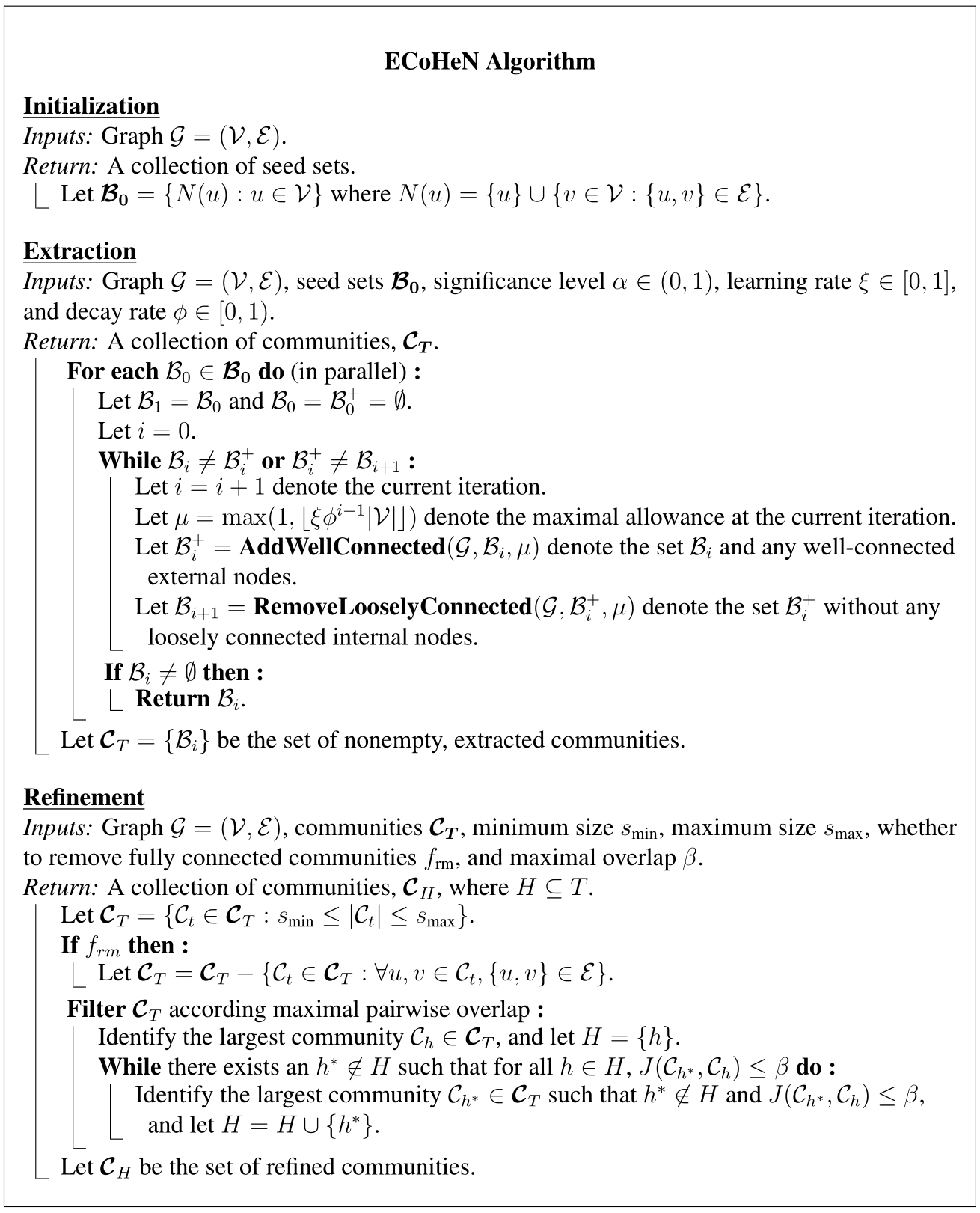}
    \captionof{palgorithm}{Pseudocode for the ECoHeN algorithm. We include pseudocode for the extraction routines AddWellConnected and RemoveLooselyConnected in Supplement D.}
    \label{algo:ecohen_algo}
\end{palgorithm}

\subsubsection{Initialization}

By default, the collection of seed sets, denoted $\boldsymbol{\mathcal{B}_0}$, consists of the neighborhood of each node in the observed network $\mathcal{G}=(\mathcal{V}, \mathcal{E})$. The neighborhood of a node $u\in\mathcal{V}$ is the set including $u$ and all nodes adjacent to $u$, that is, $N(u)=\{u\}\cup\{v\in\mathcal{V} : \{u, v\}\in\mathcal{E}\}$. Hence, the collection of seed sets considered is $\boldsymbol{\mathcal{B}_0}=\{N(u) : u\in \mathcal{V}\}$. There are other ways to define seed sets; for example, locally optimal neighborhoods has been used in other works \citep{gleich2012vertex, whang2013overlapping}.

\subsubsection{Extraction}
\label{sec:extraction}

Provided a prespecified significance level $\alpha\in(0, 1)$ and a seed set $\mathcal{B}_0\in \boldsymbol{\mathcal{B}_0}$, the extraction procedure iteratively updates the set through a two-step, dynamic procedure until no nodes are recommended for inclusion or exclusion. We denote the set of nodes after $i$ updates as $\mathcal{B}_i$ and the complement as $\mathcal{B}_i^c = \mathcal{V}-\mathcal{B}_i$. At each iteration, external nodes, being all $u\not\in\mathcal{B}_i$, which are well-connected to $\mathcal{B}_i$ are added to $\mathcal{B}_i$. Since this involves $|\mathcal{B}_i^{c}|$ tests, a false discovery rate (FDR) correction \citep{benjamini1995controlling} is evoked. The intermediate set containing $\mathcal{B}_i$ and any additions is denoted $\mathcal{B}_i^{+}$. Following, internal nodes, being all $u\in \mathcal{B}_i^{+}$, which are no longer well-connected to $\mathcal{B}_i^{+}$ are removed. Again, since this involves $|\mathcal{B}_i^{+}|$ tests, an FDR correction is evoked. The set following any removals is denoted $\mathcal{B}_{i+1}$ as this completes one iteration of the update procedure. The procedure continues in this way until no nodes are added or removed, that is, until $\mathcal{B}_{i}=\mathcal{B}_{i}^{+}=\mathcal{B}_{i+1}$.

\paragraph{Maximal Allowance} In ECoHeN's predecessor ESSC, all nodes are simultaneously considered for inclusion and exclusion at each iteration, yet this does not have to be the case. In fact, the choice of how many nodes' memberships (either in $\mathcal{B}_i$ or in $\mathcal{B}_i^{c}$) are updated can have major implications on whether the update procedure converges necessarily and the number of iterations required until convergence. We define the \textit{maximal allowance} at iteration $i$, denoted $\mu_i$, to be the maximum number of nodes permitted into $\mathcal{B}_i$ and the maximum number permitted out of $\mathcal{B}_i^{+}$ at iteration $i$. If $\mu_i=|\mathcal{V}|$ for all $i$, the update routine is relatively fast but can result in \textit{cycles}: the infinite alternation between two sets of nodes (problematic in \citet{wilson2014testing} and \citet{bodwin2018testing}). If $\mu_i=1$ for all $i$, the update routine converges but may take unnecessarily many iterations to converge. As such, we propose initializing the maximum allowance to a large value and then progressively decreasing it to guarantee convergence while reducing the number of iterations. 
We propose defining the maximal allowance using an exponential decay function parameterized by an initial learning rate, $\xi\in[0, 1]$, and a decay rate, $\phi\in[0, 1]$:
\begin{equation}
\label{eqn:max_allowance}
    \mu_i=\max(1, \lfloor \xi\phi^{i-1} |\mathcal{V}| \rfloor).
\end{equation}
The learning rate controls the maximal allowance on the first iteration, and the decay rate controls the maximal allowance thereafter. Together, $\xi$ and $\phi$ control the scale of updates where smaller values of $\xi$ and $\phi$ necessitate small, micro-level changes to $\mathcal{B}_i$ at each iteration, and larger values of $\xi$ and $\phi$ allow for larger, macro-level changes to $\mathcal{B}_i$ at early iterations. The outer maximum operation in \eqref{eqn:max_allowance} ensures that if the number of nodes dictated by the floor function reaches zero before convergence, then one node is permitted to transition into and out of $\mathcal{B}_i$ until convergence. Setting $\xi=\phi=1$ ($\xi=\phi=0$) is equivalent to setting $\mu_i=|\mathcal{V}|$ ($\mu_i=1$) for all $i$; each are problematic as previously discussed. 

We first characterize the convergence properties of the extraction procedure in Theorem \ref{thm:extraction_conv} before discussing the practical implications.

\begin{setup}
Let $\mathcal{G}=(\mathcal{V}, \mathcal{E})$ denote an observed, heterogeneous network. In the extraction procedure, let $\hat{p}_\mathcal{B}(u^{[l]})=\prod_{k=1}^{K} \mathbb{P}\left(Y_{l}^{[k]}(u^{[l]}, \mathcal{B}) \geq x^{[k]}(u^{[l]} : \mathcal{B})\right)$ for any set $\mathcal{B}$  by Corollary \ref{cor:asym_dist} where $u^{[l]}$ is an arbitrary node of type $l$ and $Y_{l}^{[k]}(u^{[l]}, \mathcal{B})\sim \text{Binom}(d^{[k]}(u^{[l]}), p_{l}^{[k]}(u^{[l]}, \mathcal{B}))$.
\end{setup}

\begin{theorem}
\label{thm:extraction_conv}
Let $\alpha\in(0, 1)$, $\xi\in[0, 1]$, and $\phi\in[0, 1)$ be fixed constants. Suppose $\mathcal{B}_0\subseteq\mathcal{V}$ denotes a seed set for the extraction procedure. If there exists a $j$ such that $|\mathcal{B}_i| < |\mathcal{V}|/2$ for all $i\geq j$, the extraction procedure will not cycle.
\end{theorem}

\noindent Based on Theorem \ref{thm:extraction_conv}, setting $\phi<1$ is enough to guarantee ECoHeN will never cycle between two sets, no matter the choice of $\xi$. While we cannot prove theoretically that ECoHeN will not alternate between more than two sets (e.g., a three-set cycle), this phenomena has not been observed empirically. When $\xi=1$ and $\phi<1-\epsilon$ for small $\epsilon>0$, the maximal allowance is relatively large for many early iterations, often prompting early convergence compared to $\xi=\phi=0$, while simultaneously guaranteeing convergence. The impact of $\xi$ and $\phi$ on the quantity and quality of communities found via simulation is documented in Section \ref{sec:effects_of_choices}. We provide a proof of Theorem \ref{thm:extraction_conv} in Supplement B.

\subsubsection{Refinement}
\label{sec:refinement}

The extraction procedure identifies a collection $\boldsymbol{\mathcal{C}}_T = \{\mathcal{C}_t\}_{t\in T}$ of communities based on the collection of seed sets $\boldsymbol{\mathcal{B}_0}$. In many applications, only communities with certain properties are considered desirable. For example, the DREAM challenge, described in \citet{choobdar2018open}, required all communities submitted for competition be non-overlapping and contain $3$-$100$ members. In other settings, community discovery is used as an aid for link prediction \citep{soundarajan2012using}, in which case fully connected communities are not of interest. The refinement process reduces the collection of discovered communities to a subset $\boldsymbol{\mathcal{C}}_H$, $H\subseteq T$, of communities satisfying the requested constraints.

While the constraints on size and connectedness are easily implemented, refining a collection of communities based on overlap requires more care. Let $\beta\in[0, 1]$ denote a user specified maximum pairwise overlap between communities based on the Jaccard similarity measure, denoted $J$. The refinement process begins by identifying the largest community $\mathcal{C}_{h}$, $h\in T$, and letting $H=\{h\}$. Following, we identify the largest community $\mathcal{C}_{h^*}$, such that $h^*\not\in H$ and $J(\mathcal{C}_{h^*}, \mathcal{C}_{h})\leq\beta$ for all $h\in H$, and set $H=H\cup\{h^{*}\}$. The process continues until no additional communities maintain the overlap constraint dictated by $\beta$, i.e., until for all $h^*\not\in H$ there exists an $h\in H$ such that $J(\mathcal{C}_{h^*}, \mathcal{C}_{h}) > \beta$. The refinement procedure is similar to that of \citet{wilson2017community}.

\subsection{Parameter Choices}
\label{sec:effects_of_choices}

ECoHeN is implemented as an \texttt{R} package \texttt{ECoHeN} \citep{gibbs2022ecohengit} available on \href{https://github.com/ConGibbs10/ECoHeN}{GitHub} with the full algorithm provided in pseudocode in Supplement D. The extraction procedure is implemented in \texttt{C++} for efficiency with parallelized extractions across initial seed sets. ECoHeN allows the user to specify the significance level $\alpha$, learning rate $\xi$, and decay rate $\phi$. A brief description of each parameter, the parameter ranges, and default recommendations are provided in Table \ref{tbl:parameter_settings}. The significance threshold, $\alpha$, controls how conservative one would like to be in defining a community; smaller values of $\alpha$ result in fewer, smaller communities which tend to be incredibly densely connected compared to the rest of the network. Larger values of $\alpha$ result in more, larger communities which are less dense.

\begin{table}[htb!]
\centering
\begin{tabular}{p{0.19\textwidth}p{0.075\textwidth}p{0.655\textwidth}}
\toprule
Parameter & Default & Description \\
\midrule
\begin{tabular}[t]{l} \hspace{-0.75em} Significance level \\ 
\hspace{0.25em} \textbullet \hspace{0.25em} $\alpha\in(0, 1)$
\end{tabular}
& $0.10$
& Determines whether $\hat{p}_\mathcal{B}(u)$ is small (large) enough to justify the inclusion (exclusion) of node $u$ to the set $\mathcal{B}$ for each $u$ at each iteration of the extraction procedure. \\
\begin{tabular}[t]{l} \hspace{-0.75em} Learning rate \\ 
\hspace{0.25em} \textbullet \hspace{0.25em} $\xi\in[0, 1]$
\end{tabular} 
& $1$
& Controls the maximal allowance on the first iteration of the extraction procedure; $\xi=1$ allows for large changes to $\mathcal{B}$ on the first iteration. \\
\begin{tabular}[t]{l} \hspace{-0.75em} Decay rate \\ 
\hspace{0.25em} \textbullet \hspace{0.25em} $\phi\in[0, 1]$
\end{tabular}
& 0.99 
& Controls the maximal allowance after the first iteration of the extraction procedure; $\phi=1-\epsilon$ for small $\epsilon$ allows for large changes to $\mathcal{B}$ for many iterations, often reducing the number of iterations until guaranteed convergence. \\
\bottomrule
\end{tabular}
\caption{Descriptions and notation for each parameter along with default suggestions. Defaults are set to speed up the extraction procedure with guaranteed convergence and minimal loss in the power to discover communities.}
\label{tbl:parameter_settings}
\end{table}

The learning and decay rates control the maximum allowance at each iteration of the extraction procedure and should be set together. If $\xi=\phi=0$, then the extraction procedure provides the greatest resolution in identifying a single community from background noise (see Supplement C); however, there are a number of downsides associated with this setting, including the extraction of many communities with substantial pairwise overlap, many required iterations before convergence, and most notably, the identification of communities in truly random networks (see Supplement C). 

When $\xi=\phi=1$, the extraction procedure often requires fewer iterations until convergence and results in fewer communities with less pairwise overlap; however, convergence of the algorithm is not guaranteed. Notably, when $\xi=\phi=1$, densely connected communities of interest may be unidentified as the extraction procedure may cycle between densely connected sets of nodes until a maximum number of iterations is reached and the extraction procedure is terminated. 

Setting $\xi=1$ and $\phi<1$ guarantees that ECoHeN will converge and resolves issues with identifying communities in random networks, allowing ECoHeN to escape the low conductance seed sets; however, communities may still have substantial pairwise overlap. The problem is abated for $\phi$ closer to one which is also shown in Supplement C, whereby the algorithm is best able to identify a single community from background noise when $\xi=1$.  Based on these findings, the default choice for the learning and decay rates are $\xi=1$ and $\phi=0.99$, respectively. The refinement procedure further reduces the degree of pairwise overlap in the final set of identified communities when $\beta<1$.

\section{Simulation Study}
\label{sec:simulation}

In this section, we investigate the conditions for which ECoHeN identifies simulated community structure. We leverage an extension of the classic stochastic block model, referred to here as the heterogeneous stochastic block model (HSBM), to generate the heterogeneous networks under study. Described in \citet{zhang2020modularity} for two node types and implemented in \citet{gibbs2022ecohengit} for $K$ node types, the HSBM is a flexible framework for generating networks with numerous node types and is detailed in Supplement C. In this study, we consider networks with 500 red (i.e., type one) and 500 blue (i.e., type two) nodes split between a background block and a high connectivity block (HCB). Nodes are assigned to blocks according to the parameter $p$ where connections between nodes are determined stochastically according to parameters $\{b, r_{11}, r_{22}, r_{12}\}$. We detail these parameters in the following paragraph using Figure \ref{fig:motivate_hsbm} to motivate their definition.

\begin{figure}[h!]
    \centering
    \includegraphics[scale=0.80]{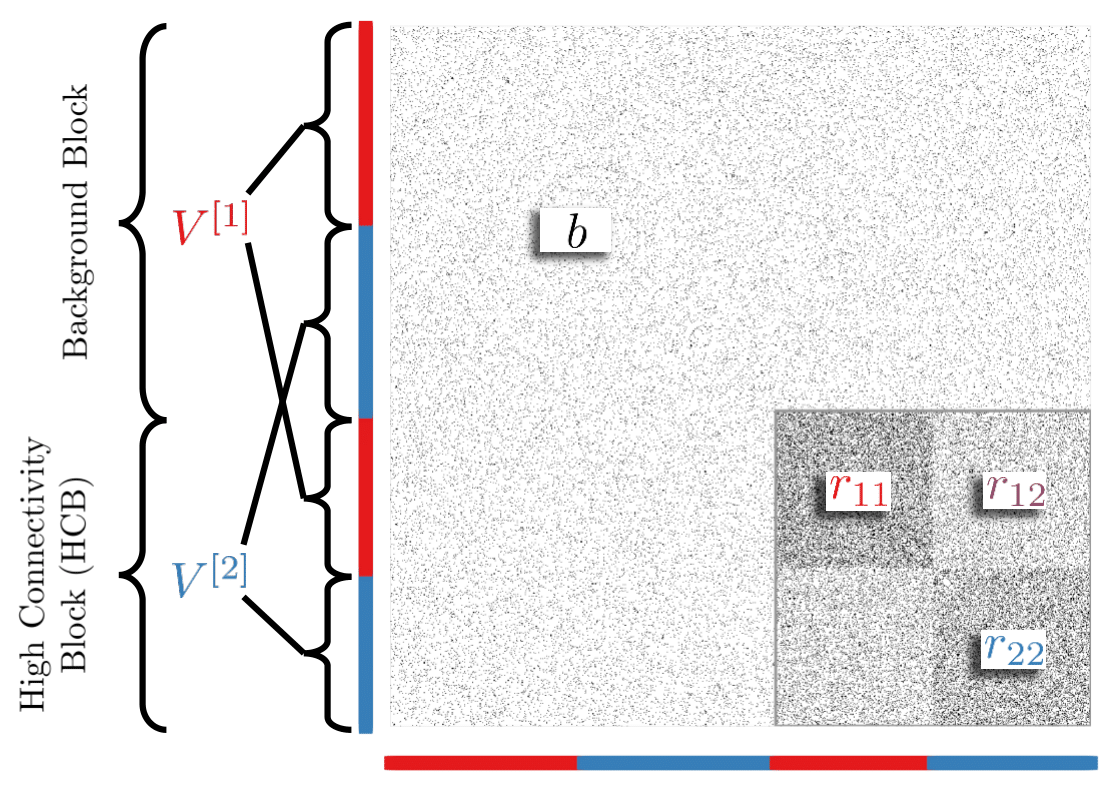}
    \caption{The adjacency matrix of a network generated under the HSBM with parameters $p=0.45$, $b=0.05, r_{11}=0.25, r_{22}=0.20$, and $r_{12}=0.10$. Node type is illustrated by the colored bars adorning the axes, and links are illustrated by black dots. Nodes are assigned to blocks according to $p$ where edges are sampled independently according to Bernoulli probabilities. Two nodes connect with probability $b$ if each are in the background block or do not share a block. The parameter $r_{ij}$, $1\leq i \leq j\leq 2$, provides the additive increase in the rate of connection between nodes of type $i$ and $j$ if each node is in the HCB. Connections between nodes in the HCB are bordered in a black square. When $r_{ij}>0$ for all $1\leq i\leq j\leq 2$, the network embodies heterogeneous community structure with one community composed of red and blue nodes. When $r_{ii}>0$ for $i\in\{1, 2\}$ and $r_{12}=0$, the network embodies homogeneous community structure with two communities: one composed of red nodes, and one composed of blue nodes. When $r_{ij}=0$ for all $1\leq i\leq j\leq 2$, the network is an ER network with no underlying community structure.}
    \label{fig:motivate_hsbm}
\end{figure}

Each network under study is generated from the HSBM with parameters $\{b, r_{11}, r_{22}, r_{12}\}$. The parameter $p$ is the proportion of nodes of each type assigned to the HCB. A larger $p$ corresponds to a larger community size. Connections between nodes are sampled according to Bernoulli random variables with expectation dependent on the block assignment and node type of each node. The probability of connection between two nodes is $b$, interpreted as the \textit{background rate}, if each node be a member of the background block or members of different blocks (i.e., one in the background and one in the HCB). The probability of connection between two nodes in the HCB is $b+r_{ij}$ where $r_{ij}$ provides the additive increase in the probability of connection between nodes of type $i$ and $j$. Figure \ref{fig:motivate_hsbm} is constructed with parameters $p=0.45$, $b=0.05$, $r_{11}=0.25$, $r_{22}=0.20$, and $r_{12}=0.10$. 

We wish to characterize ECoHeN's ability to identify two classes of community structure: heterogeneous and homogeneous community structure. The network presented in Figure \ref{fig:motivate_hsbm} embodies heterogeneous community structure provided $r_{ij}>0$ for all $1\leq i\leq j\leq 2$. In particular, when $r_{12}>0$, nodes of different type in the HCB have a higher propensity to connect compared to the background, implying the existence of one heterogeneous community composed of both red and blue nodes. When $r_{ii}>0$ for $i\in\{1, 2\}$ and $r_{12}=0$, nodes of different type in the HCB connect at the same rate as the background, implying the existence of two homogeneous communities: one composed of red nodes and one composed of blue nodes. When $r_{ij}=0$ for all $1\leq i\leq j\leq 2$, the network is a heterogeneous analog of an Erdős-Rényi-Gilbert (ER) network \citep{gilbert1959random, erdHos1960evolution} with probability of connection $b$. ECoHeN's ability to assign all nodes to the background in ER networks is explored in Supplement C. 

For each class of community structure, we consider a range of HCB sizes and network connectivity patterns. Three community discovery methods designed for homogeneous networks (i.e., ESSC \citep{wilson2014testing}, Infomap \citep{rosvall2009map}, and Walktrap \citep{pons2006computing}) and two community discovery methods designed for heterogeneous networks (i.e., ECoHeN and ZCmod \citep{zhang2020modularity}) are applied. While no community discovery method is universally optimal, Infomap and Walktrap tend to partition a network into smaller, more dense modules than other canonical community detection methods \citep{smith2020guide} and outperform many competing methods for a variety of benchmarks \citep{javed2018community}. ESSC is a special case of ECoHeN when the number of node types is one, i.e. the network is homogeneous, justifying its inclusion. To our knowledge, ZCmod is the only existing community detection designed to identify heterogeneous community structure. We ignore node type when applying ESSC, Infomap, and Walktrap on the simulated networks.

One-hundred networks are generated for a given simulated condition and are referred to as replicates. For each replicate, we apply the aforementioned community discovery methods and record the respective set of discovered communities. For each method and each replicate, we compute the maximum Jaccard similarity measure between the underlying community structure and the set of discovered communities. A value near one (zero) indicates the community was near perfectly identified (not identified) by the method. ECoHeN, ESSC, and ZCmod are implemented in the \texttt{ECoHeN} package \citep{gibbs2022ecohengit}, whereas Infomap and Walktrap are implemented in the \texttt{igraph} package \citep{csardi2006igraph} of the \texttt{R} programming language \citep{rcore2022programming}.

\subsection{Heterogeneous Community Structure}

To generate heterogeneous networks with heterogeneous community structure, we fix $b=0.05$, and let $r_{ii}\in(0.15, 0.20, 0.25, 0.30)$ for $i\in\{1, 2\}$ and $r_{12}\in(0.025, 0.05, 0.075)$. We allow the proportion of nodes assigned to the HCB to vary according to $p\in (0.05, 0.10, 0.15, 0.20)$. Hence, we consider a total of 192 simulated conditions. Random networks under these settings have heterogeneous community structure, particularly one community composed of red nodes and blue nodes.

\begin{figure}[htbp!]
    \centering
    \includegraphics[width=\textwidth]{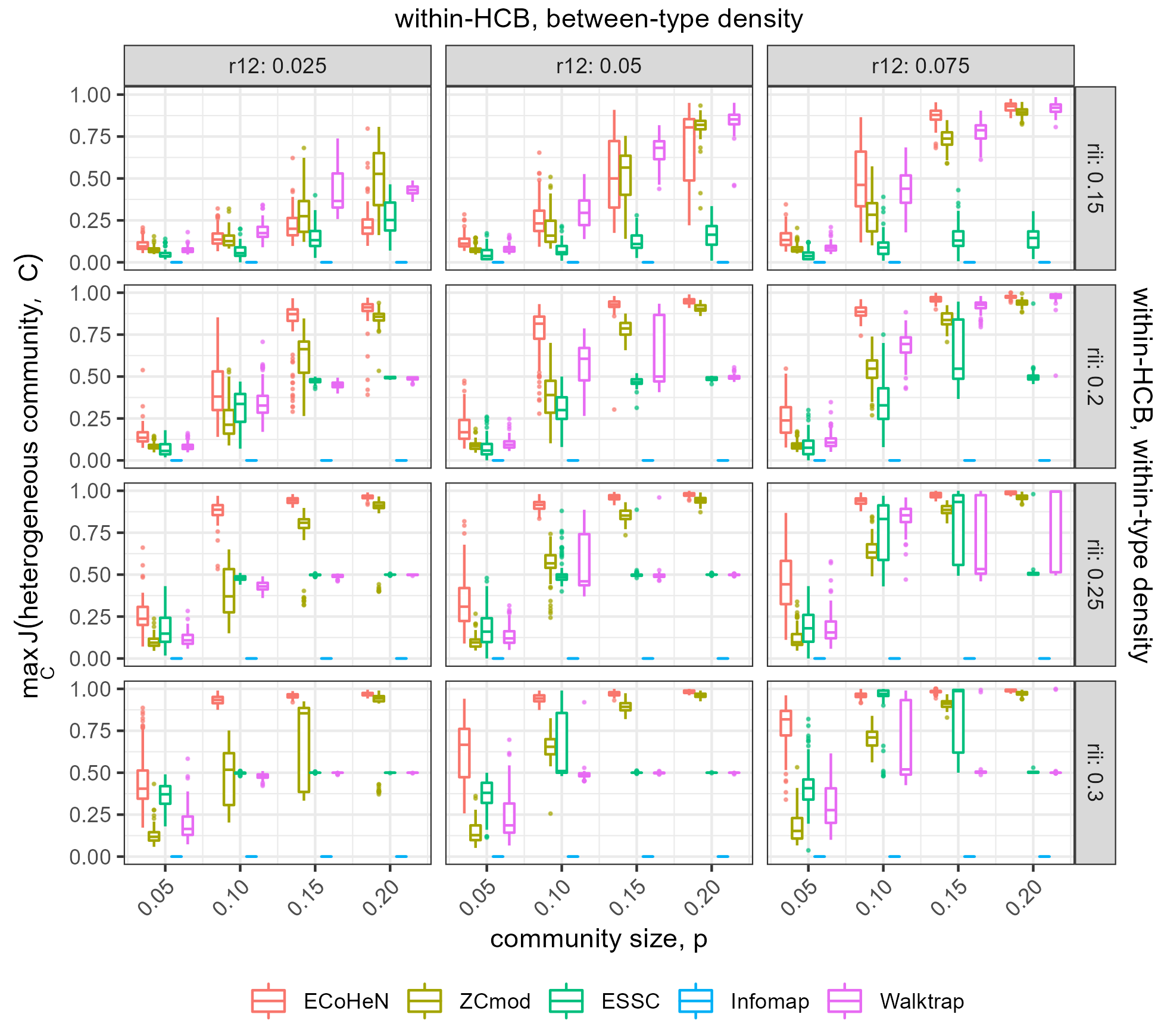}
    \caption{Each method's ability to identify the heterogeneous community is shown as a function of the size of the heterogeneous community. As the density of within-type and between-type links increases, ECoHeN identifies the heterogeneous community with increasingly better precision and marked improvements for small, heterogeneous communities. In nearly all simulated conditions, ECoHeN performs as well or better than ZCmod. The performance of ESSC and Walktrap is highly varied compared to ECoHeN and ZCmod, often varying between 0.5 and 1 (e.g., $r_{12}=0.075$ and $r_{ii}=0.25$). Unlike ECoHeN and ZCmod, these methods fail to account for node type, ergo fail to recognize  the heterogeneous community for some replicates and instead identify the two homogeneous subsets of the HCB.}
    \label{fig:best_jaccard_multicolor_ecohen}
\end{figure}

Figure \ref{fig:best_jaccard_multicolor_ecohen} shows the results for identifying the heterogeneous community when $r_{11}=r_{22}$. ECoHeN and ZCmod's ability to identify the heterogeneous community notably improves as the within-HCB, within-type density (i.e., $b + r_{ii}$) increases \emph{and} the within-HCB, between-type density (i.e., $b + r_{ij}$) increases. On the other hand, ESSC and Walktrap's ability to identify the heterogeneous community improves for increasing $r_{ii}$ and $r_{12}$ but with more variability and less precision compared to ECoHeN and ZCmod. In general, Infomap fails to find any communities, consistently returning a trivial partition of 1000 communities each containing one node. ECoHeN performs as well or better than ZCmod at each simulated condition, consistently outperforming ZCmod at identifying small, heterogeneous communities (i.e., when $p$ is small). Since ZCmod is a modularity optimization method, it appears to suffer from a known resolution limit characteristic of modularity optimization methods \citep{fortunato2007resolution}, struggling to recover small communities with a reasonable degree of accuracy. For larger values of $p$, ECoHeN performs similarly to ZCmod in its ability to recover the heterogeneous community. A broad range of simulated conditions are considered and presented in Supplement C.

For many settings (e.g., $r_{12}=0.075$ and $r_{ii}=0.25$), the maximum Jaccard for ESSC and Walktrap will vary from 0.5 to one with the median often equal to 0.5 or one. In these settings, ESSC and Walktrap identify the heterogeneous community for some replicates (resulting in a Jaccard of one) and the homogeneous communities for other replicates (resulting in a Jaccard of 0.5). On the other hand, ECoHeN and ZCmod consistently identify the heterogeneous community in these settings, recognizing that the within-HCB, between-type density is larger than the background rate (e.g., 2.5 times larger when $r_{12}=0.075$), albeit much less than the within-HCB, within-type density. By comparison, ESSC and Walktrap will only identify the heterogeneous community when the overall density of the HCB is sufficiently high because the between-type density is sufficiently high.

\subsection{Homogeneous Community Structure}

To generate heterogeneous networks with homogeneous community structure, we fix $b=0.05$, and let $r_{ii}\in(0.15, 0.20, 0.25, 0.30)$ for $i\in\{1, 2\}$ and $r_{12}=0$. Again, we allow the proportion of nodes assigned to the HCB to vary according to $p\in (0.05, 0.10, 0.15, 0.20)$. Random networks under these settings have homogeneous community structure, particularly two communities: one composed of red nodes, i.e. the red community, and one composed of blue nodes, i.e. the blue community. We start by investigating each method's ability to identify the red community when $r_{11}\in(0.20, 0.25, 0.30)$ and $r_{22}=0.25$ for a total of 36 simulated conditions.

\begin{figure}[htb!]
    \centering
    \includegraphics[width=\textwidth]{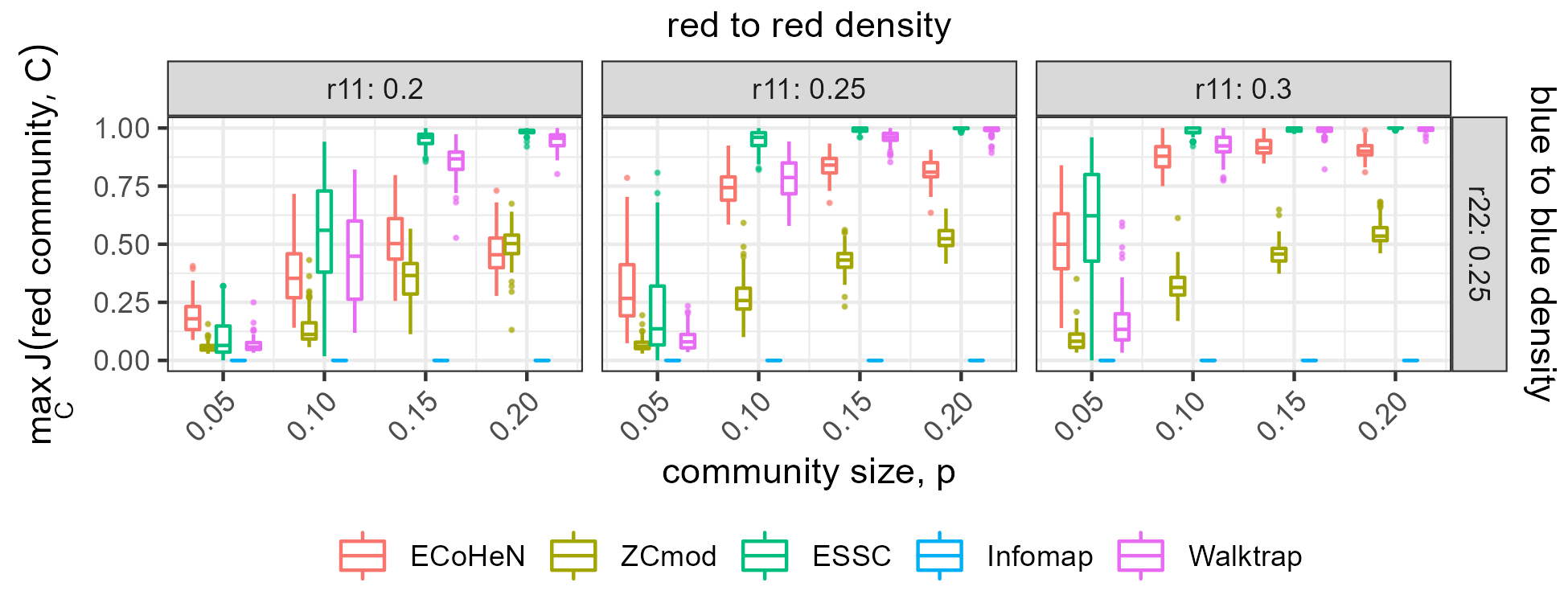}
    \caption{Each method's ability to identify the red, homogeneous community is shown for varying red community sizes $p$. As the density of red-to-red links increases, ECoHeN identifies the red community with increasingly better precision and is notably better when the  red community is small. ECoHeN can recover both homogeneous and heterogeneous community structure; however, the tradeoff for this functionality is a reduction in power for identifying homogeneous communities. Nevertheless, ECoHeN's ability to recover the homogeneous community far surpasses ZCmod, which requires each community maintain at least one node of each node type.}
    \label{fig:best_jaccard_blue_ecohen}
\end{figure}

Figure \ref{fig:best_jaccard_blue_ecohen} shows the results for identifying the red community when the density of blue-to-blue links in the HCB is fixed according to $r_{22}=0.25$. As the density of red-to-red links in the HCB, controlled by $r_{11}$, increases, ECoHeN can identify the red community with increasingly better precision and with marked improvements for small, homogeneous communities (i.e., when $p$ is small). There are no such improvements from ZCmod which requires each discovered community maintain at least one node of each type. At the same time, ESSC and Walktrap consistently outperform ECoHeN and ZCmod at identifying homogeneous community structure. This is not surprising considering these methods identify communities irrespective of node type. While ECoHeN is designed to identify both homogeneous and heterogeneous community structure, the tradeoff for this functionality is a reduction in power for identifying homogeneous communities.  This is unsurprising as ECoHeN considers both same-type and between-type edges simultaneously, rather than just same-type edges, looking for evidence of excess connectivity. 

Comparing methods designed for homogeneous networks, ESSC performs better or similarly to Walktrap at identifying dense, homogeneous communities. Infomap continues to return a trivial partition of 1000 communities containing one node each, incapable of identifying a community amongst background noise. The results over all simulated conditions, provided in Supplement C, demonstrate that the density of blue-to-blue links in the HCB, controlled by $r_{22}$, does not impact any method's ability to identify the red community. Furthermore, the results for identifying the blue community are provided in Supplement C along with a broad range of simulated conditions.

\section{Empirical Study}
\label{sec:results}

To illustrate the utility of ECoHeN in practice, we extract communities from the political blogs network of \citet{adamic2005political}. This iconic network consists of political blogs (represented as nodes) and the hyperlinks between them (represented as undirected edges). Collected shortly after the 2004 U.S. presidential election, the largest connected component of the political blogs network consists of 1222 blogs and 16,714 links. As seen in Figure \ref{fig:heterogeneous_net_examples}\textcolor{blue}{a}, blogs were classified according to their political ideology based on a text analysis of their content, where the 636 red nodes represent conservative leaning blogs and the 586 blue nodes represent liberal leaning blogs. There are drastically more connections between blogs of the same political ideology (precisely 15,139) than connections between blogs of differing political ideology (precisely 1,575). This translates to a propensity of connection between liberal (conservative) blogs of 0.043 (0.039), whereas the propensity of connection between liberal and conservative blogs is 0.004.

The political blogs network has been studied time and time again within the community discovery literature \citep{newman2006modularity, karrer2011stochastic, newman2013spectral, jin2015fast}, and in this vast body of work, political ideology is conflated with community structure. Authors deem their community discovery methods successful after dividing nodes into groups which largely align with the observed political ideology. However, this is arguably a rather trivial partition of the nodes which provides little insight about the connections between liberals and conservatives. \citet{peel2017ground} further warn against treating node metadata, like political ideology, as ground truth for community structure, recognizing that 1) community discovery is largely task dependent for which no method is universally optimal, and 2) the political blogs network has substantial substructure that is often overlooked in favor of the traditional narrative. By conditioning on political ideology, ECoHeN identifies communities of blogs which are densely connected considering the political ideology of each community's members.

When ECoHeN is applied to the political blogs network with the political ideology labels, 81 communities are found. For reasons discussed in Section \ref{sec:refinement}, the number of communities is likely overstated due to a significant amount of overlap among the discovered communities. As such, these 81 communities are refined such that each community has at least four nodes with a $\beta = 0.10$ for maximum Jaccard overlap, which results in a set of 15 communities that overlap yet are largely distinct. The largest partisan (homogeneous) and bipartisan (heterogeneous) communities identified by ECoHeN are presented in Figure \ref{fig:largest_communities_ecohen}.

\begin{figure}[htb!]
\centering
\includegraphics[width=\textwidth]{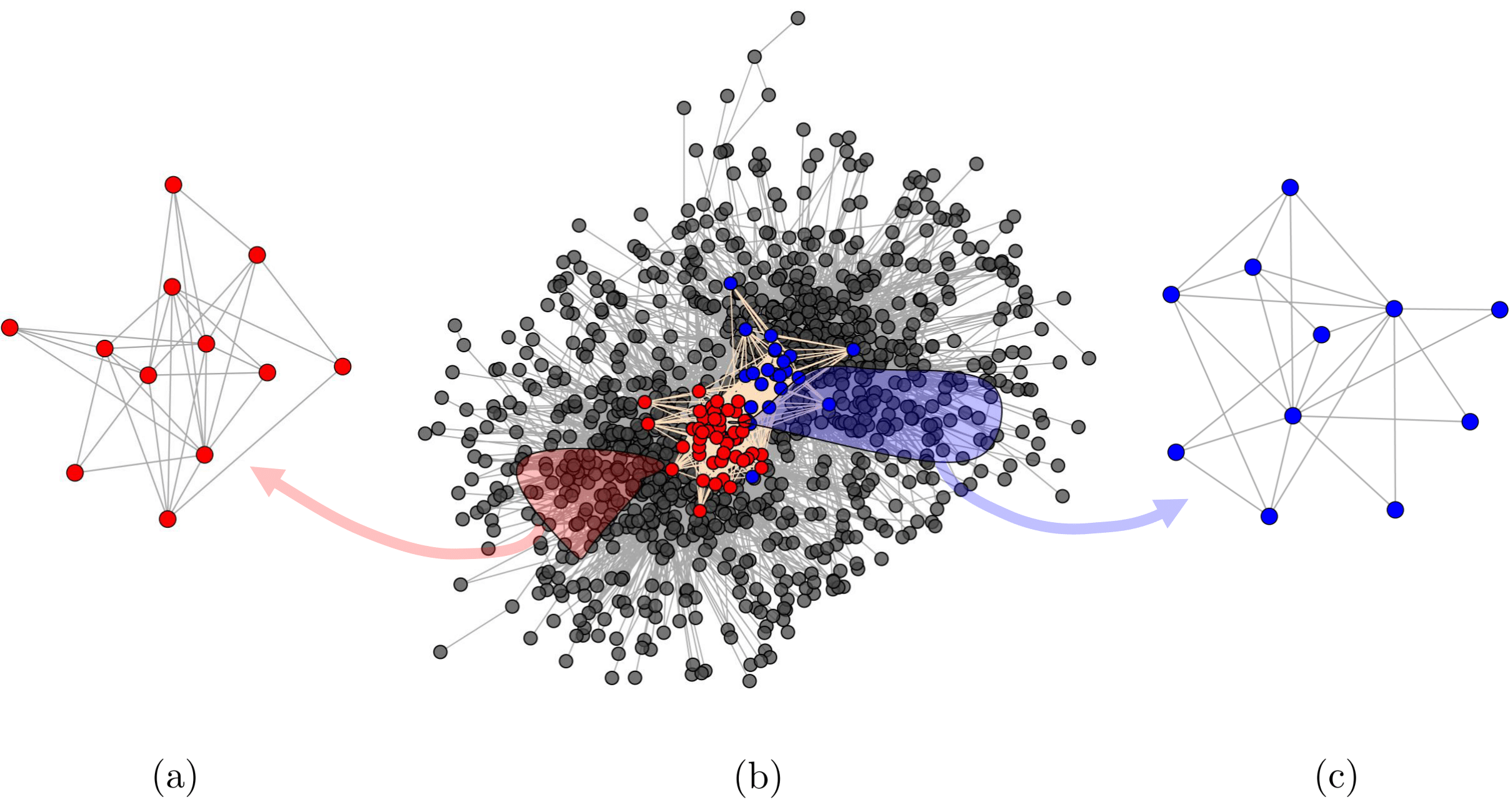}
\caption{The largest partisan and bipartisan communities identified by ECoHeN. Panel (a) depicts the largest conservative community which consists of 12 blogs with a ratio of densities of 23.8. Panel (c) depicts the largest liberal community which consists of 11 blogs with a ratio of densities of 20.8. The colored polygons in panel (b) provide the origin of each partisan community in the political blogs network. Furthermore, the largest community found by ECoHeN is a bipartisan community provided in panel (b) which consists of 73 blogs with a ratio of densities of 8.5.} 
\label{fig:largest_communities_ecohen}
\end{figure}

To gauge the quality of the 15 communities extracted by ECoHeN, we compute the \emph{ratio of densities} (RatD) for each community; that is, we compute the density of links among community members divided by the density of links between community members and the rest of the network. Figure \ref{fig:ratio_vs_size} provides the RatD for all ECoHeN communities. The RatD for all liberal (conservative) blogs is a natural baseline when assessing the assortativity of communities composed near entirely of liberals (conservatives). In general, a RatD of one implies that the density of links within a set of nodes is equivalent to the density of links to the rest of the network and is a natural baseline when assessing the assortativity of an identified community regardless of community members' political affiliation.

For comparison, we apply ZCmod and Walktrap to the political blogs network, attain respectively 11 and 5 communities with at least four members, and compute the RatD for each identified community. To assess political composition of each community, we compute the proportion of liberals (equivalently, conservatives) in each community. Figure \ref{fig:ratio_vs_size} provides the RatD for each identified community along with the size and political composition of the community. When the network is partitioned irrespective of political affiliation with Walktrap, at least one community is largely comprised of liberals and one community is largely comprised of conservatives, a rather trivial partition. When the network is partitioned via heterogeneous modularity maximization using ZCmod, the resulting communities must maintain at least one node of each node type, a rather stringent assumption. In comparison, ECoHeN is able to identify small, bipartisan and partisan communities, each with a connection density higher than competing methods.

\begin{figure}[htb!]
    \centering
    \includegraphics[width=\textwidth]{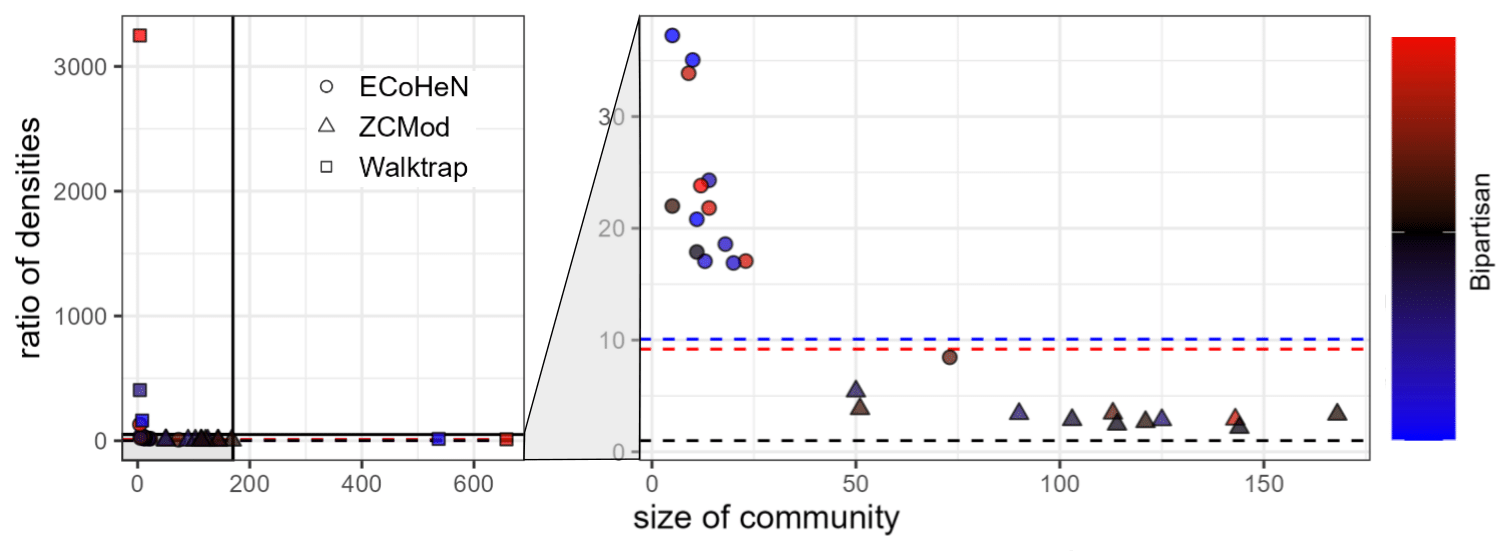}
    \caption{The ratio of density (RatD) for each of the 15 communities extracted from ECoHeN (circles), 11 communities detected from ZCmod (triangles), and five communities detected from Walktrap (squares) along with the size of the communities. Each community is colored according to the proportion of liberals (equivalently, conservatives) in the community where bluer (redder) points are largely composed of liberals (conservatives). Blacker points illustrate bipartisan communities. The black horizontal line illustrates the expected RatD in a random network. The blue (red) horizontal line illustrates the RatD provided all liberals (conservatives). As opposed to ZCmod, ECoHeN can identify both partisan and bipartisan communities. The largest bipartisan community found by ECoHeN features a RatD nearly 1.6 times larger than the largest ratio of densities attained via ZCmod. Walktrap partitions the network into at least two large communities which align with political ideology, a trivial partition which provides little insight about the connections between liberals and
conservatives.}
    \label{fig:ratio_vs_size}
\end{figure}

Both ECoHeN and ZCmod result in communities whose links are relatively more dense internally than to the rest of the network; however, the ratio of densities observed from ECoHeN communities are much higher than those communities from ZCmod. The larger RatD is partly because 1) the ECoHeN communities tend to be smaller than the ZCmod communities and 2) the fact that ZCmod is a partitioning method and must assign each node to a community, potentially diluting the density of otherwise well-connected collections of nodes. Nevertheless, the largest ECoHeN community (shown in Figure \ref{fig:largest_communities_ecohen}\textcolor{blue}{b}) has a RatD that is about 1.6 times larger than the ZCmod community with the largest RatD, both of which are largely bipartisan communities. 

This application highlights the importance of conditioning on node type when performing community discovery. Since partisan links are particularly common compared to bipartisan links, the partisan communities identified by ECoHeN are also particularly dense---denser than the natural baseline induced from taking all liberals or all conservatives, respectively. On the other hand, the RatD required by ECoHeN to consider a set of bipartisan nodes a community is naturally lower, a testament to the strengths of ECoHeN as it leverages differences in the connectivity between node types. If a less connected political party, say independents, were included in the network, ECoHeN would be uniquely positioned to identify both partisan and bipartisan communities including independents since ECoHeN identifies communities considering the density with respect to type and does not place constraints or assumptions on the political composition of each community.

\section{Discussion}
\label{sec:discussion}

ECoHeN is a generalization of an existing community extraction method called the extraction of statistically significant communities (ESSC). ECoHeN iteratively updates a candidate community by assessing the significance of connections between each node and the candidate community through a reference distribution derived under the heterogeneous degree configuration model. Like its predecessor ESSC, ECoHeN can identify background nodes and overlapping communities, two common properties of realistic networks. Compared to ECoHeN, many community discovery methods assign background nodes to otherwise tightly connected communities, reducing the overall density of the community, and assume that communities are disjoint, unrealistically positing that no node may be tightly connected to more than one collection of nodes. Unlike ESSC, ECoHeN takes advantage of differences between nodes' types and any resulting differences in the density of connections between them to identify communities. A key advantage of ECoHeN is its ability to discover communities that are topologically dense with respect to the node types of each community's members without making assumptions or imposing constraints on the resulting type composition of each community. ECoHeN is the first extraction method capable of identifying both homogeneous and heterogeneous community structure. Furthermore, ECoHeN can be parameterized such that it is guaranteed to converge, resolving issues with cycles present in ESSC's implementation.

Generalizations of ECoHeN are possible and an area for future work. One particular generalization of interest is the extension to directed, heterogeneous multigraphs which would be possible by generalizing the heterogeneous degree sequence of each node to include both in- and out-degrees. Other avenues of work include the derivation of a finite sampling distribution for the measure of connectivity $p_\mathcal{B}(u)$ as defined in \eqref{eqn:pbu_ind} and a temporal network extension. Furthermore, while ECoHeN extracts communities parallelized across the initial seed sets and is partially implemented in \texttt{C++} for efficiency, scalability to large networks is still a concern. One way to improve scalability is to reconsider how the method is initialized (e.g., define locally optimal seed sets for heterogeneous networks), and conduct tests for inclusion and exclusion locally (e.g., test only direct neighbors for inclusion in the extraction procedure). The implications of such changes on the multiple testing correction and convergence properties of the algorithm is also an area of future work.

% ACKNOWLEDGEMENTS
\bigskip
\begin{center}
{\large\bf ACKNOWLEDGEMENTS}
\end{center}

This work was inspired by an ongoing collaboration with Dr. Ryan Layer and Michael Bradshaw. We would like to thank Dr. Jingfei Zhang and Dr. Yuguo Chen for providing code which made comparisons to ZCmod possible. This work was partially supported by grant IOS-1856229 from the National Science Foundation. This work utilized the Summit supercomputer, which is supported by the National Science Foundation (awards ACI-1532235 and ACI-1532236), the University of Colorado Boulder, and Colorado State University.

% SUPPLEMENT
\bigskip
\begin{center}
{\large\bf SUPPLEMENTARY MATERIAL}
\end{center}

\begin{description}

\item[Supplement A - Heterogeneous Degree Configuration Model (HDCM): ] Detailed description and schematics of the heterogeneous degree configuration model introduced in Section \ref{sec:hdcm}. (.pdf file)

\item[Supplement B - Simulation Study: ] More details and investigations from Section \ref{sec:simulation}, including empirical results on the effects of the learning rate and decay rate on the number and quality of the communities identified by ECoHeN and ESSC. (.pdf file)

\item[Supplement C - Proofs: ] Statement and proofs of Theorem \ref{thm:asym_dist}, Corollary \ref{cor:asym_dist}, and Theorem \ref{thm:extraction_conv}. (.pdf file)

\item[Supplement D - ECoHeN Algorithm Pseudocode and Code: ] Pseudocode and code for the ECoHeN algorithm, including all extraction routines. We provide the GitHub link to the \texttt{R} package \texttt{ECoHeN}. (.pdf file)

\end{description}

% BIBLIOGRAPHY
\bibliography{references.bib}

\newpage

\appendix
\renewcommand\thefigure{\thesection.\arabic{figure}}
\setcounter{figure}{0}

\newpage

\section{Heterogeneous Degree Configuration Model (HDCM)}

Assume $\mathcal{G}=(\mathcal{V}, \mathcal{E})$ denotes an observed, heterogeneous network with the collection of heterogeneous degree sequences $\boldsymbol{\mathcal{D}}$. The ``Heterogeneous Degree Configuration Model'' (HDCM) section of the manuscript provides a generalization of a degree configuration model (DCM) capable of preserving not only the degree of each node but the heterogeneous degree sequence of each node. Using the corresponding notation and assumptions provided in the ``Heterogeneous Networks'' and ``Heterogeneous Degree Configuration Model'' sections of the manuscript, one can efficiently conduct the wiring process of the HDCM through a through a constrained permutation of the node labels in the edge multiset $\mathcal{E}$.

\paragraph{Efficient Generative Process for the HDCM}

An efficient, generative process for HDCM$(\mathcal{T}, \boldsymbol{\mathcal{D}})$ is depicted in Figure \ref{fig:hdcm_flow} and takes advantage of the existing edge set $\mathcal{E}$. The process begins by partitioning the edge multiset into $K + K(K-1)/2$ subsets according to the observed adjacent nodes' type: $\mathcal{E}=\bigcup_{1\leq k\leq l\leq K}E^{[kl]}$ where $$E^{[kl]}=\left\{\{u, v\}\in\mathcal{E} : u\text{ is of type } k, v \text{ is of type } l \right\}.$$
That is, the edge multiset $E^{[kl]}\subseteq \mathcal{E}$ is the set of all $m^{[kl]}=|E^{[kl]}|$ undirected links between nodes of type $k$ and $l$. To construct a network with the same collection of heterogeneous degree sequences, we permute the node labels of $E^{[kl]}$ to construct a new edge multiset, denoted $\widetilde{E}^{[kl]}$, for each $1\leq k\leq l\leq K$. The process for attaining $\widetilde{E}^{[kl]}$ is dependent on whether $k=l$ or $k\neq l$, so we consider these cases.

\begin{figure}[!htbp]
    \centering
    \includegraphics[scale=0.25]{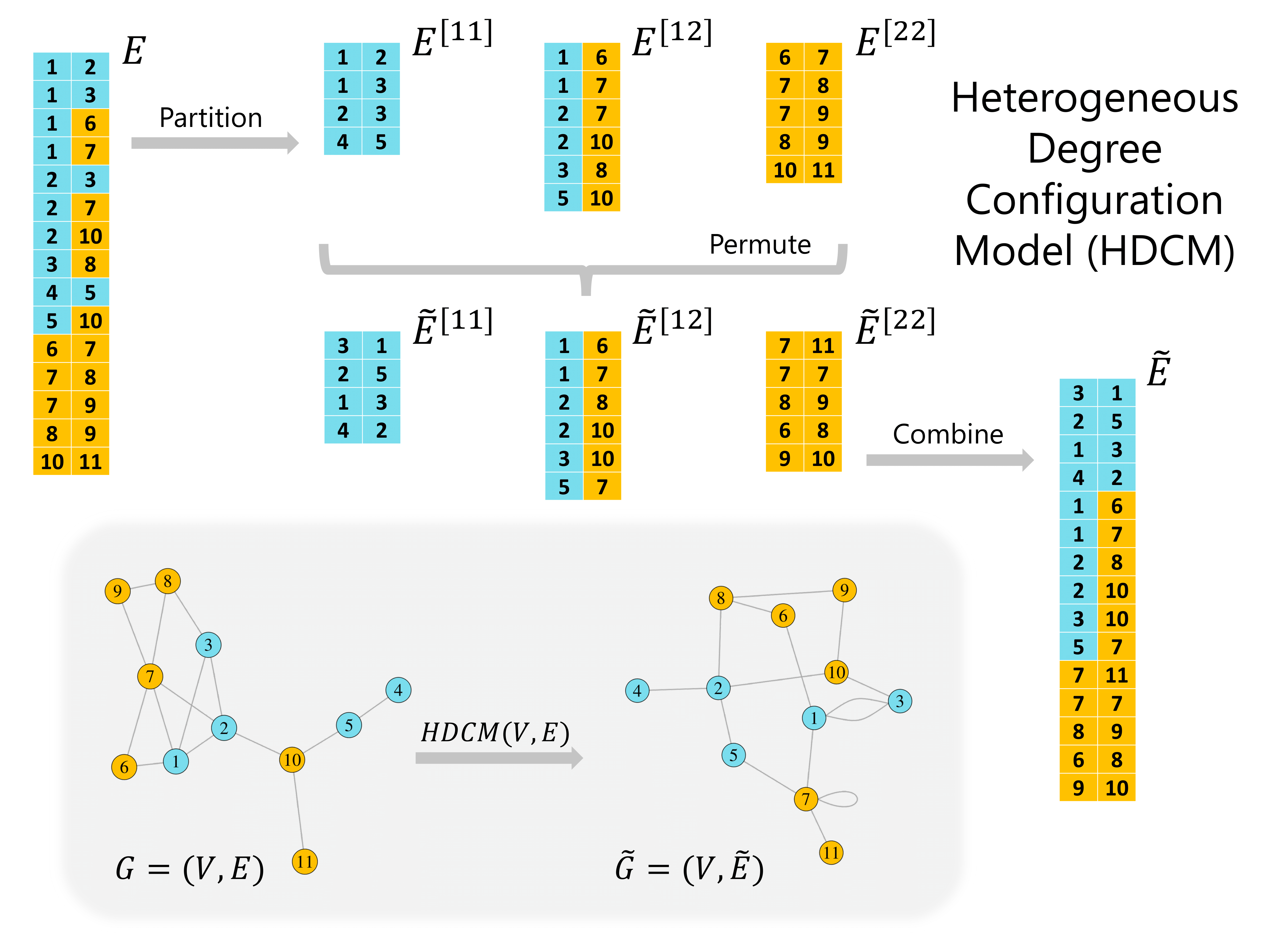}
    \caption{A schematic of the heterogeneous degree configuration model: a model of a random network maintaining the same collection of heterogeneous degree sequences, $\boldsymbol{\mathcal{D}}$, as an observed network, $\mathcal{G}$. The process begins by partitioning the edge multiset, $\mathcal{E}$, according to adjacent nodes' type. Upon completion, node labels are permuted (or rearranged) within the elements of the partition to preserve the respective type $k$ degree of each node. The new edge multisets are combined into one edge multiset, $\widetilde{\mathcal{E}}$, where a random network $\widetilde{\mathcal{G}}=(\mathcal{V}, \widetilde{\mathcal{E}})$ has the same degree collection $\boldsymbol{\mathcal{D}}$.}
    \label{fig:hdcm_flow}
\end{figure}

\textbf{Case 1: Attaining} $\boldsymbol{\widetilde{E}^{[kl]}}$\textbf{ for }$\boldsymbol{k=l}.$ For same type connections, the idea is to permute all node labels and create adjacent pairs by grouping the reordered node labels two by two. This process will preserve the number of same type connections for each node. To detail formally, let $E^{[kk]}$ be structured as an $m^{[kk]}$ by $2$ matrix where each row vector provides an observed edge between two type $k$ nodes. Note, the order of the rows and columns of $E^{[kk]}$ carry no specific meaning but is used for bookkeeping. Furthermore, if $\mathcal{G}$ is a multigraph then there exist duplicate rows. Let $\boldsymbol{w} = \left[w_1, \dots, w_{2m^{[kk]}} \right]$ denote a row vector (with $2m^{[kk]}$ entries) attained by concatenating the transpose of the column vectors of $E^{[kk]}$. Let $\tilde{\boldsymbol{w}}$ denote a permutation (or rearrangement) of the entries of $\boldsymbol{w}$. The random configuration of $k$ to $k$ connections is represented similarly as an $m^{[kk]}$ by 2 matrix, denoted $\widetilde{E}^{[kl]}$, where the $i$th row vector of $\widetilde{E}^{[kl]}$ is given by $\left[\tilde{w}_i, \tilde{w}_{i+1}\right]$ for $i\in\{1, \dots, m^{[kk]}\}$.

\textbf{Case 2: Attaining} $\boldsymbol{\widetilde{E}^{[kl]}}$\textbf{ for }$\boldsymbol{k< l}.$ For between type connections, the idea is to isolate and permute node labels of one type, say $l$, before creating adjacent pairs by joining the permuted node labels with those node labels of type $k$. To detail formally, let $E^{[kl]}$ be structured as an $m^{[kl]}$ by $2$ matrix where each row vector provides an edge between a type $k$ node and a type $l$ node, respectively. That is, the first column vector of $E^{[kl]}$, denoted $\boldsymbol{z}$, contains only type $k$ nodes and the second column vector of $E^{[kl]}$, denoted $\boldsymbol{w}$, contains only type $l$ nodes. Note, the order of the rows of $E^{[kl]}$, however, carry no specific meaning and is used for bookkeeping. Furthermore, if $\mathcal{G}$ is a multigraph then there exist duplicate rows. Let $\tilde{\boldsymbol{w}}$ denote a permutation (or rearrangement) of the entries of $\boldsymbol{w}$. The random configuration of $k$ to $l$ connections is represented similarly as an $m^{[kl]}$ by 2 matrix, denoted $\widetilde{E}^{[kl]}$, where the first and second columns of $\widetilde{E}^{[kl]}$ are $\boldsymbol{z}$ and $\tilde{\boldsymbol{w}}$, respectively.

The newly attained edge multisets $\widetilde{E}^{[kl]}$ (represented as $m^{[kl]}$ by 2 matrices) for each $1\leq k\leq l \leq K$ are combined rowwise (in any order) to attain an $m$ by 2 matrix, denoted $\widetilde{\mathcal{E}}$, where $m\vcentcolon=|\widetilde{\mathcal{E}}|=|\mathcal{E}|$. The random network $\widetilde{G}=(\mathcal{V}, \widetilde{\mathcal{E}})$ has the same collection of heterogeneous degree sequences as the observed network $\mathcal{G}$. That is, each node in $\widetilde{\mathcal{G}}$ has the same type $k$ degree for all $k\in\{1, \dots, K\}$ as in $\mathcal{G}$. The general process for attaining the random edge multiset, $\widetilde{\mathcal{E}}$, from the observed edge multiset, $\mathcal{E}$, is presented in Figure \ref{fig:hdcm_flow} where the permutation of node labels is depicted in Figure \ref{fig:perm_flow}. Notice, the process for permuting node labels differs for same type connections (e.g., $E^{[11]}$ and $E^{[22]}$) compared to between type connections (e.g., $E^{[12]}$) as described in cases one and two.

\begin{figure}[!htbp]
    \centering
    \includegraphics[scale=0.25]{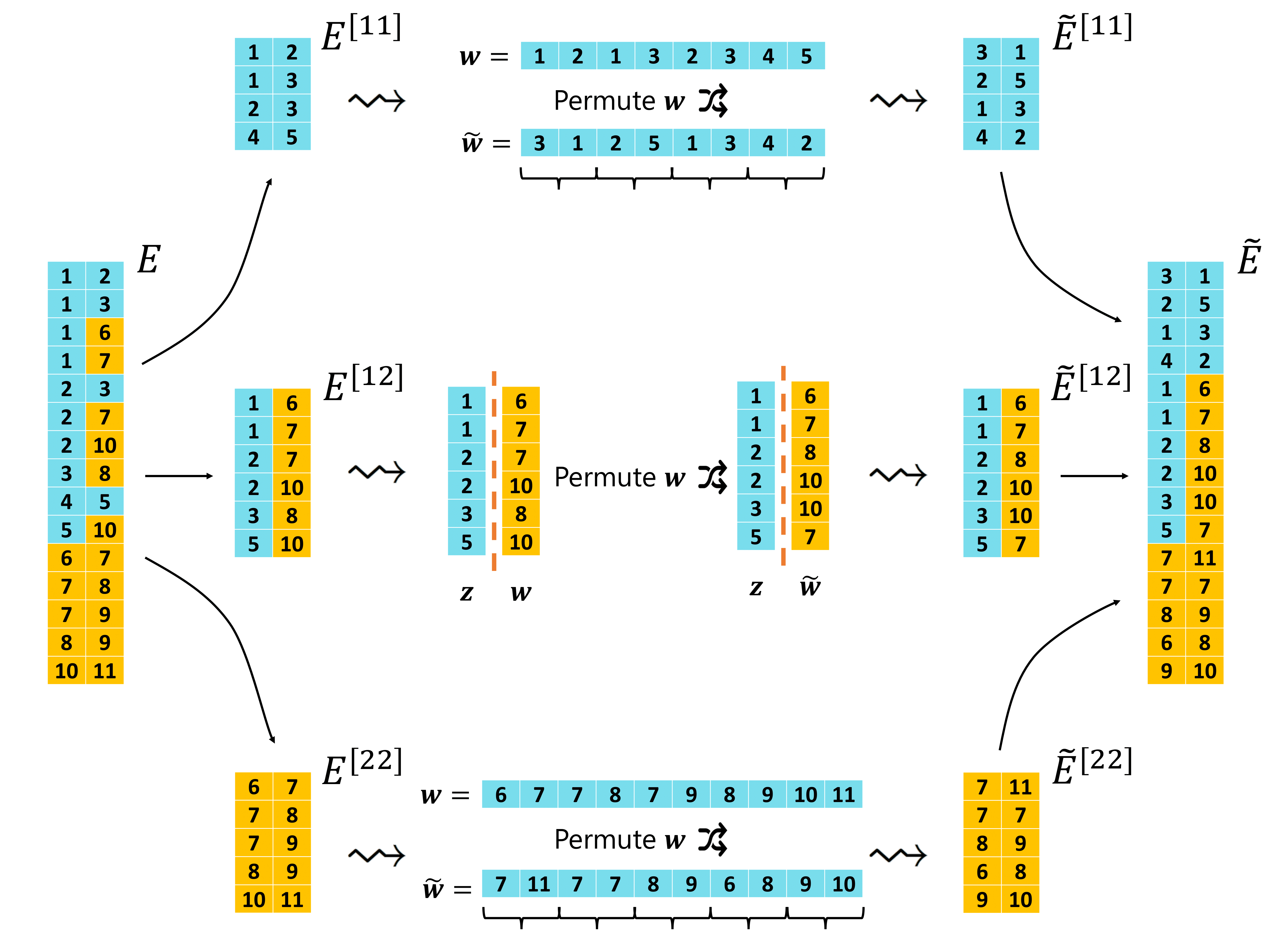}
    \caption{A schematic describing the random permutation of node labels for each element of the partitioned edge multiset. The process for permuting node labels is different for same type connections compared to between type connections. For same type connections (i.e., those described by $E^{[11]}$ and $E^{[22]}$), all node labels as they appear in $E^{[kk]}$ are permuted and paired two by two to create a new edge multiset $\widetilde{E}^{[kk]}$. For between type connections (i.e., those described by $E^{[12]}$), only node labels of type $l$ as they appear in $E^{[kl]}$ are permuted before being subsequently paired to the type $k$ node labels to create the new edge multiset $\widetilde{E}^{[kl]}$. The resulting edge multisets (represented as matrices) are combined rowwise.}
    \label{fig:perm_flow}
\end{figure}

\paragraph{Generalization of a Degree Configuration Model}

The heterogeneous degree configuration model is a generalization of a degree configuration model. A formal definition of a degree configuration model (DCM) for a homogeneous network is provided in \citet[Chapter~13.2]{newman2018networks}; however, \citet{fosdick2018configuring} provides expansive coverage on the topic. To clarify the similarities and differences between $\text{HDCM}(\mathcal{T}, \boldsymbol{\mathcal{D}})$ and $\text{DCM}(\mathcal{D})$, we will consider a random graph, $\tilde{\mathcal{G}}$, generated from the $\text{HDCM}(\mathcal{T}, \boldsymbol{\mathcal{D}})$ for $K=1$ and $K>1$. For a general $K$, the heterogeneous degree sequence (introduced in the Heterogeneous Network section of the manuscript) of node $u$ in $\tilde{\mathcal{G}}$ is equal to the observed heterogeneous degree sequence of node $u$ in $\mathcal{G}$ for all $u\in\mathcal{V}$. That is, node $u$ is connected to the same number of type $k$ nodes in $\tilde{\mathcal{G}}$ as in $\mathcal{G}$ for all $k\in\{1, \dots, K\}$. As a result, node $u$ is connected to the same number of nodes in $\tilde{\mathcal{G}}$ as in $\mathcal{G}$. Thus, the sample space of the $\text{HDCM}(\mathcal{T}, \boldsymbol{\mathcal{D}})$ is a subspace of the sample space of the $\text{DCM}(\mathcal{D})$. The two sample spaces are equal when $K=1$ (i.e., the network is homogeneous) since the heterogeneous degree sequence of a node simplifies to a singleton set with the degree of a node. When $K>1$, it is easy to construct an example where the degree of each node is the same as in the observed graph, but the heterogeneous degree sequences are different. When $K>1$ the sample space of the $\text{HDCM}(\mathcal{T}, \boldsymbol{\mathcal{D}})$ is a strict subset of the $\text{DCM}(\mathcal{D})$. This illustrates how the heterogeneous degree configuration model is simply a generalization of the degree configuration model used in \citet{wilson2014testing}.

\section{Proof of Theorems}

The statement of Theorem 3.1, Corollary 3.1, and Theorem 3.2, along with corresponding assumptions and notation, are provided in the manuscript.

\paragraph{Asymptotic Distribution of $X_n^{[k]}(u_n^{[l]}:\mathcal{B}_n)$}

We address the asymptotic behavior of a random quantity of interest.

\begin{proof}[Proof of Theorem \ref{thm:asym_dist}]
From Assumption \ref{eqn:aconvmean}, we have
\begin{align}
\label{eq:datamean_conv}
\mu_{l,n}^{[k]} = \int_{\mathbb{R}^{+}}{x\,dF_{l,n}^{[k]}(x)} = 
\sum_{t=0}^{\infty}{t\frac{N_{l,n}^{[k]}(t)}{|V_n^{[l]}|}} =
\frac{2^{\mathbb{I}(k=l)}|E_{n}^{[kl]}|}{|V_n^{[l]}|} \to
\mu_{l}^{[k]}
\end{align}
as $n\to\infty$ where
$N_{l,n}^{[k]}(t)$ is the number of type $l$ nodes with a type $k$ degree of $t$, and $E_{n}^{[kl]}$ is the subset multiset of the edge multiset which contains links between type $k$ nodes and type $l$ nodes. Equation \eqref{eq:datamean_conv} implies
\begin{align*}
1 = \lim_{n\to\infty}\frac{2^{\mathbb{I}(k=l)}|E_n^{[kl]}|}{\mu_{l}^{[k]}|V_n^{[l]}|} \quad (\text{i.e., } 2^{\mathbb{I}(k=l)}|E_n^{[kl]}| \sim \mu_{l}^{[k]}|V_n^{[l]}|)
\end{align*}
where $\mu_{l}^{[k]}< \infty$ by Assumption \ref{eqn:nonzero_ktypedeg_mean}.
By Assumption \ref{eqn:lttype_prop}, we have
\begin{align*}
\frac{V_n^{[l]}}{\mathcal{V}_n} = \frac{V_n^{[l]}}{n} \to \gamma^{[l]}\in(0, 1) \text{ as } n\to\infty \implies 1 = \lim_{n\to\infty}\frac{|V_n^{[l]}|}{\gamma^{[l]}n} \quad (\text{i.e., } |V_n^{[l]}|\sim \gamma^{[l]}n).
\end{align*}
Hence, we have $2^{\mathbb{I}(k=l)}|E_n^{[kl]}|\sim \mu_l^{[k]}\gamma^{[l]}n$. Note, when $\mu_{l}^{[k]}=0$ then $X_n^{[k]}(u_n^{[l]}, \mathcal{B}_n)=0,Y_n\sim\text{Binom}(c, 0)$ and $d_{TV}\left(X_n^{[k]}(u_n^{[l]}, \mathcal{B}_n), Y_n \right)\to 0$ as $n\to\infty$, trivially. We proceed assuming $\mu_{l}^{[k]}>0$ before addressing this edge case formally.

We assume (WLOG) that $k\leq l$. We wish to understand the distribution of $X_{n}^{[k]}(u_n^{[l]} : \mathcal{B}_n)$: the random number of type $k$ nodes in $\mathcal{B}_n$ adjacent to $u_n^{[l]}$ in $\mathcal{H}_n$ constructed via the heterogeneous degree configuration model. We make use of the procedure for sampling a network from the heterogeneous degree configuration model, denoted $\text{HDCM}(\mathcal{V}_n, \mathcal{E}_n)$, and discussed in Supplement A of the manuscript. Under the $\text{HDCM}(\mathcal{V}_n, \mathcal{E}_n)$, the edge multiset $\mathcal{E}_n=\cup_{1\leq i\leq j \leq K} E_n^{[ij]}$ is partitioned according to the adjacent nodes' types (where $E_n^{[ij]}$ contains the links between nodes of type $i$ and $j$), the edges in $E_n^{[ij]}$ are randomly rearranged to preserve degree (resulting in $\tilde{E}^{[ij]}_n$ for all $1\leq i \leq j \leq K$), and the resulting $\tilde{E}^{[ij]}_n$ are combined into one edge set (denoted $\tilde{\mathcal{E}}_n$). Each node in a random network $\mathcal{H}_n = (\mathcal{V}_n, \tilde{\mathcal{E}}_n)$ maintains the same heterogeneous degree sequence as observed in $\mathcal{G}_n = (\mathcal{V}_n, \mathcal{E}_n)$, implying there are precisely $c$ edges incident to $u_n^{[l]}$ also incident a type $k$ node in $\mathcal{H}_n$. Let $\tilde{v}_{ni}^{[k]}$ denote the $i^{\text{th}}$ type $k$ node where $\{\tilde{v}_{ni}^{[k]}, u_n^{[l]}\}\in \tilde{E}^{[kl]}$ for $i=1,\dots, c$. Note that the ordering carries no specific meaning but is used for bookkeeping. Self-loops and multi-edges complicate the interpretation and discussion surrounding $\tilde{v}_{ni}^{[k]}$. We refer to $\tilde{v}_{ni}^{[k]}$ simply as the $i$th type $k$ node adjacent to $u_n^{[l]}$; however, we recognize the elements of $\{\tilde{v}_{ni}^{[k]}\}_{i=1}^{c}$ are not necessarily unique and count, for example, a node as twice adjacent to $u_n^{[l]}$ should there exist two edges between them. Furthermore, if $\tilde{v}_{ni}^{[k]}=u_n^{[l]}$ then $\{\tilde{v}_{ni}^{[k]}, u_n^{[l]}\}$ would constitute a self-loop; a self-loop increases the degree of a node by two, so we would say then that $u_n^{[l]}$ is twice adjacent to itself. In a simple setting (no self-loops and multi-edges), there is no distinction between counting the nodes adjacent to a node $w$ and counting the nodes incident to an edge which is incident to a node $w$, so we treat them as the same objective here.

To understand the distribution of $X_{n}^{[k]}(u_n^{[l]} : \mathcal{B}_n)$, we focus on the adjacent type $k$ nodes, $\{\tilde{v}_{ni}^{[k]}\}_{i=1}^{c}$, and reveal whether each $\tilde{v}_{ni}^{[k]}$ is a member of $\mathcal{B}_n$ (equivalently, a member of $B_n^{[k]}$). For $i\in\{1, \dots, c\}$, let $A_i$ denote a binary random variable indicating whether $\tilde{v}^{[k]}_{ni}$ is a member of $\mathcal{B}_n$. That is, $A_i=1$ if the $i$th type $k$ node adjacent to $u_n^{[l]}$ in $\mathcal{H}_n$ is also in $\mathcal{B}_n$ and 0 otherwise. Note that $X_{n}^{[k]}(u_n^{[l]} : \mathcal{B}_n)=\sum_{i=1}^{c} A_i$. Let $r_{l,i}^{[k]}(\mathcal{B}_n)$ denote the conditional probability of $A_i=1$ conditional on the previous $i-1$ revelations: $\{A_1, \dots, A_{i-1}\}$. 

Consider $r_{l,1}^{[k]}(\mathcal{B}_n)$. In this case, no type $k$ nodes adjacent to $u_n^{[l]}$ have been revealed, so $r_{l,1}^{[k]}(\mathcal{B}_n)$ is the ratio of the number of type $k$ nodes in $\mathcal{B}_n$ adjacent to type $l$ nodes to the total number of type $k$ nodes adjacent to type $l$ nodes:
\begin{align*}
    r_{l,1}^{[k]}(\mathcal{B}_n) = 
    \frac{\left(\sum_{w\in B^{[k]}_n} d^{[l]}(w)\right) - \mathbb{I}(k=l)\mathbb{I}(u_n^{[l]}\in\mathcal{B}_n)}
    {\left(\sum_{z\in V^{[k]}_n} d^{[l]}(z)\right) -\mathbb{I}(k=l)} =
    \frac{\left(\sum_{w\in B^{[k]}_n} d^{[l]}(w)\right) - \mathbb{I}(k=l)\mathbb{I}(u_n^{[l]}\in\mathcal{B}_n)}
    {2^{\mathbb{I}(k=l)}|E_n^{[kl]}| -\mathbb{I}(k=l)}.
\end{align*}
The indicators in the numerator ensure we do not count $u_n^{[l]}$ itself when counting the number of type $k$ nodes in $\mathcal{B}_n$ which could be adjacent to $u_n^{[l]}$ (only a concern when $k=l$ and $u_n^{[l]}\in\mathcal{B}_n$). Similarly, the indicator in the denominator ensures we do not count $u_n^{[l]}$ itself when counting the number of type $k$ nodes in $\mathcal{V}_n$ which could be adjacent to $u_n^{[l]}$ (only a concern when $k=l$). Now, consider $r_{l,2}^{[k]}(\mathcal{B}_n)$. One type $k$ node adjacent to $u_n^{[l]}$ has been revealed. Hence, there is one fewer type $k$ node which can possibly serve as $\tilde{v}^{[k]}_{n2}$. If $\tilde{v}^{[k]}_{n1}$ was revealed to be a member of $\mathcal{B}_n$ (i.e., $A_1=1$), then there is also one fewer type $k$ node in $\mathcal{B}_n$ which can possibly serve as $\tilde{v}^{[k]}_{n2}$. Otherwise, if $A_1=0$, then the number of type $k$ nodes in $\mathcal{B}_n$ which can possibly serve as $\tilde{v}^{[k]}_{n2}$ is unaffected; only the overall count is affected. Thus,
\begin{align*}
    r_{l,2}^{[k]}(\mathcal{B}_n) \in
    \left[
    \frac{\left(\sum_{w\in B^{[k]}_n} d^{[l]}(w)\right) - \mathbb{I}(k=l)\mathbb{I}(u_n^{[l]}\in\mathcal{B}_n)-1}
    {2^{\mathbb{I}(k=l)}|E_n^{[kl]}| -\mathbb{I}(k=l)-1},
    \frac{\left(\sum_{w\in B^{[k]}_n} d^{[l]}(w)\right) - \mathbb{I}(k=l)\mathbb{I}(u_n^{[l]}\in\mathcal{B}_n)}
    {2^{\mathbb{I}(k=l)}|E_n^{[kl]}| -\mathbb{I}(k=l)-1}
    \right]
\end{align*}
where the lowerbound arises if $A_1=1$ and the upperbound arises if $A_1=0$. Arguing analogously for $1\leq i \leq c$, we have that $r_{l,i}^{[k]}(\mathcal{B}_n)$ is bounded uniformly on all prior $i-1$ revelations by
\begin{align*}
    r_{l,i}^{[k]}(\mathcal{B}_n) \in
    \left[
    \frac{\left(\sum_{w\in B^{[k]}_n} d^{[l]}(w)\right) - \mathbb{I}(k=l)\mathbb{I}(u_n^{[l]}\in\mathcal{B}_n)-(i-1)}
    {2^{\mathbb{I}(k=l)}|E_n^{[kl]}| -\mathbb{I}(k=l)-(i-1)},
    \frac{\left(\sum_{w\in B^{[k]}_n} d^{[l]}(w)\right) - \mathbb{I}(k=l)\mathbb{I}(u_n^{[l]}\in\mathcal{B}_n)}
    {2^{\mathbb{I}(k=l)}|E_n^{[kl]}| -\mathbb{I}(k=l)-(i-1)}
    \right]
\end{align*}
where the lowerbound arises if $A_j=1$ for all $j\in\{1, \dots, i-1\}$, and the upperbound arises if $A_j=0$ for all $j\in\{1, \dots, i-1\}$. Recall from Equation \eqref{eqn:psuccess}
that
\begin{align*}
    p_{l,n}^{[k]}(\mathcal{B}_n)=\frac{\sum_{w\in B_n^{[k]}}d^{[l]}(w)}{\sum_{z\in V_n^{[k]}}d^{[l]}(z)}.
\end{align*}
We will show that $\sup_{1\leq i \leq c} |r_{l,i}^{[k]}(\mathcal{B}_n) - p_{l,n}^{[k]}(\mathcal{B}_n)|\to 0$ as $n\to\infty$ by considering the case when $k=l$ and $k\neq l$.

%%%%%%%%%%%%%%%%%%%%%%%%%%%%%%%%%%%%%%%%%%%%%%%%%%%%%%

\item{\underline{\textit{Case 1}: $k = l$}} 

\noindent If $k=l$, then
\begin{align*}
    r_{l,i}^{[l]}(\mathcal{B}_n) \in
    \left[
    \frac{\left(\sum_{w\in B^{[l]}_n} d^{[l]}(w)\right) - \mathbb{I}(u_n^{[l]}\in\mathcal{B}_n)-(i-1)}
    {2|E_n^{[ll]}| - i},
    \frac{\left(\sum_{w\in B^{[l]}_n} d^{[l]}(w)\right) - \mathbb{I}(u_n^{[l]}\in\mathcal{B}_n)}
    {2|E_n^{[ll]}| - i}
    \right]
\end{align*}
and
\begin{align*}
    p_{l,n}^{[l]}(\mathcal{B}_n)=
    \frac{\sum_{w\in B_n^{[l]}}d^{[l]}(w)}
    {\sum_{z\in V_n^{[l]}}d^{[l]}(z)} =
    \frac{\sum_{w\in B_n^{[l]}}d^{[l]}(w)}
    {2|E_n^{[ll]}|}.
\end{align*}
Hence,
\begin{align*}
    r_{l,i}^{[l]}(\mathcal{B}_n) - p_{l,n}^{[l]}(\mathcal{B}_n) &\geq
    \frac{\left(\sum_{w\in B^{[l]}_n} d^{[l]}(w)\right) - \mathbb{I}(u_n^{[l]}\in\mathcal{B}_n)-(i-1)}
    {2|E_n^{[ll]}| - i} - 
        \frac{\sum_{w\in B_n^{[l]}}d^{[l]}(w)}
    {2|E_n^{[ll]}|} \\ &=
    \frac{i\sum_{w\in B^{[l]}_n} d^{[l]}(w) - 2i|E_n^{[ll]}|-2\mathbb{I}(u_n^{[l]}\in\mathcal{B}_n)|E_n^{[ll]}| + 2|E_n^{[ll]}|}{2|E_n^{[ll]}|(2|E_n^{[ll]}| - i)} \\ &\geq
    \frac{i\sum_{w\in B^{[l]}_n} d^{[l]}(w) - 2i|E_n^{[ll]}|}{2|E_n^{[ll]}|(2|E_n^{[ll]}| - i)} \text{ taking } \mathbb{I}(u^{[l]}_n\in\mathcal{B}_n) = 1 \\ &\geq
    \frac{-i}{2|E_n^{[ll]}| - i} \text{ since } \sum_{w\in B^{[l]}_n} d^{[l]}(w) \geq 0,
\end{align*}
and
\begin{align*}
    r_{l,i}^{[l]}(\mathcal{B}_n) - p_{l,n}^{[l]}(\mathcal{B}_n) &\leq
    \frac{\left(\sum_{w\in B^{[l]}_n} d^{[l]}(w)\right) - \mathbb{I}(u_n^{[l]}\in\mathcal{B}_n)}
    {2|E_n^{[ll]}| - i} - 
        \frac{\sum_{w\in B_n^{[l]}}d^{[l]}(w)}
    {2|E_n^{[ll]}|} \\ &=
    \frac{i\sum_{w\in B^{[l]}_n} d^{[l]}(w) - 2\mathbb{I}(u_n^{[l]}\in\mathcal{B}_n)|E_n^{[ll]}|}{2|E_n^{[ll]}|(2|E_n^{[ll]}| - i)} \\ &\leq
    \frac{i\sum_{w\in B^{[l]}_n} d^{[l]}(w)}{2|E_n^{[ll]}|(2|E_n^{[ll]}| - i)} \text{ taking } \mathbb{I}(u^{[l]}_n\in\mathcal{B}_n) = 0 \\ &\leq
    \frac{i}{2|E_n^{[ll]}|-i}  \text{ since } \sum_{w\in B^{[l]}_n} d^{[l]}(w) \leq 2|E_n^{[ll]}|.
\end{align*}
Thus,
\begin{align*}
     |r_{l,i}^{[l]}(\mathcal{B}_n) - p_{l,n}^{[l]}(\mathcal{B}_n)| \leq \frac{i}{2|E_n^{[ll]}|-i} \leq
     \frac{c}{2|E_n^{[ll]}|-i} \leq
     \frac{c}{2|E_n^{[ll]}|-c},
\end{align*}
implying
\begin{align}
\label{eq:sup0_keql}
    \sup_{1\leq i \leq c} |r_{l,i}^{[l]}(\mathcal{B}_n) - p_{l,n}^{[l]}(\mathcal{B}_n)| \leq \frac{c}{2|E_n^{[ll]}|-c} \to 0
\end{align}
as $n\to\infty$ since $2|E_n^{[ll]}|\sim \mu_l^{[l]}\gamma^{[l]}n$ and $c$ is fixed.

%%%%%%%%%%%%%%%%%%%%%%%%%%%%%%%%%%%%%%%%%%%%%%%%%%%%%%

\item{\underline{\textit{Case 2}: $k\neq l$}}

\noindent If $k\neq l$, then
\begin{align*}
    r_{l,i}^{[k]}(\mathcal{B}_n) \in
    \left[
    \frac{\left(\sum_{w\in B^{[k]}_n} d^{[l]}(w)\right) -(i-1)}
    {|E_n^{[kl]}| - (i-1)},
    \frac{\left(\sum_{w\in B^{[k]}_n} d^{[l]}(w)\right)}
    {|E_n^{[kl]}| - (i-1)}
    \right]
\end{align*}
and
\begin{align*}
    p_{l,n}^{[k]}(\mathcal{B}_n)=
    \frac{\sum_{w\in B_n^{[k]}}d^{[l]}(w)}
    {\sum_{z\in V_n^{[k]}}d^{[l]}(z)} =
    \frac{\sum_{w\in B_n^{[k]}}d^{[l]}(w)}
    {|E_n^{[kl]}|}.
\end{align*}
Hence,
\begin{align*}
    r_{l,i}^{[k]}(\mathcal{B}_n) - p_{l,n}^{[k]}(\mathcal{B}_n) &\geq
    \frac{\left(\sum_{w\in B^{[k]}_n} d^{[l]}(w)\right)-(i-1)}
    {|E_n^{[kl]}| - (i-1)} - 
        \frac{\sum_{w\in B_n^{[k]}}d^{[l]}(w)}
    {|E_n^{[kl]}|} \\ &=
    \frac{(i-1)\left[\sum_{w\in B^{[k]}_n} d^{[l]}(w) - |E_n^{[kl]}|\right]}{|E_n^{[kl]}|(|E_n^{[kl]}| - (i-1))} \\ &\geq
    \frac{-(i-1)|E_n^{[kl]}|}{|E_n^{[kl]}|(|E_n^{[kl]}| - (i-1))} \text{ since } \sum_{w\in B^{[k]}_n} d^{[l]}(w) \geq 0  \\ &=
    \frac{-(i-1)}{|E_n^{[kl]}| - (i-1)},
\end{align*}
and
\begin{align*}
    r_{l,i}^{[k]}(\mathcal{B}_n) - p_{l,n}^{[k]}(\mathcal{B}_n) &\leq
    \frac{\sum_{w\in B^{[k]}_n} d^{[l]}(w)}
    {|E_n^{[kl]}| - (i-1)} - 
        \frac{\sum_{w\in B_n^{[k]}}d^{[l]}(w)}
    {|E_n^{[kl]}|} \\ &=
    \frac{(i-1)\sum_{w\in B^{[k]}_n} d^{[l]}(w)}{|E_n^{[kl]}|(|E_n^{[kl]}| - (i-1))} \\ &\leq
    \frac{(i-1)|E_n^{[kl]}|}{|E_n^{[kl]}|(|E_n^{[kl]}| - (i-1))} \text{ since } \sum_{w\in B^{[k]}_n} d^{[l]}(w) \leq |E^{[kl]}|  \\ &=
    \frac{i-1}{|E_n^{[kl]}|-(i-1)}.
\end{align*}
Thus,
\begin{align*}
     |r_{l,i}^{[k]}(\mathcal{B}_n) - p_{l,n}^{[k]}(\mathcal{B}_n)| \leq \frac{i-1}{|E_n^{[kl]}|-(i-1)} \leq
     \frac{c-1}{|E_n^{[kl]}|-(i-1)} \leq
     \frac{c-1}{|E_n^{[kl]}|-(c-1)},
\end{align*}
implying
\begin{align}
\label{eq:sup0_keql}
    \sup_{1\leq i \leq c} |r_{l,i}^{[k]}(\mathcal{B}_n) - p_{l,n}^{[k]}(\mathcal{B}_n)| \leq \frac{c-1}{|E_n^{[kl]}|-(c-1)} \to 0
\end{align}
as $n\to\infty$ since $|E_n^{[kl]}|\sim \mu_l^{[k]}\gamma^{[l]}n$ and $c$ is fixed.

%%%%%%%%%%%%%%%%%%%%%%%%%%%%%%%%%%%%%%%%%%%%%%%%%%%%%%

Thus, $\sup_{1\leq i \leq c} |r_{l,i}^{[k]}(\mathcal{B}_n) - p_{l,n}^{[k]}(\mathcal{B}_n)|\to 0$ as $n\to\infty$, implying $d_{TV}\left(X_n^{[k]}(u_n^{[l]}, \mathcal{B}_n), Y_n \right)\to 0$ as $n\to\infty$ where $Y_n\sim \text{Binom}(c, p_{l,n}^{[k]}(\mathcal{B}_n))$.

%%%%%%%%%%%%%%%%%%%%%%%%%%%%%%%%%%%%%%%%%%%%%%%%%%%%%%

We now address the trivial cases when $\mu_{l}^{[k]}=0$. Let $Y_n\sim\text{Binom}(c, 0)$, implying $P(Y_n = 0) = 1$. Suppose $\mu_{l}^{[k]}=0$, and note that $X_n^{[k]}(u_n^{[l]}, \mathcal{B}_n)$ is a non-negative random variable. Then, for all $\epsilon>0$
\begin{align*}
    \mathbb{P}(X_n^{[k]}(u_n^{[l]}, \mathcal{B}_n) > \epsilon) \leq 
    \frac{\mathbb{E}(X_n^{[k]}(u_n^{[l]}, \mathcal{B}_n))}{\epsilon} =
    \frac{\mu_{l,n}^{[k]}}{\epsilon} \to 
    \frac{\mu_{l}^{[k]}}{\epsilon} = 0
\end{align*} 
as $n\to\infty$ by Markov's inequality. Hence, $d_{TV}\left(X_n^{[k]}(u_n^{[l]}, \mathcal{B}_n), Y_n \right)\to 0$ as $n\to\infty$.
\end{proof}

\begin{proof}[Proof of Corollary \ref{cor:asym_dist}]
To demonstrate that $d_{TV}\left(X_n^{[k]}(u_n^{[l]}, \mathcal{B}_n), Y_{l,n}^{[k]}(u_n^{[l]}, \mathcal{B}_n) \right)\to 0$ as $n\to\infty$, it suffices to show that $\sup_{1\leq i \leq c} |r_{l,i}^{[k]}(\mathcal{B}_n) - p_{l,n}^{[k]}(u_n^{[l]}, \mathcal{B}_n)|\to 0$ as $n\to\infty$ when $k=l$ and $0<\mu_{l}^{[k]}<\infty$. Using a similar proof strategy, we have
\begin{align*}
    r_{l,i}^{[k]}(\mathcal{B}_n) - p_{l,n}^{[k]}(u_n^{[l]}, \mathcal{B}_n) &\geq 
    -\frac{2|E_n^{[kl]}|(c-1)-c(i-1)}{(2|E_n^{[kl]}| - i)(2|E_n^{[kl]}| - c)},
\end{align*}
and
\begin{align*}
    r_{l,i}^{[k]}(\mathcal{B}_n) - p_{l,n}^{[k]}(u_n^{[l]}, \mathcal{B}_n) &\leq 
    \frac{2|E_n^{[kl]}|(c-1)-c(i-1)}{(2|E_n^{[kl]}| - i)(2|E_n^{[kl]}| - c)}.
\end{align*}
Thus,
\begin{align}
     |r_{l,i}^{[k]}(\mathcal{B}_n) - p_{l,n}^{[k]}(u_n^{[l]}, \mathcal{B}_n)| &\leq \nonumber
     \frac{2|E_n^{[kl]}|(c-1)-c(i-1)}{(2|E_n^{[kl]}| - i)(2|E_n^{[kl]}| - c)} \\ &\leq \label{eqn:sup_ineq}
     \frac{(2|E_n^{[kl]}| - i)(c+1)}{(2|E_n^{[kl]}| - i)(2|E_n^{[kl]}| - c)} \\ &= \nonumber
     \frac{c+1}{2|E_n^{[kl]}|-c}
\end{align}
where \eqref{eqn:sup_ineq} holds since $c\leq|E_n^{[kl]}|$, implying 
\begin{align}
    \label{eq:conv_total_var_thm2}
    \sup_{1\leq i \leq c} |r_{l,i}^{[k]}(\mathcal{B}_n) - p_{l,n}^{[k]}(u_n^{[l]}, \mathcal{B}_n)| \leq \frac{c+1}{2|E_n^{[kl]}|-c} \to 0
\end{align}
as $n\to\infty$ since $|E_n^{[kl]}|\sim \mu_l^{[k]}\gamma^{[l]}n$ and $c$ is fixed. Equation \eqref{eq:conv_total_var_thm2} implies convergence in total variation, a result mirroring Theorem \ref{thm:asym_dist}.
\end{proof}

\paragraph{Convergence of the ECoHeN Algorithm}

We address the convergence properties of the ECoHeN algorithm.

\begin{proof}[Proof of Theorem \ref{thm:extraction_conv}]
Since $\xi\in[0, 1]$ and $\phi\in[0, 1)$, there exists an iteration $i^{'}$ such that $\mu_{i'}=1$ and $\mu_i=1$ for all $i\geq i^{'}$. Furthermore, there exists a $j^{'}$ such that $|\mathcal{B}_i| < |\mathcal{V}|/2$ for all $i\geq j^{'}$. Let $j=\max(i^{'}, j^{'})$. For all $i\geq j$, the ECoHeN extraction procedure iteratively adds (at most) a one external node, $v\in\mathcal{B}_j^c$, with the smallest FDR adjusted $p$-value less than (or equal to) $\alpha$ before removing (at most) one internal node, $v\in\mathcal{B}_j^+$, with the largest FDR adjusted $p$-value greater than $\alpha$. To illustrate convergence, it suffices to show that a node, $v$, added to the set $\mathcal{B}_j$ at iteration $j$ is not subsequently removed from the set $\mathcal{B}_j^{+}=\mathcal{B}_j\cup\{v\}$ at iteration $j$. We proceed by contradiction. 

Importantly, note that the quantity $\hat{p}_{\mathcal{B}_{j}}(u)$ for an arbitrary node $u$ of arbitrary type $l$ is the same regardless of whether $u\in\mathcal{B}_j$ or $u\not\in\mathcal{B}_j$, implying $\hat{p}_{\mathcal{B}_{j}}(v)=\hat{p}_{\mathcal{B}^{+}_{j}}(v)$. Let $q_{\text{ext}}\vcentcolon=q_{\text{ext}}(j)=|\mathcal{B}_j^{c}|$ denote the number of external nodes, and $q_{\text{int}}\vcentcolon=q_{\text{int}}(j)=|\mathcal{B}_j^{+}|$ denote the number of internal nodes at iteration $j$. Since $q_{\text{int}} < q_{\text{ext}}$ and $\hat{p}_{\mathcal{B}_{j}}(v)=\hat{p}_{\mathcal{B}^{+}_{j}}(v)$, we have $\tilde{\hat{p}}_{\mathcal{B}_{j}^+}(v) < \tilde{\hat{p}}_{\mathcal{B}_{j}}(v)$ where $\tilde{\hat{p}}$ the FDR adjusted $p$-value. Intuitively, since the $p$-value for $v$ is the same regardless whether $v$ is a member of $\mathcal{B}_j$, then the FDR adjusted $p$-value for $v$ will be smaller when $v$ is a member of $\mathcal{B}_j$ (i.e. $\mathcal{B}_j^{+}$) than when $v$ is not a member of $\mathcal{B}_j$ (i.e. $\mathcal{B}_j$) since there are more external nodes than internal nodes at iteration $j$. At the same time, since $v$ was added to $\mathcal{B}_j$ then $p_{\mathcal{B}_{j}}(v) \leq \alpha$. Since $v$ was subsequently removed from $\mathcal{B}_j^+$, then $p_{\mathcal{B}^{+}_{j}}(v) > \alpha$. Hence we have $\alpha < \tilde{p}_{\mathcal{B}_{j}^+}(v) < \tilde{p}_{\mathcal{B}_{j}}(v) \leq \alpha$, a clear contradiction: $\alpha < \alpha$. As such, $v$ will never be added and subsequently removed in the extraction procedure, implying the ECoHeN algorithm will not cycle.
\end{proof}

\section{Simulation Study}

We provide a detailed account of the heterogeneous stochastic block model: a model for generating heterogeneous networks with block structure. Furthermore, we provide a more expansive set of simulated conditions as well as discussion of the evaluation metrics used in the manuscript.

\subsection{Heterogeneous Stochastic Block Model}

The \textit{heterogeneous stochastic block model} (HSBM) is a flexible framework for generating heterogeneous networks with block structure and is implemented in the ECoHeN package provided in Supplement D. To describe the size and connectivity of sampled networks, consider a heterogeneous network with $K$ node types with $C$ blocks. The block sizes are described by the $K$ by $C$ matrix $N=[n_{kc}]_{\overset{1\leq k\leq K}{1\leq c \leq C}}$ where $n_{kc}$ provides the number of type $k$ nodes assigned to the $c^{\text{th}}$ block. The connectivity of a sampled network is then summarized by the symmetric matrices $P$ and $R$, each of size $K$ by $K$, which provide the probability of connections between nodes of type $k$ and $l$ for all $k,l\in\{1, \dots, K\}$. 

To detail, suppose $v^{[k]}$ and $u^{[l]}$ represent two arbitrary nodes of type $k$ and $l$, respectively. In particular, the matrix $P=[p_{kl}]_{1\leq k\leq l \leq K}$ provides the probability of connection between two nodes should these nodes exist in separate blocks. That is, if $v^{[k]}$ and $u^{[l]}$ do not share a block, then an edge is placed between them according to a Bernoulli random variable with rate $p_{kl}$. The matrix $R=[r_{kl}]_{1\leq k\leq l\leq K}$ provides the additive increase in the rate of connection between two nodes if they share a block; that is, if $v^{[k]}$ and $u^{[l]}$ share a block, then an edge is placed between them according to a Bernoulli random variable with rate $p_{kl} + r_{kl}$ where $0\leq p_{kl}, r_{kl}\leq 1$ and $0 \leq p_{kl} + r_{kl} \leq 1$ for all $1\leq k\leq l\leq K$.

Networks presented in the manuscript are generated by parameters $b$, $p$, $r_{11}$, $r_{22}$, and $r_{12}$ according to the following matrices of the HSBM:
\begin{align}
    N(p)&= \begin{bmatrix} 500(1-p) & 500p \\
    500(1-p) & 500p \end{bmatrix}, \label{eqn:Nmatrix} \\
    P(b)&=\begin{bmatrix} b & b \\b & b \end{bmatrix}, \label{eqn:Pmatrix} \text{ and} \\
    R(r_{11}, r_{22}, r_{12})&=\begin{bmatrix} r_{11} & r_{12} \\ r_{12} & r_{22} \end{bmatrix}. \label{eqn:Rmatrix}
\end{align}
Networks of this form maintain two blocks: a background block and a high connectivity block (HCB) where the parameter $p$ in \eqref{eqn:Nmatrix} dictates the size of the HCB, ergo, the size of the community. The parameter $b$ in \eqref{eqn:Pmatrix} dictates the background rate. The parameters $r_{ij}$ in \eqref{eqn:Rmatrix} dictate the type and degree of community structure. Figure \ref{fig:hsbm_examples} provides three examples of a heterogeneous network from this space, including one a) with heterogeneous community structure, b) with homogeneous community structure, and c) without community structure, each visualized as an adjacency matrix with node type illustrated by the colored bars adorning the axes. In particular, Figure \ref{fig:hsbm_examples}\textcolor{blue}{a} is said to have heterogeneous community structure as the within-block nodes are highly connected within node type and (to a lesser degree) between node type, implying the existence of one heterogeneous community. On the other hand, Figure \ref{fig:hsbm_examples}\textcolor{blue}{b} is said to have homogeneous community structure as the within-block nodes are highly connected within node type yet sparsely connected between node type (equivalent to the background rate, $b$), implying the existence of two homogeneous communities: one composed of type one (red) nodes and one composed of type two (blue) nodes. Figure \ref{fig:hsbm_examples}\textcolor{blue}{c} is a heterogeneous analog of an Erdős-Rényi network; hence, there is no underlying community structure.
\begin{figure}[htbp!]
\centering
   \begin{subfigure}{0.32\textwidth} \centering
     \includegraphics[scale=0.55]{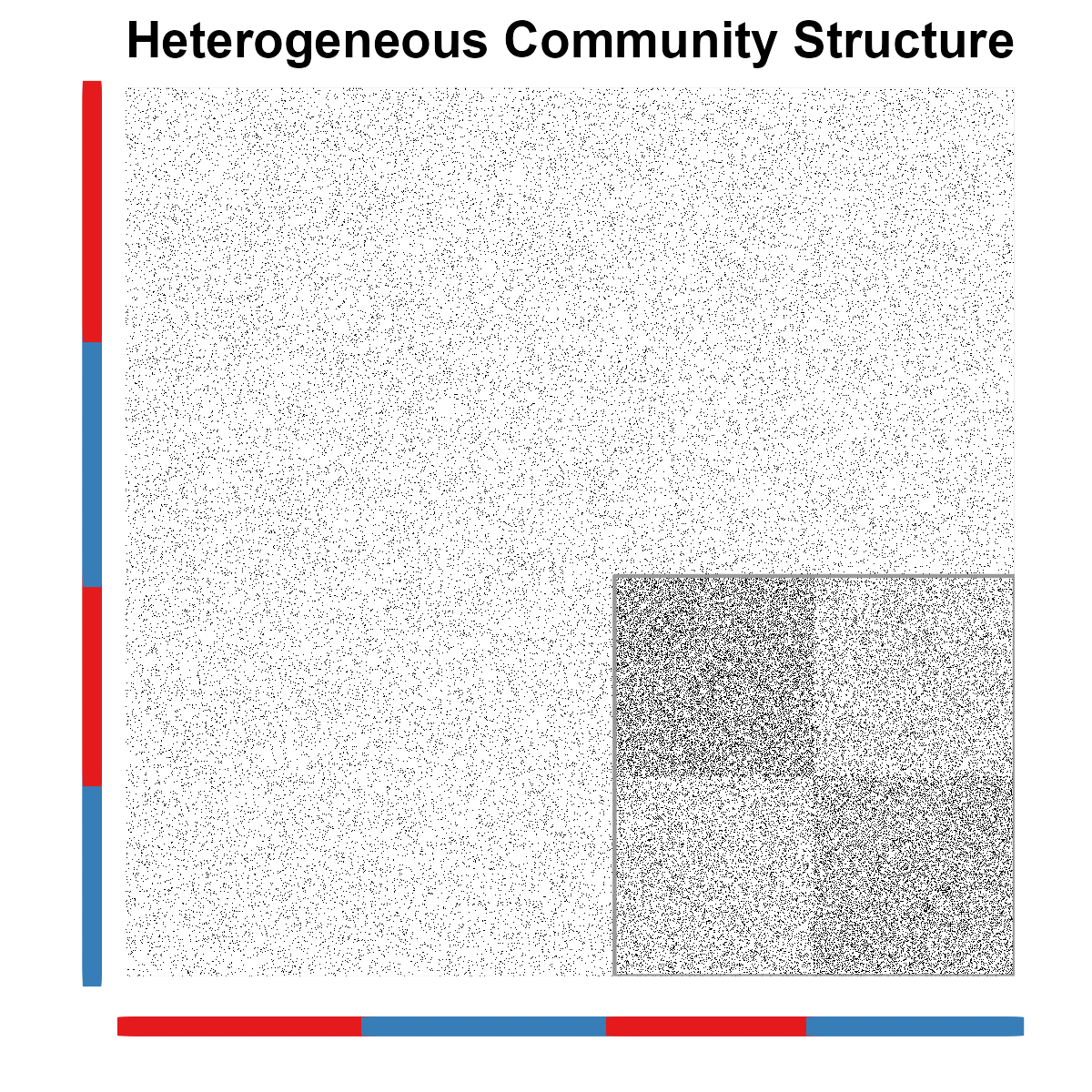}
     \caption{$R(0.25, 0.20, 0.10)$}
   \end{subfigure}
   \begin{subfigure}{0.32\textwidth} \centering
     \includegraphics[scale=0.55]{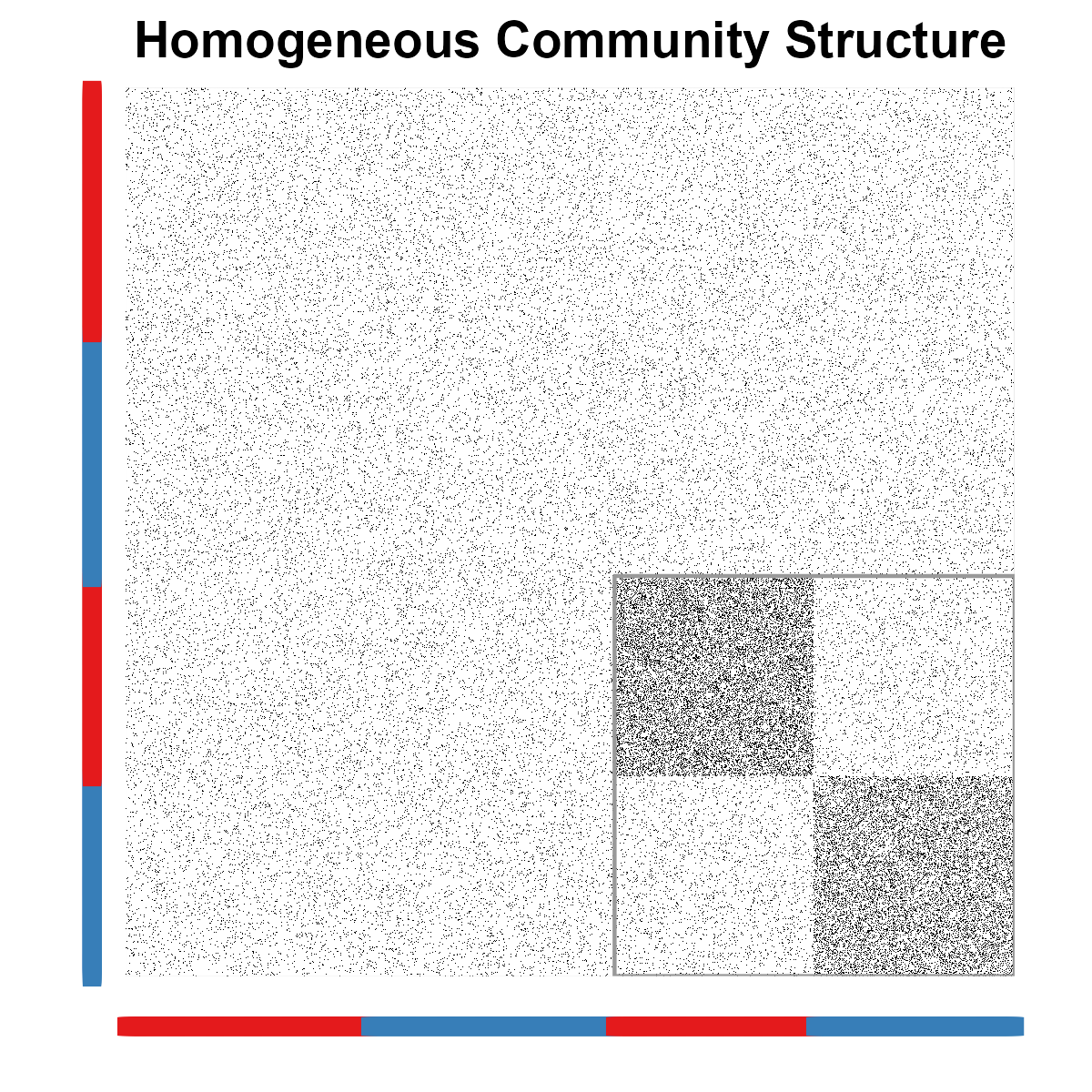}
     \caption{$R(0.25, 0.20, 0)$}
   \end{subfigure}
   \begin{subfigure}{0.32\textwidth} \centering
     \includegraphics[scale=0.55]{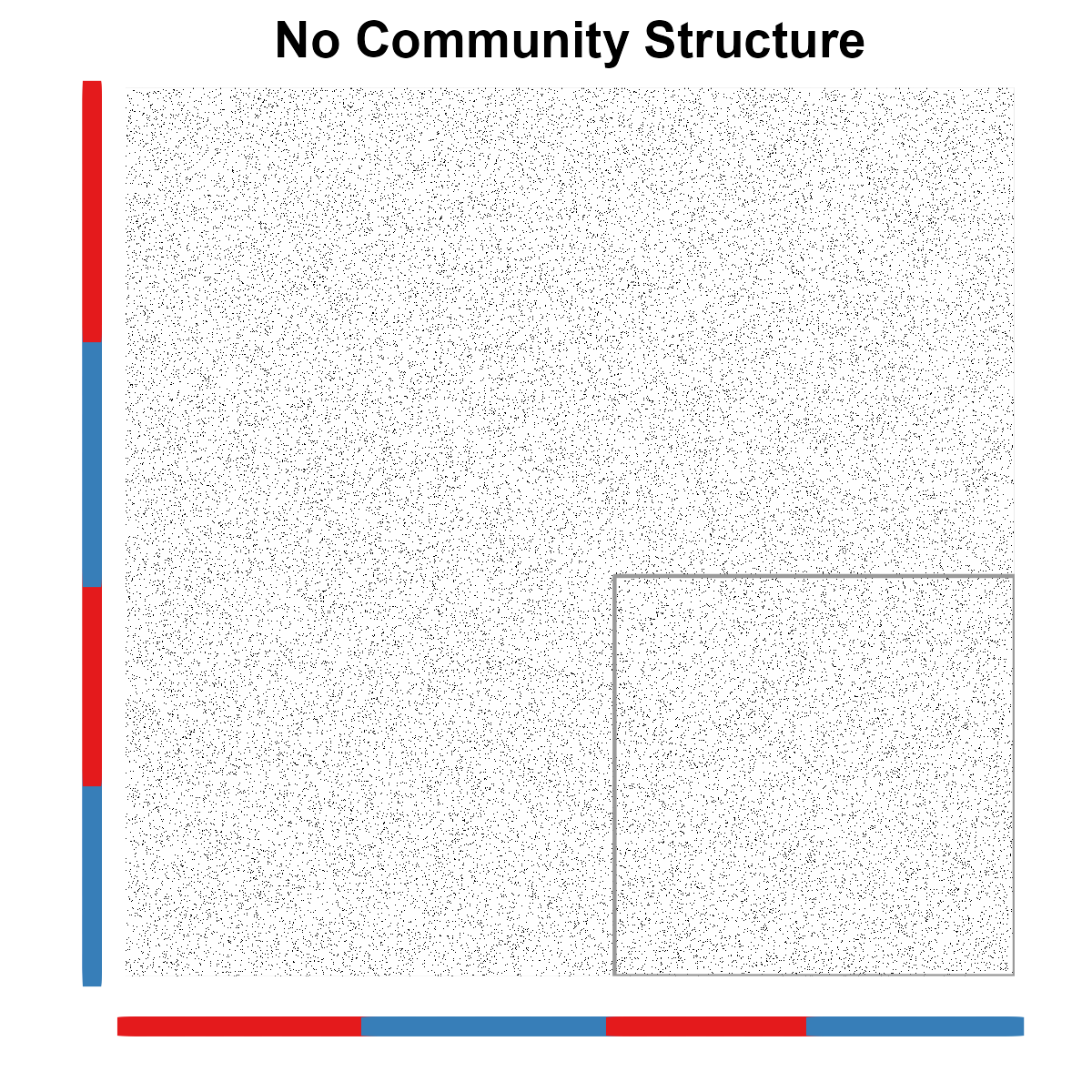}
     \caption{$R(0, 0, 0)$}
   \end{subfigure}
\caption{Three example heterogeneous networks generated from the heterogeneous stochastic block model with fundamentally different community structure. Each network is composed of 500 type one and 500 type two nodes (illustrated by the colored bars adorning the axes) and the subsequent connections among them. The HCB (outlined in black) contains the connections dictated by $r_{ij}$ for $1 \leq i\leq j\leq 2$. The other connections are dictated by the background rate $b$. Edges are sampled according to the Bernoulli rate matrices $P(b)$ and $R(r_{11}, r_{22}, r_{12})$ of \eqref{eqn:Pmatrix} and \eqref{eqn:Rmatrix} with background rate specified by $b=0.05$.}
\label{fig:hsbm_examples}
\end{figure}

\subsection{Evaluation Metrics}

The maximum Jaccard similarity measure is used in simulation to capture each community discovery method's ability to identify the simulated community structure. For each method, let $D$ represent the set of nodes to discover which will vary depending on the underlying community structure. Let $\boldsymbol{C}_m$ represent a collection of the discovered communities by method $m$. To gauge each method's ability to identify $D$, we consider the \textit{maximum Jaccard similarity measure}, denoted $J^{*}(D, \boldsymbol{C}_m)$, between the set of nodes to discover, $D$, and the collection of communities, $\boldsymbol{C}_m$: 
\begin{equation}
\label{eq:max_jaccard}
    J^{*}(D, \boldsymbol{C}_m) = \max_{C\in\boldsymbol{C}_m} \frac{|D\cap C|}{|D\cup C|} \vcentcolon= \max_{C\in\boldsymbol{C}_m} J(D, C).
\end{equation}
A value of one indicates that the set $D$ was perfectly identified by the community discovery method; a value of zero indicates the set $D$ was not identified in any capacity by the community discovery method. Hence, the maximum Jaccard similarity measure illustrates each method's ability to recover set $D$ where larger values indicate more overlap.

In real networks, there is no established ground truth. To assess the assortativity of a discovered community, $C$, we compute the ratio of densities. The \text{ratio of densities} of $C$, denoted $\text{RatD}(C)$, is the ratio of the internal edge density, $p_i(C)$, to the between edge density, $p_b(C)$: $\text{RatD}(C)=p_i(C)/p_b(C)$ where 
\begin{align*}
    p_i(C) = \frac{m_i}{|C|(|C|-1)/2}, \text{ and } p_b(C) = \frac{m_b}{|C||C^{c}|}.
\end{align*}
We use $m_i$ to denote the number of edges between nodes in $C$, and $m_b$ to denote the number of edges between a node in $C$ and a node in $C^{c}=\mathcal{V}-C$. A RatD of one implies that the density of edges within the set is the same as to the density to the rest of the network, indicating poor assortativity. Any set with a RatD sufficiently greater than one can arguably be called a community. The formulation assumes a simple network which holds for every network presented in the manuscript.

\subsection{Simulation Study}

The simulation study of the manuscript features a subset of the results depicting each method's ability to identify heterogeneous and homogeneous community structure. This section presents all other results as well as the effects of $\xi$ and $\phi$ on the quantity and quality of the communities found. We also provide an investigation into each method's ability to assign nodes to background in random networks.

\subsubsection{Heterogeneous Community Structure}

In the manuscript, we present each method's ability to identify the heterogeneous community when $r_{ii}\in(0.15, 0.20, 0.25, 0.30)$ for $i\in\{1, 2\}$ and $r_{12}\in(0.025, 0.05, 0.075)$; however, a broader range of simulated conditions are constructed and explored. We present results when $r_{11}=r_{22}$ and $r_{11}\neq r_{22}$ before exploring the effects of $\xi$ and $\phi$. As in the manuscript, one hundred networks are generated at each simulated condition, referred to as replicates. The median maximum Jaccard similarity measure is plotted at each simulated condition with uncertainty reflected by the range of first and third quartile. 

\paragraph{When $r_{11}=r_{22}$}

Figure \ref{fig:heterogeneous_r11eqr22} compares each method's ability to identify the heterogeneous community when $r_{11}=r_{22}$ at a broader range of simulated conditions. Each point represents the median maximum Jaccard at each simulated condition. The vertical range represents the middle 50\% of observed maximum Jaccard measures. ECoHeN and ZCmod's ability to recover the heterogeneous community notably improves as the within-block, within-type density (i.e., $b + r_{ii}$) increases \emph{and} the within-block between-type density (i.e., $b + r_{ij}$) increases. Each method poorly recovers the heterogeneous community when $r_{ii}<0.15$. When $r_{ii}<0.15$, Walktrap generally seems preferable, having a larger maximum Jaccard. However, ECoHeN and ZCmod outperform Walktrap if $r_{12}$ is relatively large, and if $r_{12}$ is sufficiently large, ECoHeN outperforms each of the competing methods. When $r_{ii}\geq 0.15$, ECoHeN performs relatively better than each competing methods at recovering small heterogeneous communities (i.e., when $p$ is small), especially for relatively small $r_{12}$. For larger $p$, ECoHeN and ZCmod perform similarly well.

\begin{figure}[ht]
  \begin{adjustbox}{addcode={\begin{minipage}{\width}}{\caption{
      Each method's ability to identify the heterogeneous community along with the size of the heterogeneous community. The vertical interval around the median represents the middle 50\%. ECoHeN and ZCmod's ability to recover the heterogeneous community notably improves as the within-block, within-type density (i.e., $b + r_{ii}$) increases \emph{and} the within-block between-type density (i.e., $b + r_{ij}$) increases. Each method poorly recovers the heterogeneous community $r_{ii}<0.15$. When $r_{ii}<0.15$, Walktrap generally seems preferable, while ECoHeN and ZCmod perform similarly in these conditions. When $r_{ii}\geq 0.15$, ECoHeN performs relatively better than all methods at recovering small heterogeneous communities (i.e., when $p$ is small), especially for relatively small $r_{12}$.
      }
      \label{fig:heterogeneous_r11eqr22}
      \end{minipage}},rotate=90,center}
      \includegraphics[scale=1]{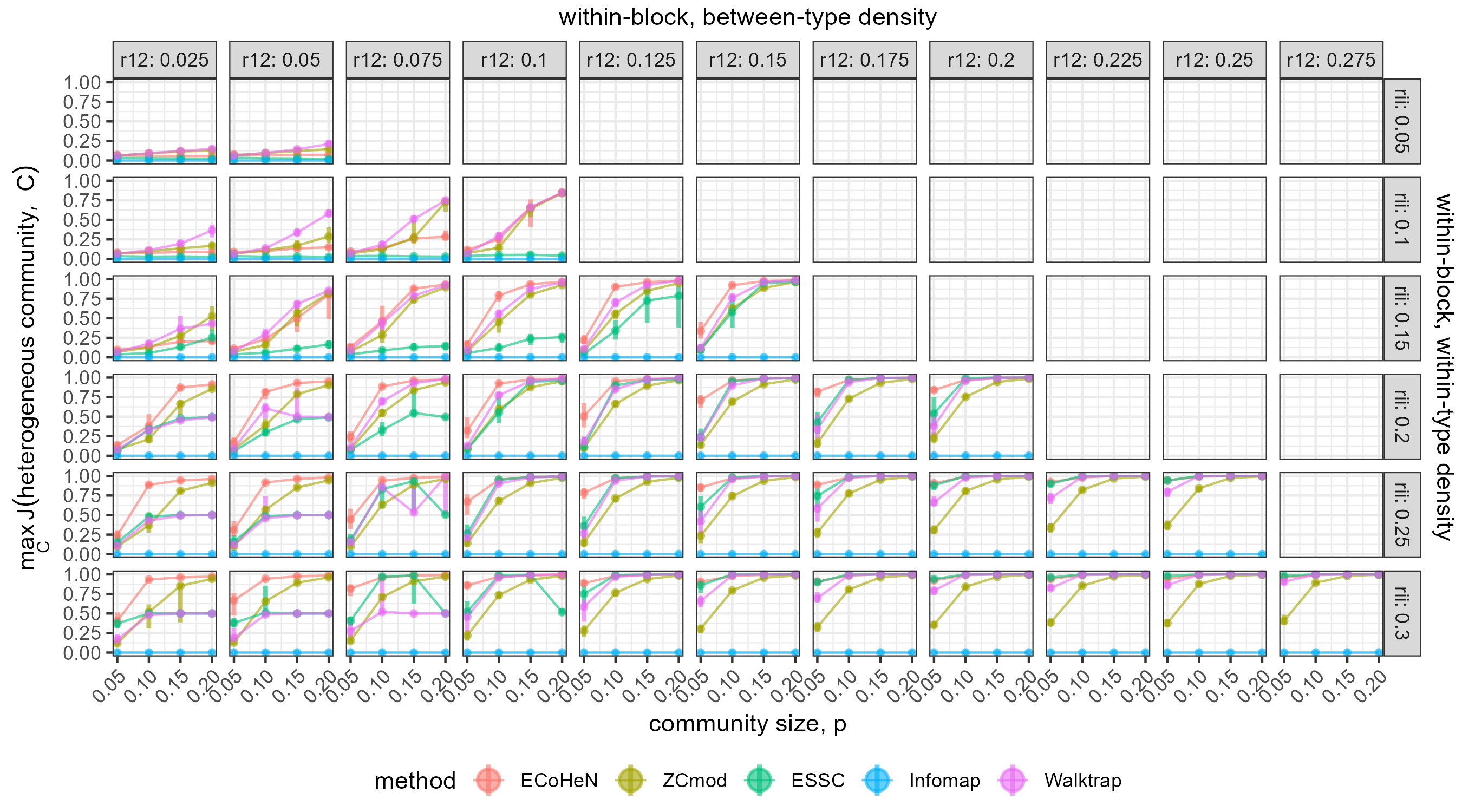}
  \end{adjustbox}
\end{figure}

\paragraph{When $r_{11}\neq r_{22}$}

We now consider simulated conditions when $r_{11}\neq r_{22}$, presented in Figure \ref{fig:heterogeneous_best_r11neqr22}. The within-block, within-type densities are provided in the facets. For each community size, we show the maximum Jaccard for increasing values of $r_{12}$ which range from 0.025 to $r_{22}$. Each method's ability to recover the heterogeneous community improves as 1) the density of red block nodes' connections increases (down the y-axis facets), 2) the density of the blue block nodes' connections increases (right across the x-axis facets), and 3) the density of the within-block, between-type connections increases (provided by the line type). Characteristically, however, ECoHeN's ability to recover the heterogeneous community drastically improves for small increases to $r_{12}$, outperforming each competing method at recovering small, heterogeneous communities. By comparison, ZCmod struggles to identify small, heterogeneous communities, performing well as the community size increases. When $r_{ii}>0.20$ for $i\in\{1, 2\}$ (not shown), all methods perform similarly for large $r_{12}$. For small $r_{12}$, ESSC and Walktrap perform relatively worse than ECoHeN and ZCmod. ECoHeN outperforms ZCmod for small $r_{12}$, and the two methods perform more similarly as $r_{12}$ increases.

\begin{figure}[ht]
  \begin{adjustbox}{addcode={\begin{minipage}{\width}}{\caption{
      The within-block, within-type (between-type) densities are provided in the facets (line type). For each community size, we show the maximum Jaccard for increasing values of $r_{12}$ which range from 0.025 to $r_{22}$. ECoHeN's ability to identify the heterogeneous community drastically improves for small increases to $r_{12}$, outperforming each competing method at recovering small, heterogeneous communities. By comparison, ZCmod struggles to identify small, heterogeneous communities, performing well as the community size increases.
      }
      \label{fig:heterogeneous_best_r11neqr22}
      \end{minipage}},rotate=90,center}
      \includegraphics[scale=0.8]{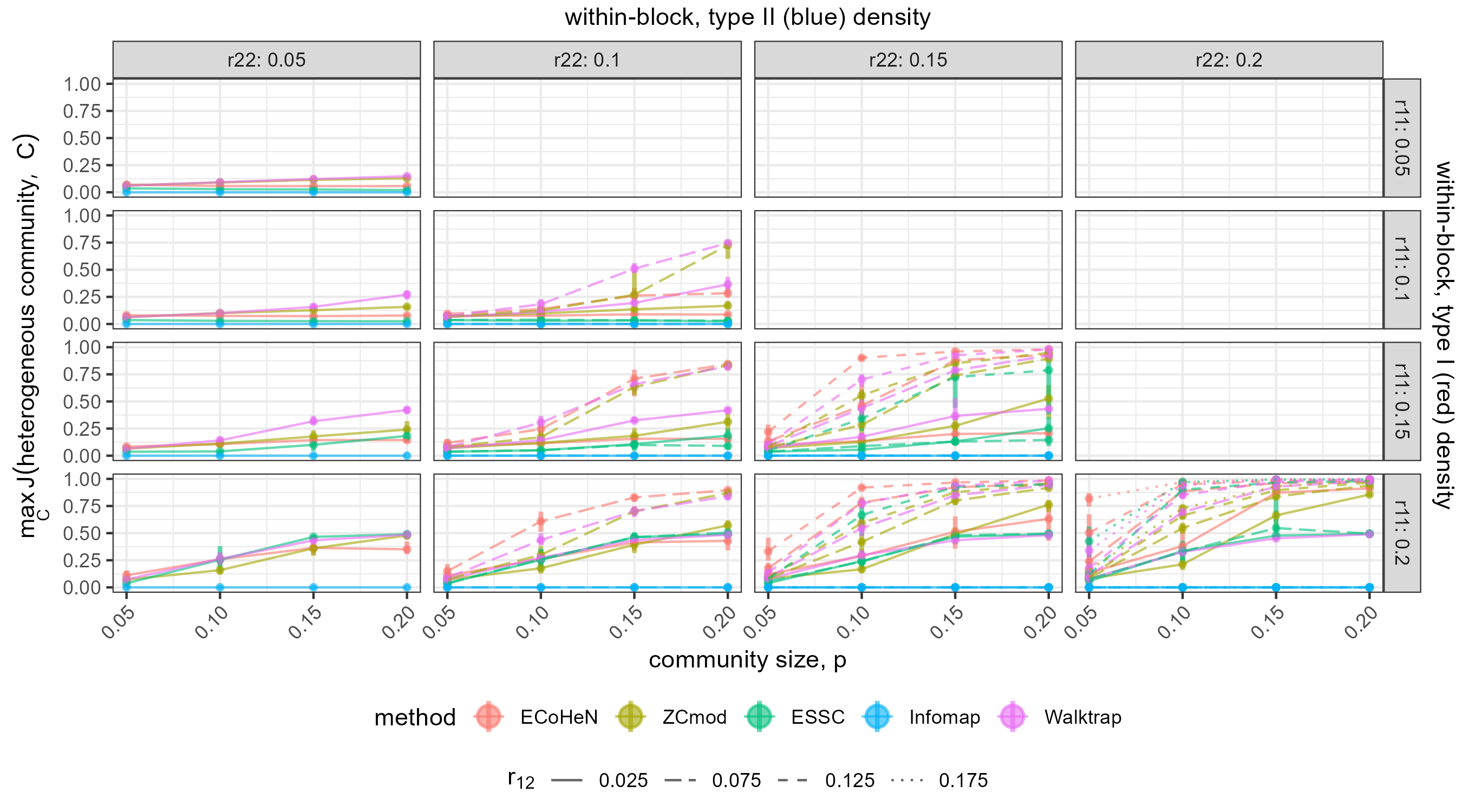}
  \end{adjustbox}
\end{figure}

\paragraph{Effects of $\xi$ and $\phi$}

We consider setting $(\xi, \phi)$ to $(0, 0)$, and $(1, \phi)$ for $\phi = 0, 0.33, 0.66$, and $0.99$ and running ECoHeN and ESSC to see the effect the parameter settings have on the ability for the methods to recover the heterogeneous community. The results are provided in Figure \ref{fig:heterogeneous_best_r11eqr22_effects}. In general, it appears that setting the maximal allowance to 1 for each iteration of the extraction with $\xi=0$ and $\phi=0$ provides the best resolution for uncovering the heterogeneous community at a wide range of simulated conditions. When $\xi=1$, a larger $\phi$ provides the best resolution, suggesting that if we are to speed up the algorithm by increasing the maximal allowance for early iterations, it is best to allow for a larger maximal allowance for many early iterations. However, the effect is relatively minute for most simulated conditions, suggesting that the choice of $\phi$ when $\xi=1$ will have minimal impact on the methods' ability to recover communities from background. We do not consider setting $(\xi, \phi)=(1, 1)$ as this setting is not guaranteed to converge. Should the user wish to set $(\xi,\phi)=(1, 1)$, we suggest considering a maximum number of iterations fewer than the number of nodes in the network.

\begin{figure}[ht]
  \begin{adjustbox}{addcode={\begin{minipage}{\width}}{\caption{
      Setting a maximal allowance to one for each iteration with $\xi=0$ and $\phi=0$ provides the best resolution for uncovering the heterogeneous community at a wide range of simulated conditions. When $\xi=1$, a larger $\phi$ provides the best resolution, suggesting that if we are to speed up the algorithm by allowing a larger maximal allowance for early iterations, it is best to allow a larger maximal allowance for many early iterations. However, the effect is relatively minute for most simulated conditions, suggesting that the choice of $\phi$ when $\xi=1$ will have minimal impact on the methods' ability to recover communities from background.
      }
      \label{fig:heterogeneous_best_r11eqr22_effects}
      \end{minipage}},rotate=90,center}
      \includegraphics[scale=1]{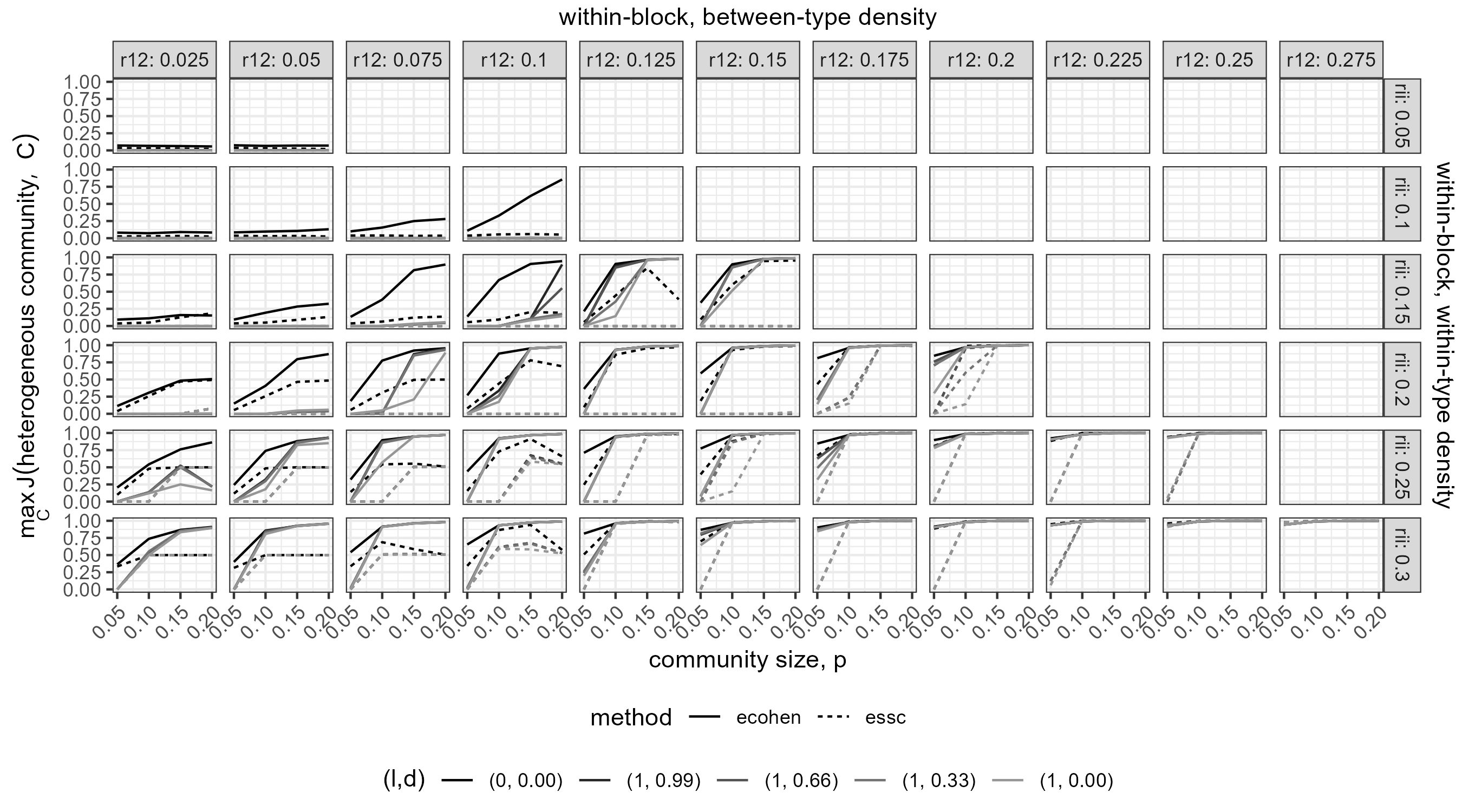}
  \end{adjustbox}
\end{figure}

\subsubsection{Homogeneous Community Structure}

In the manuscript, we present each method's ability to identify the red community when $r_{11}\in(0.20, 0.25, 0.30)$, $r_{22}=0.25$, and $r_{12}=0$; however, a broader range of simulated conditions are constructed and explored. We present results for identifying both the red and blue community for general $r_{11}\in(0.15, 0.20, 0.25, 0.30)$ where $i\in\{1, 2\}$ before exploring the effects of $\xi$ and $\phi$. As in the manuscript, one hundred networks are generated at each simulated condition, referred to as replicates. The median maximum Jaccard similarity measure is plotted at each simulated condition with uncertainty reflected by the range of first and third quartile.

Figure \ref{fig:homogeneous_maxjaccard_red} compares each method's ability to identify the red community at a broader range of simulated conditions. Each point represents the median maximum Jaccard at each simulated condition. The vertical range represents the middle 50\% of observed maximum Jaccard measures. As the red to red density increases within the HCB (along the y-axis facets), ECoHeN can identify ECoHeN can identify the homogeneous community composed of red nodes with increasingly better precision, featuring marked improvements for small, homogeneous communities. There are no such improvements from ZCmod which partitions a network into modules each of which must maintain at least one node of each node type. Notably, the blue to blue density does not have an impact on any method's ability to identify the homogeneous community composed of red nodes.

While ECoHeN consistently outperforms ZCmod, ESSC and Walktrap consistently outperform ECoHeN and ZCmod at uncovering homogeneous community structure. This is not surprising considering these methods identify communities irrespective of node type. By construction, the within-block, between-type density is no different from the background density for networks under study. While ECoHeN is designed to identify both homogeneous and heterogeneous community structure, the tradeoff for this functionality is a reduction in power for uncovering homogeneous communities since there is no information gained through a comparison of the within-block, between-type density to the background. Notably, ESSC performs better or similarly to Walktrap at uncovering dense, homogeneous communities, and Infomap continues to place each node into it's own community, incapable of identifying a community amongst background noise.

Lastly, Figure \ref{fig:homogeneous_maxjaccard_blue} compares each method's ability to identify the blue community at the same range of simulated conditions. The resulting conclusions are the same.

\begin{figure}[p!]
    \centering
    \includegraphics[width=\textwidth]{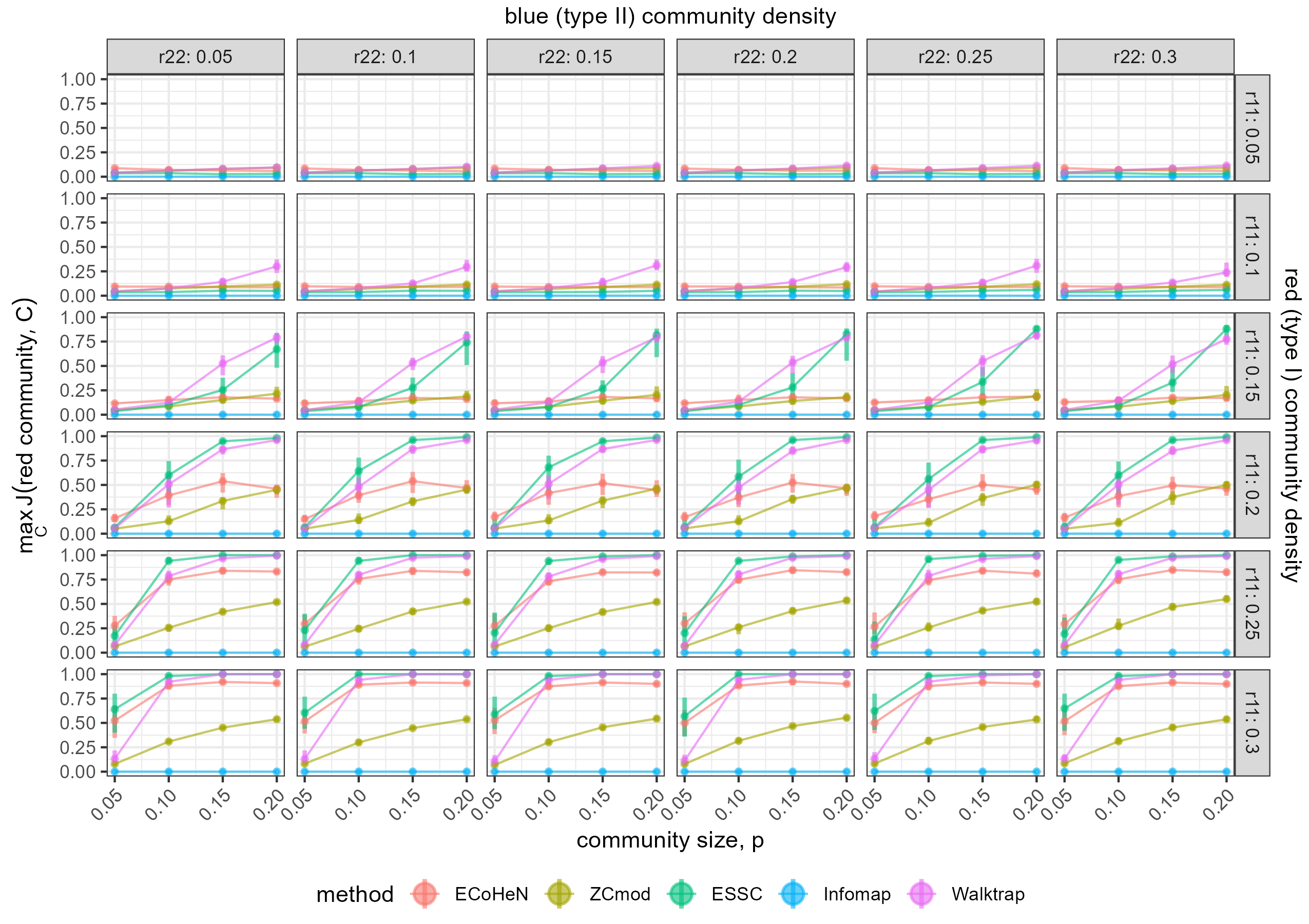}
    \caption{The blue to blue (red to red) community density is provided on the x-axis (y-axis) facets. Each method's ability to identify the red community is provided alongside the red community size. As the red to red density increases, ECoHeN can identify the red community with increasingly better precision, marked improvements for small, homogeneous communities. There are no such improvements for ZCmod. ESSC and Walktrap consistently outperform ECoHeN and ZCmod since the between-type density is no larger than the background. In general, ECoHeN has less power in identifying homogeneous community structure than ESSC and Walktrap: a tradeoff for the flexibility of finding both homogeneous and heterogeneous community structure. Notably, the density of blue to blue connections does not impact any method's ability to identify the red community.}
    \label{fig:homogeneous_maxjaccard_red}
\end{figure}

\begin{figure}[p!]
    \centering
    \includegraphics[width=\textwidth]{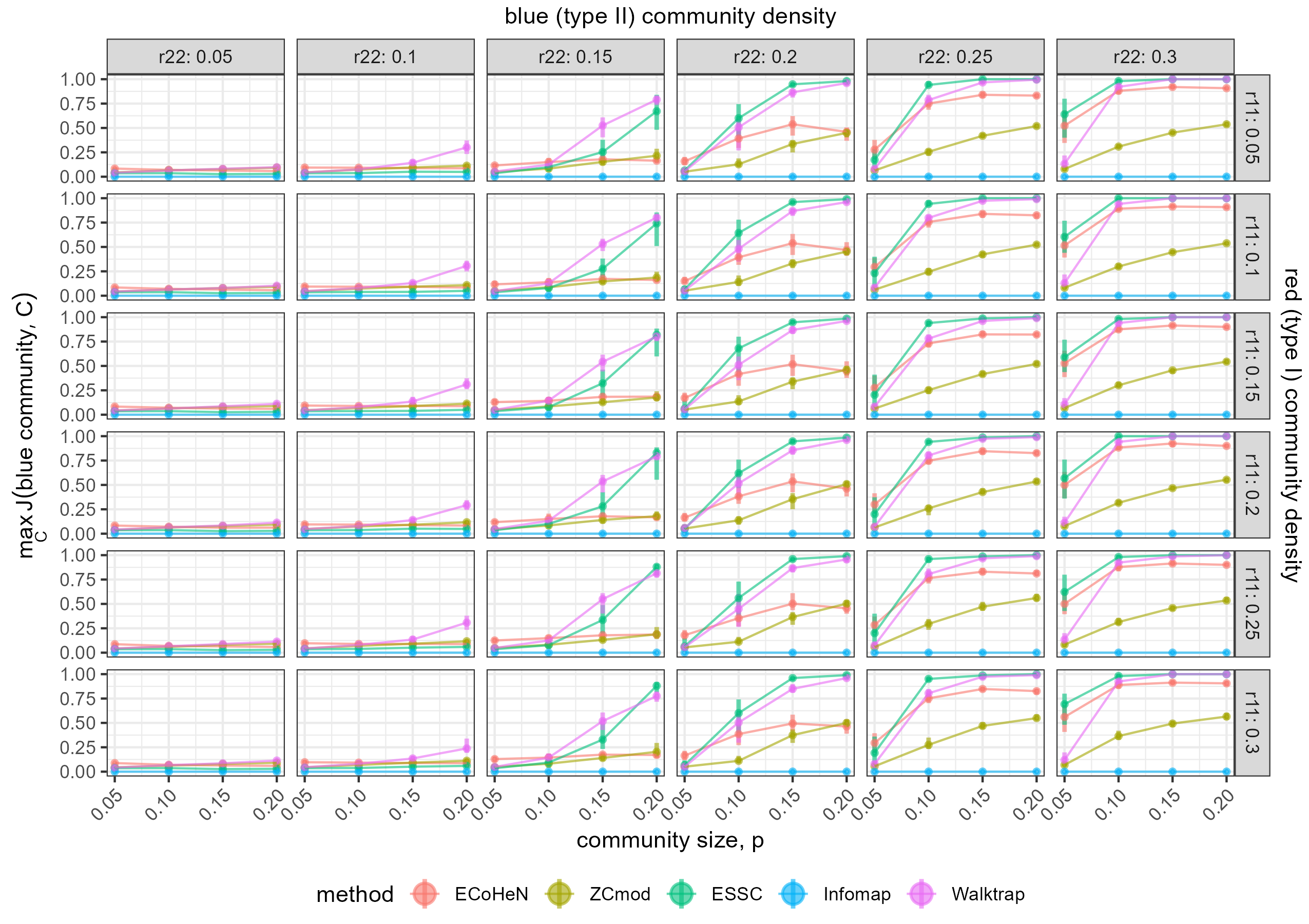}
    \caption{The blue to blue (red to red) community density is provided on the x-axis (y-axis) facets. Each method's ability to identify the blue community is provided alongside the blue community size. As the blue to blue density increases, ECoHeN can identify the blue community with increasingly better precision, marked improvements for small, homogeneous communities. There are no such improvements for ZCmod. ESSC and Walktrap consistently outperform ECoHeN and ZCmod since the between-type density is no larger than the background. In general, ECoHeN has less power in identifying homogeneous community structure than ESSC and Walktrap: a tradeoff for the flexibility of finding both homogeneous and heterogeneous community structure. Notably, the density of red to red connections does not impact any method's ability to identify the blue community.}
    \label{fig:homogeneous_maxjaccard_blue}
\end{figure}

\paragraph{Effects of $\xi$ and $\phi$}

We consider setting $(\xi, \phi)$ to $(0, 0)$, and $(1, \phi)$ for $\phi = 0, 0.33, 0.66$, and $0.99$ and running ECoHeN and ESSC to see the effect the parameter settings have on the ability for the methods to recover the homogeneous communities. The results are provided in Figure \ref{fig:homogeneous_maxjaccard_effects_red} and \ref{fig:homogeneous_maxjaccard_effects_blue}. As previously founded, it appears setting the maximal allowance to one for each iteration with $\xi=0$ and $\phi=0$ provides the best resolution for uncovering the homogeneous communities. When $\xi=1$, a larger $\phi$ provides the best resolution, suggesting that if we are to speed up the algorithm by allowing a larger maximal allowance for early iterations, it is best to allow for a larger maximal allowance for many early iterations. We again do not consider setting $(\xi, \phi)=(1, 1)$ as this setting is not guaranteed to converge to a solution. Should the user wish to set $(\xi, \phi)=(1, 1)$, we suggest considering a maximum number of iterations fewer than the number of nodes in the network.

\begin{figure}[p!]
    \centering
    \includegraphics[width=\textwidth]{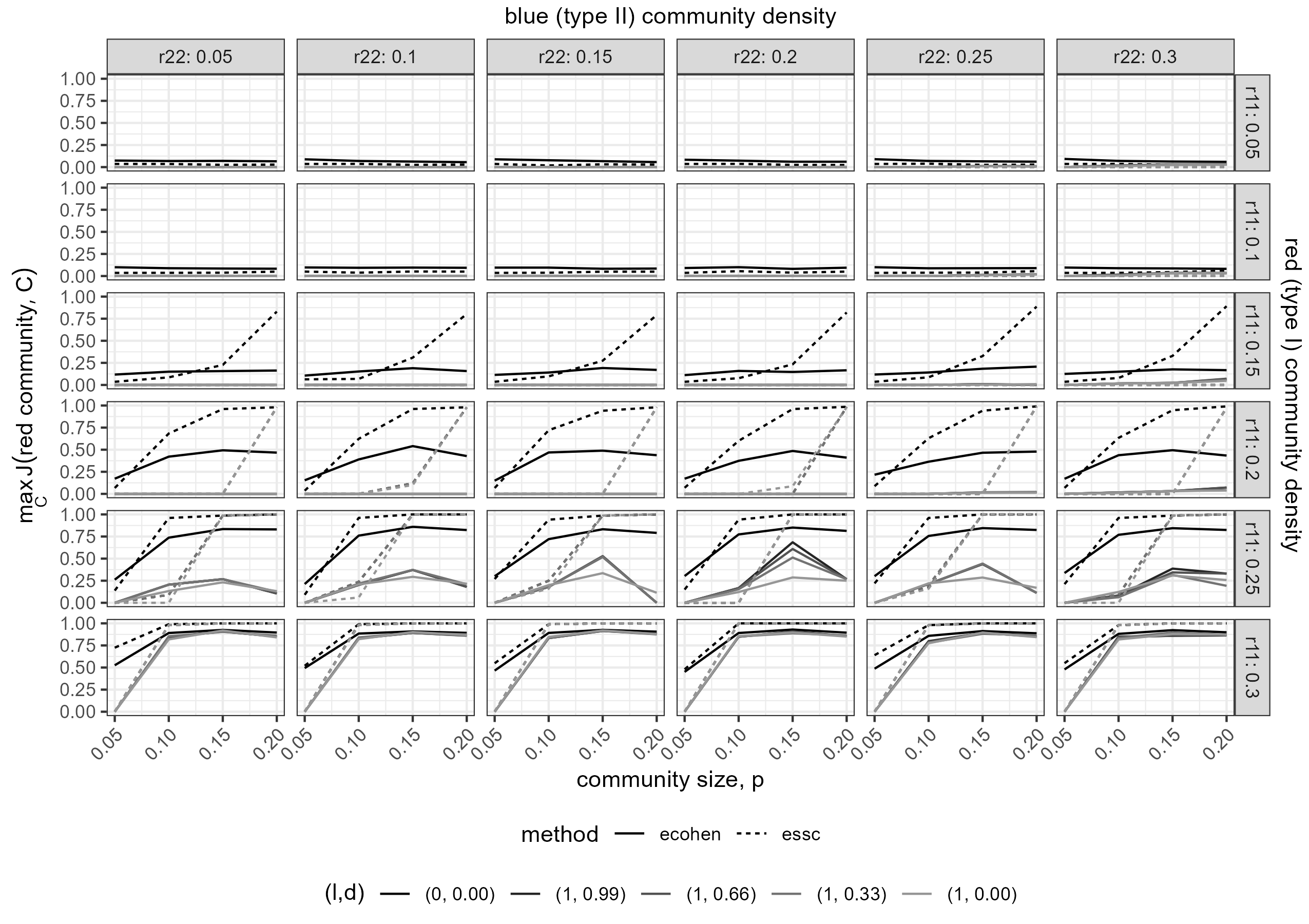}
    \caption{The parameter setting $(0, 0)$ provides the best resolution for uncovering the homogeneous communities. When $\xi=1$, a larger $\phi$ provides the best resolution, suggesting that if we are to speed up the algorithm by allowing a larger proportion of the node set to transition into and out of the candidate set, it is best to allow a large proportion to transition for more iterations.}
    \label{fig:homogeneous_maxjaccard_effects_red}
\end{figure}

\begin{figure}[p!]
    \centering
    \includegraphics[width=\textwidth]{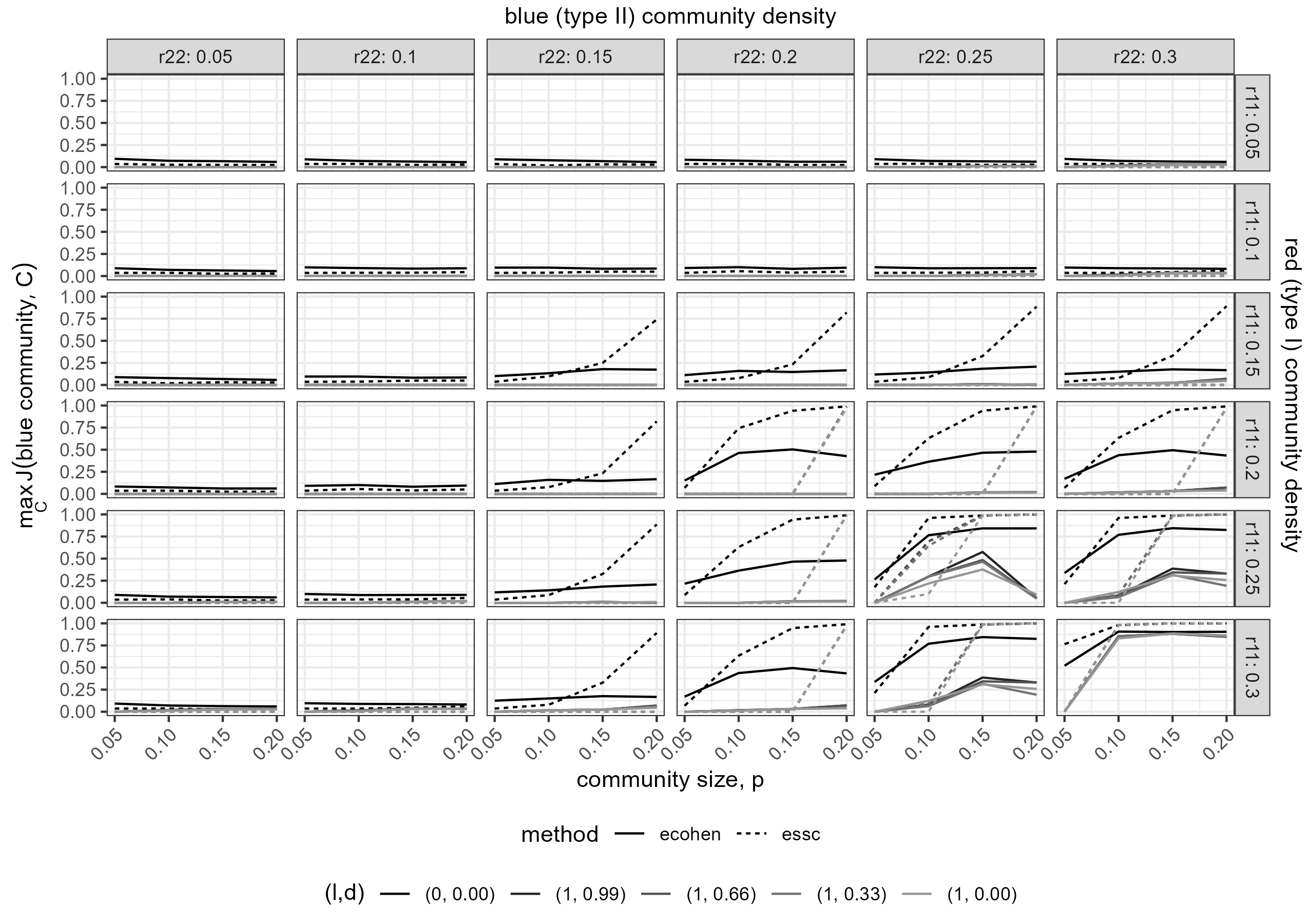}
    \caption{The fundamental conclusions are the same as founded and discussed in Figure \ref{fig:homogeneous_maxjaccard_effects_red} as they pertain to block type II (blue) nodes.}
    \label{fig:homogeneous_maxjaccard_effects_blue}
\end{figure}

\subsubsection{No Community Structure}

In this section, we investigate each method's ability to identify background nodes: nodes which are not preferentially attached to any well defined community. We generate networks with no community structure by fixing $R = R(0, 0, 0)$ and letting $P=P(b)$ where $b\in(0.05, 0.10, \dots, 0.35)$. The proportion of nodes assigned to the HCB, $p$, is meaningless provided the choice of $R$, so we fix $p=0.45$ as in Figure \ref{fig:hsbm_examples}\textcolor{blue}{c}. The random networks under this setting are Erdős-Rényi networks with rate $b$, so no community structure exists. Equivalently, the set $D=\emptyset$. For each $b$, we produce twenty replicates.

We compare each method's ability to identify background. Both ESSC and ECoHeN are capable of assigning nodes to background; however, since ZCmod, Infomap, and Walktrap are partitioning methods, they are at an innate disadvantage to identify background nodes. For a fair comparison, \textbf{we consider the largest identified community to reflect the background of each partitioning method}. If a trivial partition is assigned (i.e., all nodes are assigned to one community or each node is assigned to its own community) then the partition is assumed to identify each node as a background node. For each replicate, we compute the number of communities found by each method and the proportion of nodes identified as background nodes. The results are provided in Figure \ref{fig:best_jaccard_no_structure}.

\begin{figure}[p!]
\centering
   \begin{subfigure}{0.475\textwidth} \centering
     \includegraphics[width=\textwidth]{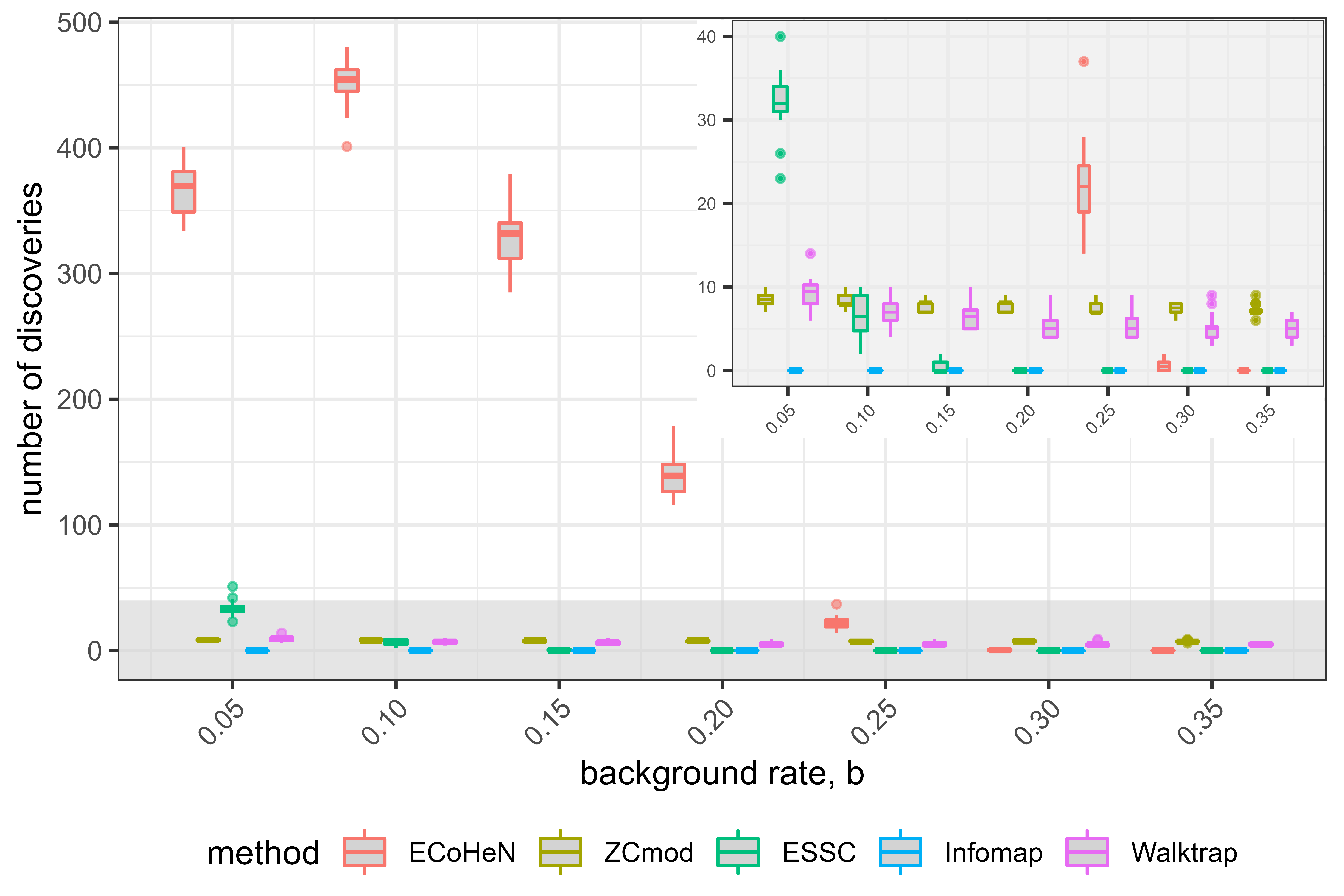}
     \caption{$(\xi,\phi)=(0, 0)$}
   \end{subfigure}
   \begin{subfigure}{0.475\textwidth} \centering
     \includegraphics[width=\textwidth]{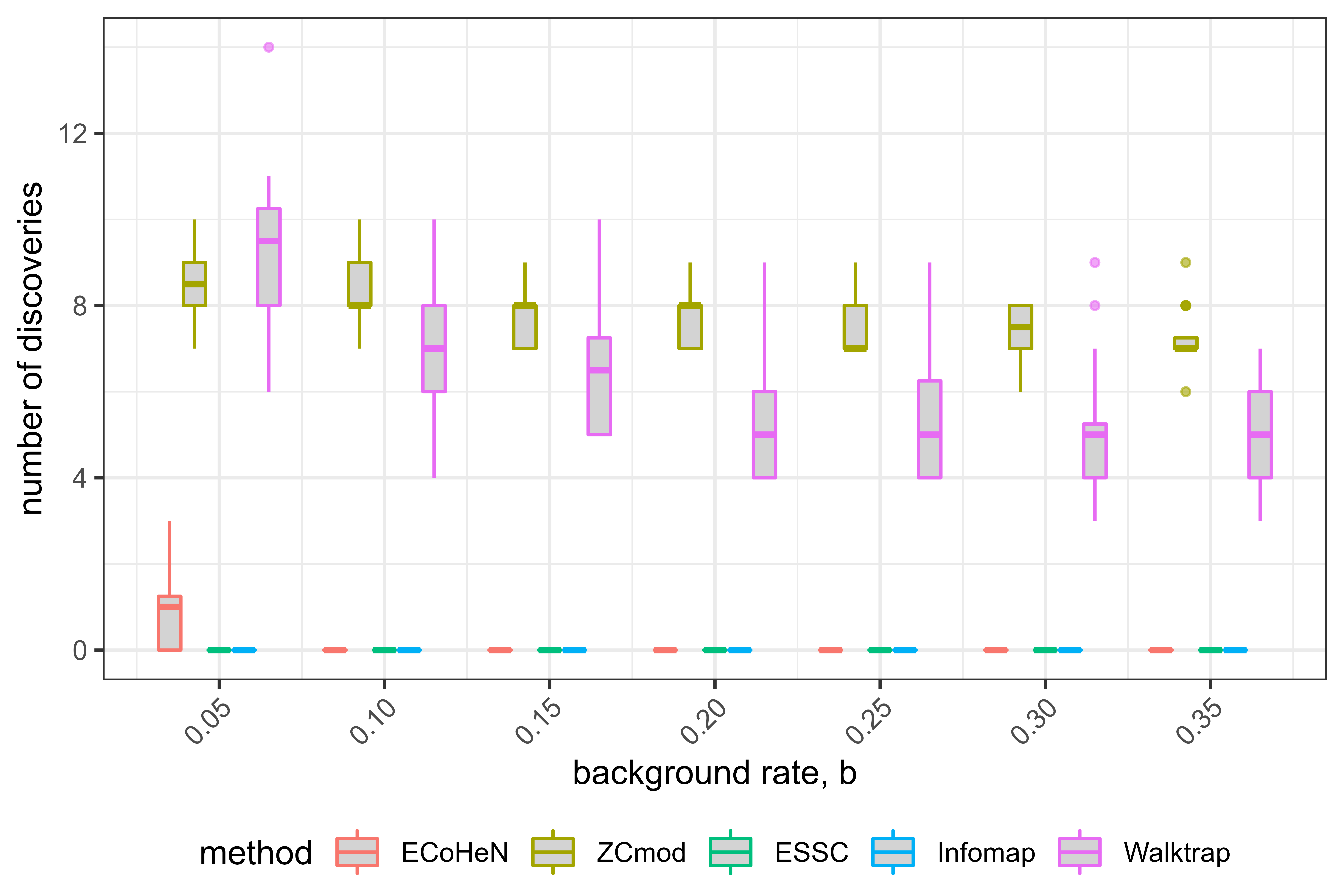}
     \caption{$(\xi,\phi)=(1, 0.99)$}
   \end{subfigure}
   \begin{subfigure}{0.475\textwidth} \centering
     \includegraphics[width=\textwidth]{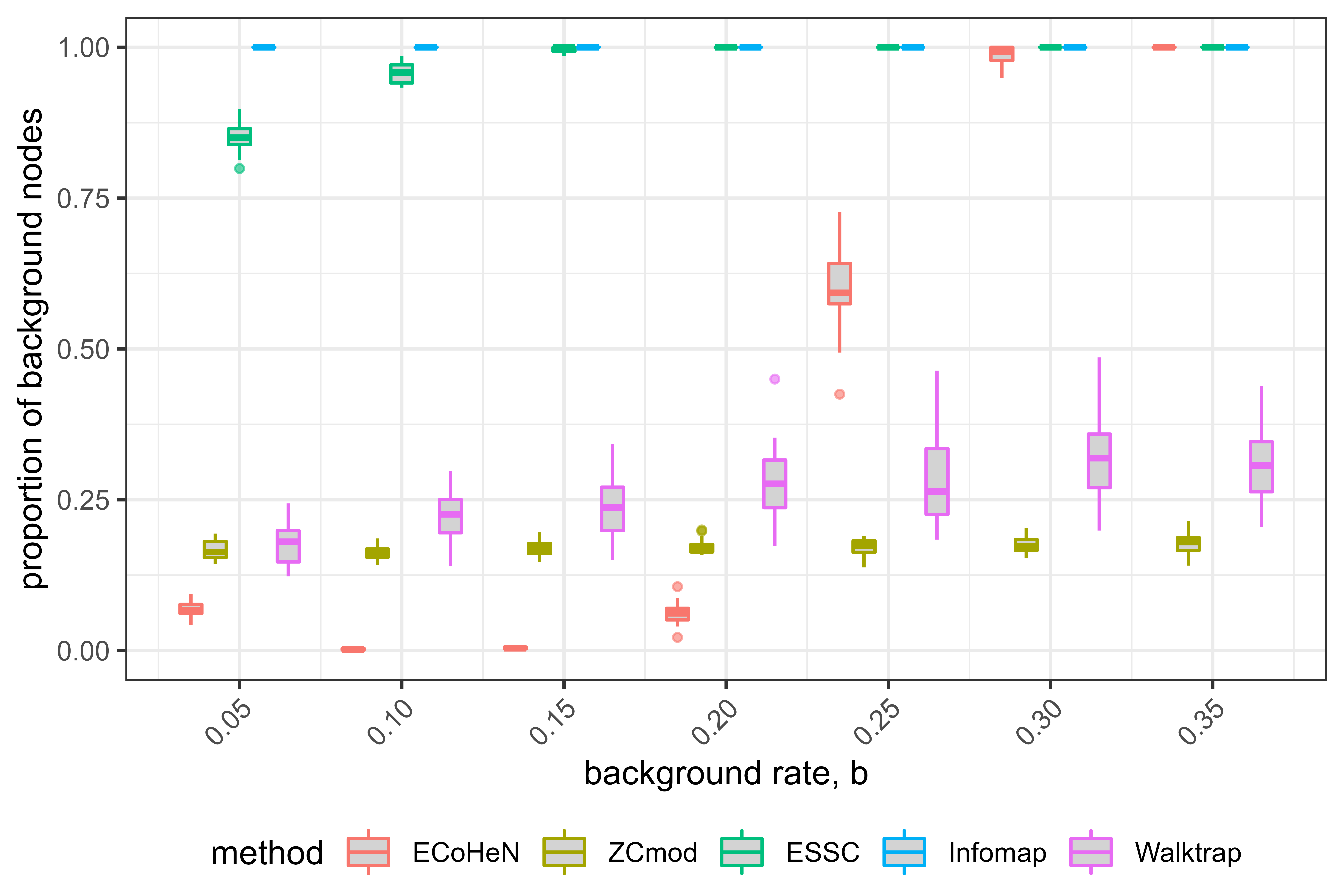}
     \caption{$(\xi,\phi)=(0, 0)$}
   \end{subfigure}
   \begin{subfigure}{0.475\textwidth} \centering
     \includegraphics[width=\textwidth]{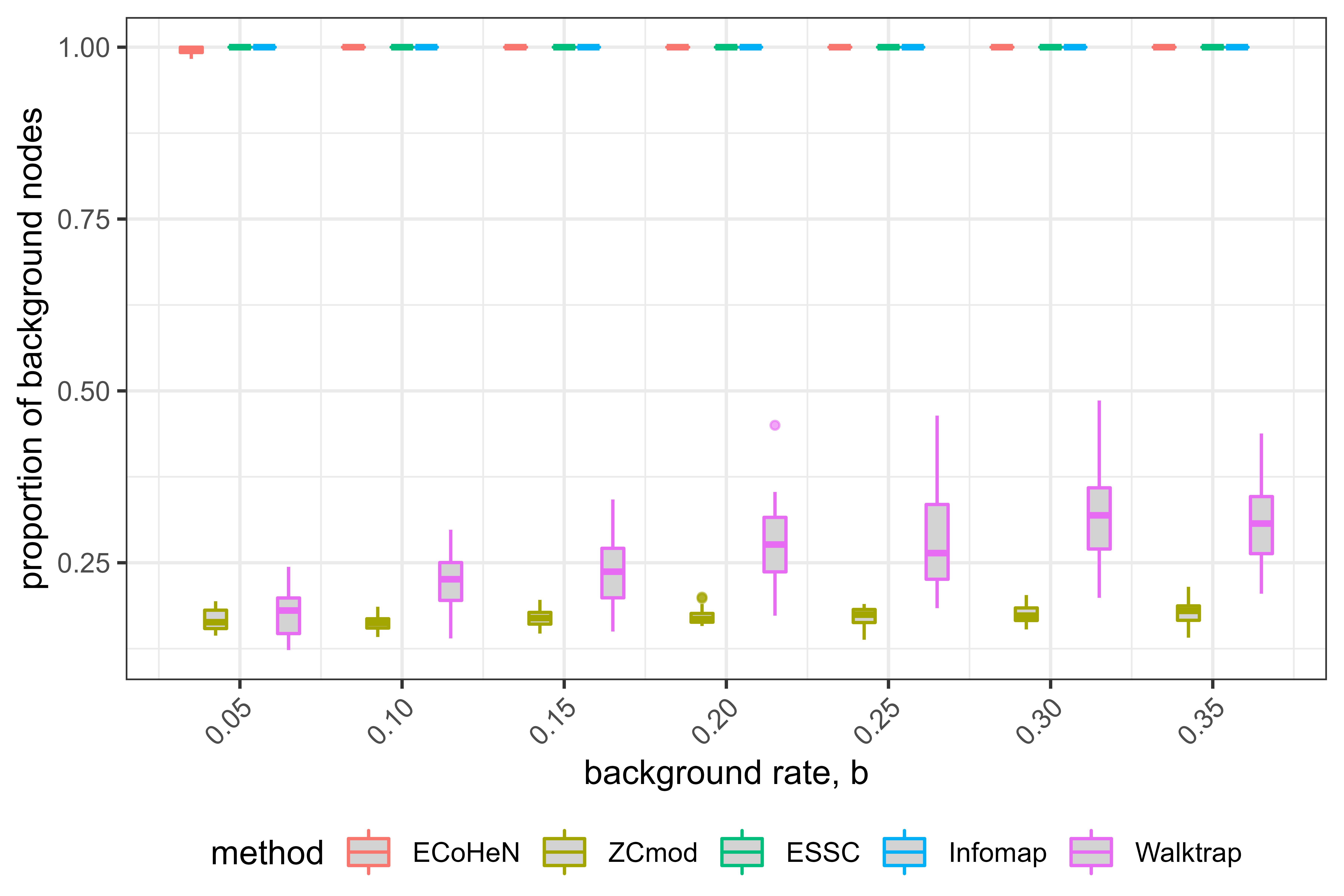}
     \caption{$(\xi,\phi)=(1, 0.99)$}
   \end{subfigure}
\caption{The learning, $\xi$, and decay, $\phi$, rate control the scale of changes when updating a candidate set from initialization to convergence. While the parameter setting $\xi=0$ and $\phi=0$ (microscopic changes) provides the best resolution for uncovering a community amongst background noise, it results in many small densely connected communities in a random network (see panels (a) and (c)). On the other hand, when $\xi=1$ and $\phi<1$ (macroscopic changes), ECoHeN identifies each node in a random network as background. Early macroscopic changes allow ECoHeN to break out of the low conductance initializations (see panels (b) and (d)). A practitioner's guide to selecting $\xi$ and $\phi$ are provided in the ``Parameter Choices'' subsection of the manuscript.} 
\label{fig:best_jaccard_no_structure}
\end{figure}

As discussed in ``Initialization'' subsection of the manuscript, ECoHeN is initialized at the neighborhood of each node. By nature, the neighborhoods in an Erdős-Rényi network have low conductance when compared to a random set of nodes. When the extraction procedure is parameterized to a maximal allowance of one for each iteration (i.e., $\xi=0$ and $\phi=0$), the neighborhoods are updated one node at a time, resulting in a small set of densely connected nodes. Thus, ECoHeN returns many small communities (see Figure \ref{fig:best_jaccard_no_structure}\textcolor{blue}{a}). As the density of the Erdős-Rényi network gets larger, the number of communities found by ECoHeN tends to zero and the proportion of nodes assigned to background approaches one (see Figure \ref{fig:best_jaccard_no_structure}\textcolor{blue}{c}). We demonstrate in the next subsection that these small communities are indeed particularly dense, at times ten to twenty times more connected internally than connected to the rest of the network.

To avoid finding communities in a random network, we can parameterize the extraction procedure such that the maximal allowance is one on the first iteration and tends to one with each passing iteration by setting the learning rate, $\xi=1$, and the decay rate, $\phi<1$. In this case, the algorithm is guaranteed to converge, and the number of communities found by ECoHeN is negligible for small $b$ and quickly tends to 0 as the $b$ gets larger (see Figure \ref{fig:best_jaccard_no_structure}\textcolor{blue}{b}). In all cases, the proportion of nodes assigned to background is near one (see Figure \ref{fig:best_jaccard_no_structure}\textcolor{blue}{d}). If one wishes to avoid finding dense subsets of nodes in an Erdős-Rényi graph, then it suffices to set $\xi=1$ regardless of choice of $\phi$. When $\xi=1$, a larger choice of $\phi$ tends to yield slightly better performance at recovering simulated community structure; however, the impact of $\phi$ is minimal when $\xi=1$. 

\paragraph{Investigation of Identified Communities}

We previously demonstrated that the extraction procedure parameterized by $\xi=0$ and $\phi=0$ provides the best resolution for extracting simulated heterogeneous and homogeneous communities from background. At the same time, we demonstrated that such a parameterization results in many small communities in a heterogeneous Erdős-Rényi network. Finding communities in an Erdős-Rényi network is not ideal behavior, so we investigate the characteristics of these identified communities. Considering the heterogeneous and homogeneous simulations were conducted at a background rate of 0.05, we will be examining the communities found by ECoHeN when $b=0.05$.

In each of twenty replicates, ECoHeN uncovers between 300 and 400 communities (see Figure \ref{fig:no_num_communities_hist}). To gauge the assortativity of each identified community, we compute the ratio of densities for each. In an Erdős-Rényi network, the expected ratio of densities is one, so any set with a ratio of densities sufficiently greater than one can arguably be called a community.

\begin{figure}[htbp!]
    \centering
    \includegraphics[scale = 0.9]{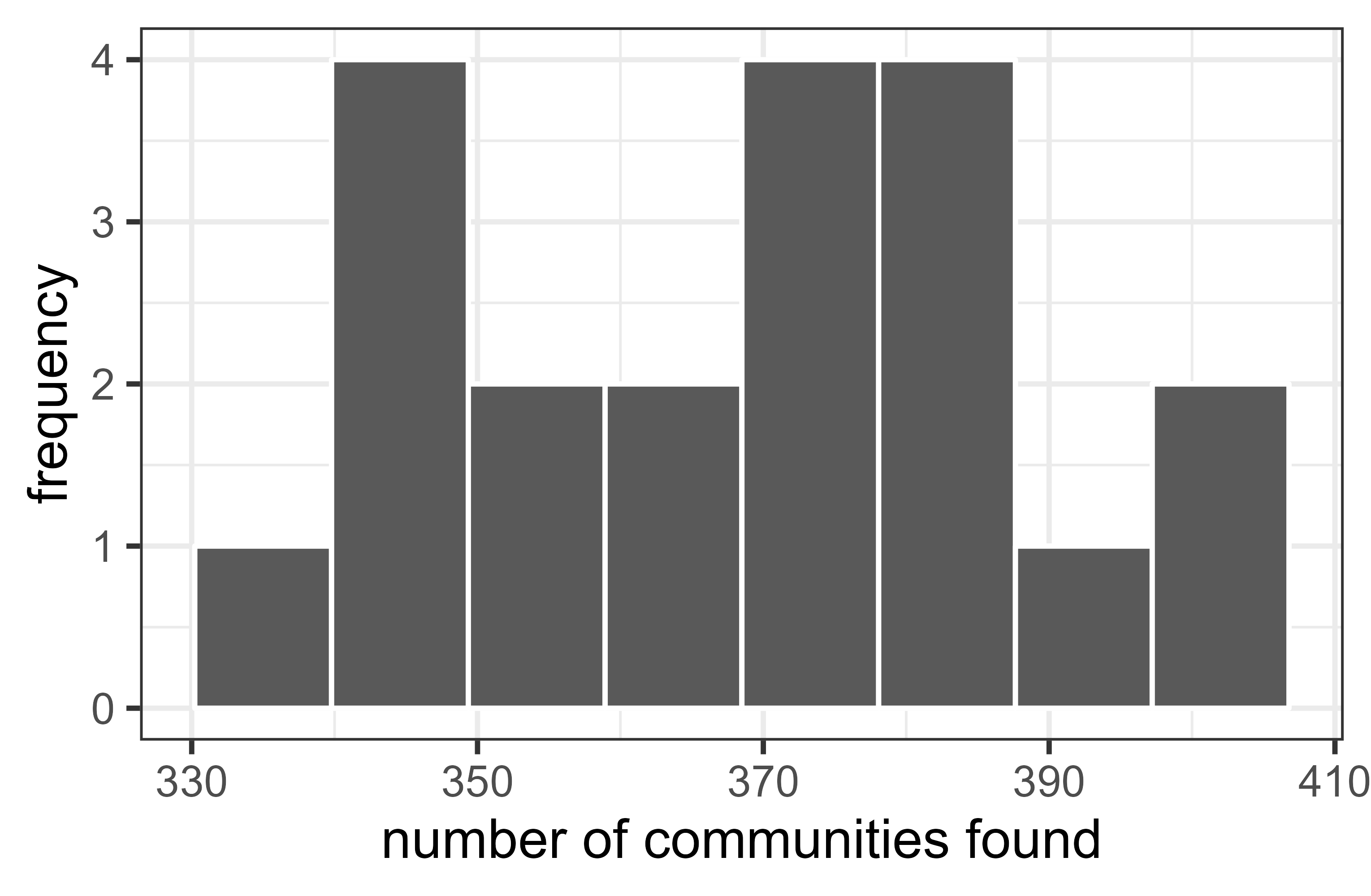}
    \caption{When $b=0.05$, ECoHeN uncovers between 300 and 400 communities for each of the 20 replicates.}
    \label{fig:no_num_communities_hist}
\end{figure}

\noindent Each of the communities found by ECoHeN is particularly dense where density scales naturally with the size of the community (see Figure \ref{fig:no_size_rat_scatter}). The majority of communities found by ECoHeN are small, and these communities feature a large ratio of densities, sometimes proving twenty times more dense internally than to the rest of the network. The larger communities found by ECoHeN tend to have relatively smaller (albeit large) ratio of densities, proving at least five times more dense internally than to the rest of the network.

\begin{figure}[htbp!]
    \centering
    \includegraphics[scale = 0.9]{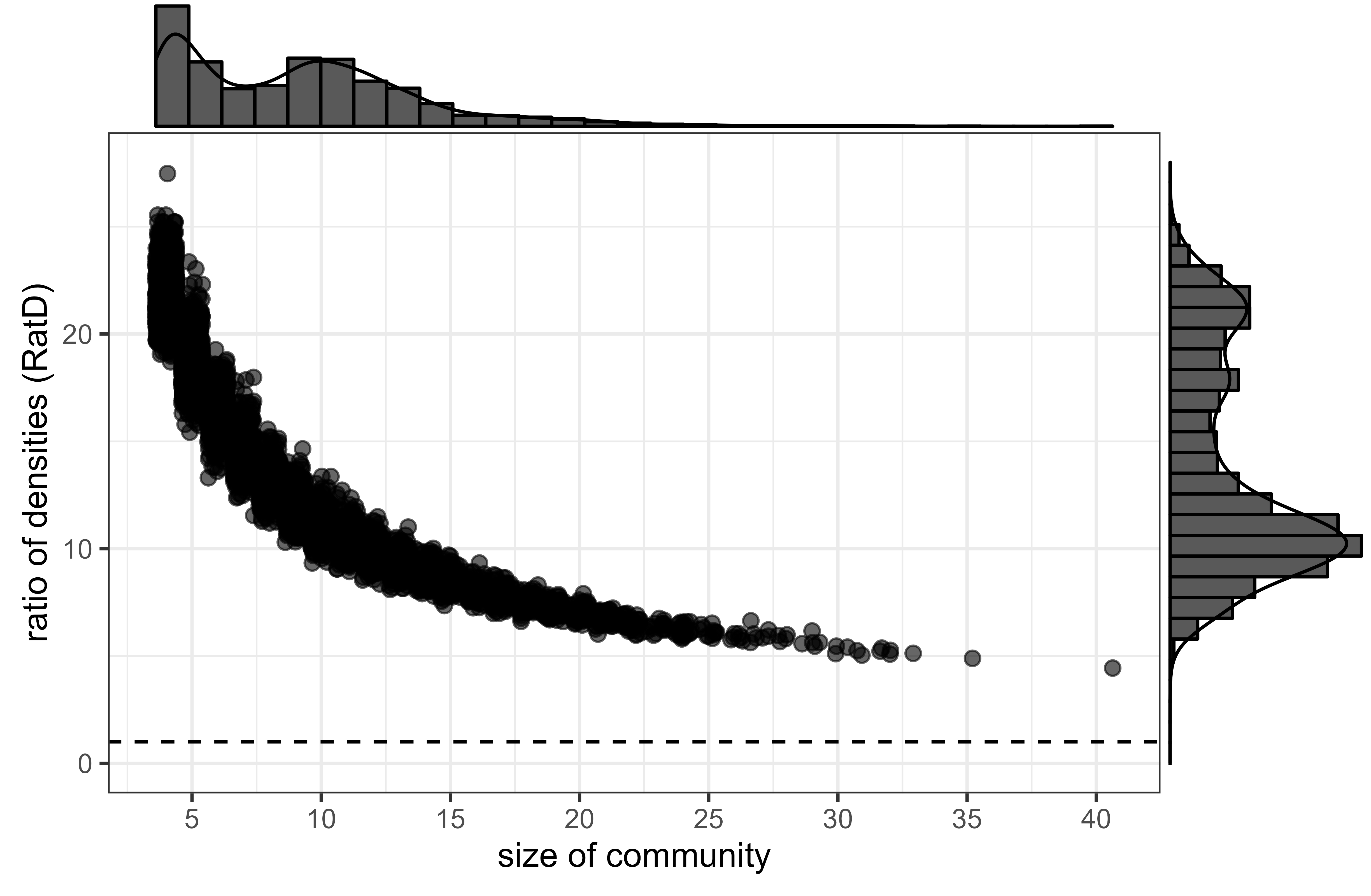}
    \caption{Each of the communities found by ECoHeN across the 20 replicates is particularly dense. Small communities necessitate a higher ratio of densities to be deemed a community, sometimes twenty times more dense internally than to the rest of the network. The larger communities found by ECoHeN tend to have relatively smaller, albeit still large, ratio of densities, proving to be at least five times more dense internally than to the rest of the network.}
    \label{fig:no_size_rat_scatter}
\end{figure}

To gauge how unlikely it would be to attain the observed ratio of densities, we gather 1000 snowball samples for each community, compute the ratio of densities for each sample, and record the 95\% quantile. Any observed ratio of densities larger than the respective 95\% quantile is deemed sufficiently dense to be considered a community. For computational feasibility, we isolate the largest community found by ECoHeN at each of the twenty replicates, plotted in Figure \ref{fig:no_largest_size_rat_scatter_jitter} in relation to the observed ratio of densities. Notice, as the community size increases, the threshold ratio of densities required decreases as it becomes increasingly unlikely to observe large communities with a large ratio of densities than a small community with an equally large ratio of densities. Each of the largest ECoHeN communities is sufficiently dense compared to the respective threshold, indicating that while these communities were attained from an Erdős-Rényi network, they are still sufficiently dense.

To describe the snowball sampling routine, consider an ECoHeN community of size $n_C$. One snowball sample is attained by first picking a node uniformly at random and recording its neighbors. There are $n_C-1$ nodes left to select. Should there be more neighbors than left to select, the remaining nodes left to select are chosen from the recorded neighbors uniformly at random. Otherwise, all of the neighbors are recorded, and the number of nodes left to select is updated accordingly. The process continues until $n_C$ nodes have been selected. Since the snowball sampling routine results in dense, well-connected sets of nodes, it is a reasonable (albeit computationally burdensome) null model when assessing community structure in an Erdős-Rényi network.

\begin{figure}[htbp!]
    \centering
    \includegraphics[scale=1]{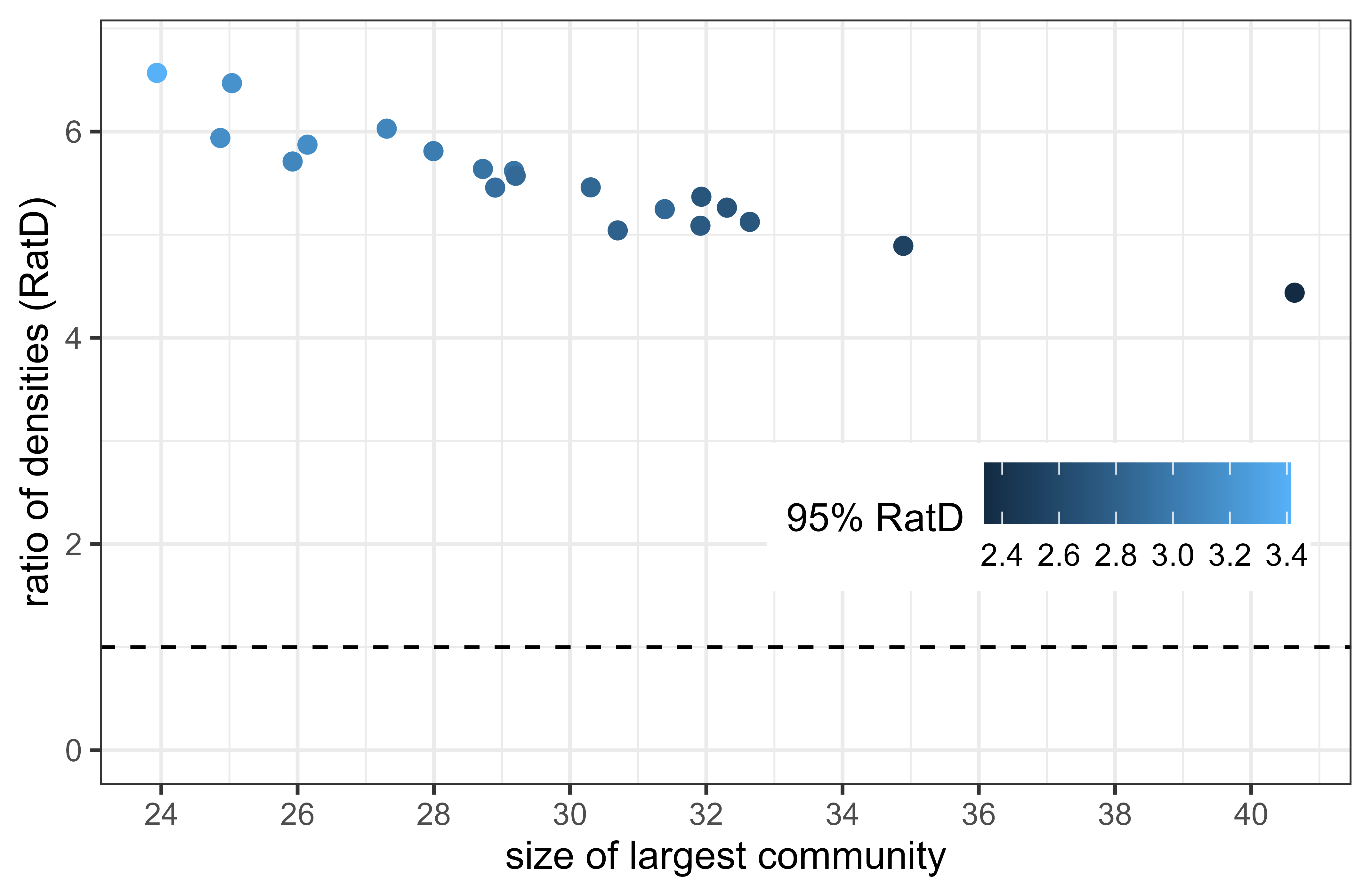}
    \caption{The largest communities across the twenty replicates are isolated and plotted with respect to the observed ratio of densities. To gauge how unlikely it would be to attain the observed ratio of densities, we gather 1000 snowball samples for each community, compute the ratio of densities for each sample, and record the 95\% quantile. Any observed value larger than the 95\% quantile is deemed sufficiently dense and arguably embodies community structure (even for a random network). Each of the largest ECoHeN communities maintains a ratio of density that is roughly twice as large as the estimated respective 95\% quantile.}
    \label{fig:no_largest_size_rat_scatter_jitter}
\end{figure}

\section{Extraction Routines}

\noindent The extraction routines AddWellConnected and RemoveLooselyConnected are provided in pseudocode in Algorithm \ref{algo:ecohen_algo_subroutine}. For the most up-to-date \texttt{R} implementation of ECoHeN (with a \texttt{C++} backend), see the following GitHub url: 
\begin{center}
    \url{https://github.com/ConGibbs10/ECoHeN}.
\end{center}

\begin{palgorithm}[htbp!]
    \centering
    \includegraphics[scale=0.98]{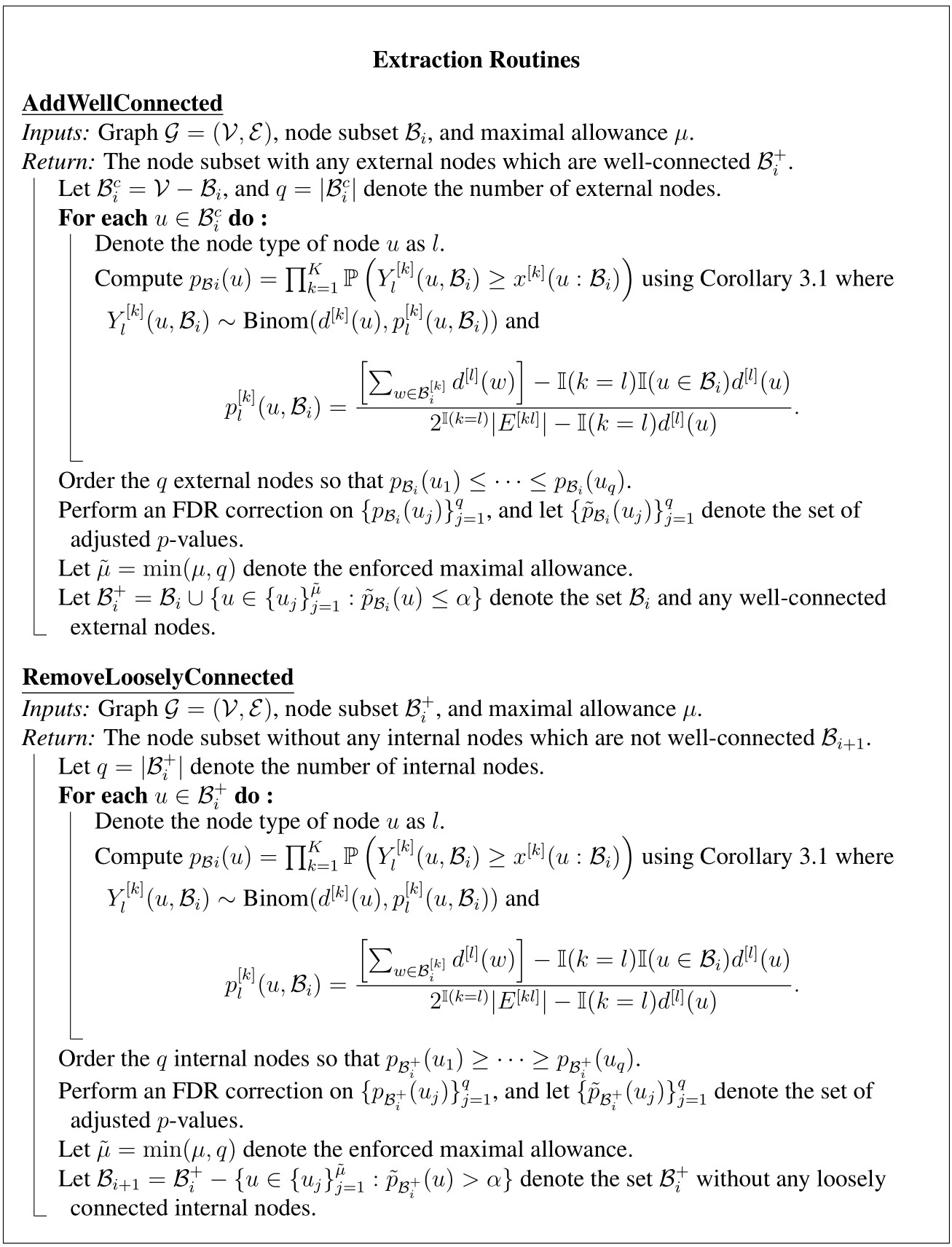}
    \captionof{palgorithm}{Pseudocode for the ECoHeN extraction routines AddWellConnected and RemoveLooselyConnected.}
    \label{algo:ecohen_algo_subroutine}
\end{palgorithm}

\end{document}